\documentclass[a4paper,UKenglish]{lipics-v2018}
\usepackage{microtype}%if unwanted, comment out or use option "draft"
\usepackage{verbatim}
\usepackage{wrapfig}

\makeatletter
\DeclareOldFontCommand{\rm}{\normalfont\rmfamily}{\mathrm}
\DeclareOldFontCommand{\sf}{\normalfont\sffamily}{\mathsf}
\DeclareOldFontCommand{\tt}{\normalfont\ttfamily}{\mathtt}
\DeclareOldFontCommand{\bf}{\normalfont\bfseries}{\mathbf}
\DeclareOldFontCommand{\it}{\normalfont\itshape}{\mathit}
\DeclareOldFontCommand{\sl}{\normalfont\slshape}{\@nomath\sl}
\DeclareOldFontCommand{\sc}{\normalfont\scshape}{\@nomath\sc}
\makeatother

\usepackage{mathpartir}
\usepackage{amsmath}
\usepackage{amsthm}

\theoremstyle{plain}

\newcounter{assumption}
\newtheorem{Assumption}[assumption]{Assumption}
\newcounter{lemma}
\newtheorem{Lemma}[lemma]{Lemma}

\nolinenumbers
\usepackage{xspace}
\usepackage{listings}
\usepackage{tikz}
\usetikzlibrary{calc}
\usetikzlibrary{shapes,snakes}
\usepackage{mathtools}
\usepackage[normalem]{ulem}

\newcommand{\SetAsGlobal}[1]{%
  \global\setbox#1\hbox{\usebox{#1}}%
  }
\newcommand{\makeCsName}[1]{\csname#1\endcsname}
\newcommand{\makeNewSaveBox}[1]{%
   \expandafter\newsavebox\csname#1\endcsname%
   }
\newcommand{\provideSaveBox}[1]{
  \ifcsname#1\endcsname%
   \relax%
  \else%
    \makeNewSaveBox{#1}%
  \fi%
}
\newcommand{\SaveWithName}[2]{%
  \provideSaveBox{privateNameForSaveWithName#1}%
  \SetAsGlobal{\makeCsName{privateNameForSaveWithName#1}}%
  \sbox{\makeCsName{privateNameForSaveWithName#1}}{#2}%
  \SetAsGlobal{\makeCsName{privateNameForSaveWithName#1}}%
}
\newcommand{\UseName}[1]{%
  \ifcsname% %this (commented) end line is fundamental
privateNameForSaveWithName#1\endcsname%
\usebox{\makeCsName{privateNameForSaveWithName#1}}%
  \else%
    [???]%
  \fi%
}

\newsavebox{\savePremise}
\newenvironment{premise}{%
  \begin{lrbox}{\savePremise}%
  $\begin{array}{c}%
  }{%
  \end{array}$%
  \end{lrbox}%
  \SetAsGlobal{\savePremise}
  }
\newsavebox{\saveConsequence}
\newenvironment{consequence}{%
  \begin{lrbox}{\saveConsequence}%
  $\begin{array}{l}%
  }{%
  \end{array}$%
  \end{lrbox}%
  \SetAsGlobal{\saveConsequence}
  }
\newsavebox{\saveSideCondition}
\newenvironment{sideCondition}{%
  \begin{lrbox}{\saveSideCondition}%
  $\begin{array}{l}%
    \mbox{with }\\[-0.5ex]%  %DECOMMENT FOR WITH
    \,\,\begin{array}{l}% %DECOMMENT FOR WITH
  }{%
  \end{array}% %DECOMMENT FOR WITH
  \end{array}$%
  \end{lrbox}%
  \SetAsGlobal{\saveSideCondition}
  }
\newcommand{\emptyPremise}{
  \begin{lrbox}{\savePremise}%
  \end{lrbox}%
  \SetAsGlobal{\savePremise}
  }
\newcommand{\emptyConsequence}{
  \begin{lrbox}{\saveConsequence}%
  \end{lrbox}%
  \SetAsGlobal{\saveConsequence}
  }
\newcommand{\emptySideCondition}{
  \begin{lrbox}{\saveSideCondition}%
  \end{lrbox}%
  \SetAsGlobal{\saveSideCondition}
  }

\newsavebox{\saveMetaRule}
\newcommand*{\metaRuleName}{1}

\newcommand*{\metaRuleNameAux}{1}

\newcommand\metaRuleScale{1}

\newcommand{\DisplayMetaRule}[4]{
  \ensuremath{
    {\tiny{\textsc{(#1)}}}
    \displaystyle
	\frac{{#2}}{{#3}}{\scalebox{\metaRuleScale}[\metaRuleScale]{#4}}
    }
  }

\newcommand\DisplayRuleName[1]{
  \!\!\!\!%
  \begin{array}[c]{l}%
    \rotatebox{90}{\scalebox{0.9}{\textsc{(#1)}}}%
  \!\!\!\!\!\!%
  \end{array}%
  }
\newcommand\DisplayPremise[1]{%
  \begin{array}{c}%
  #1%
  \\[0.5ex]%
  \end{array}%
  }
\newcommand\DisplayConsequence[1]{%
  \begin{array}{c}%
  \\[-2ex]%
  #1%
  \end{array}%
  }
\newcommand{\DisplayMetaRuleBis}[4]{
\begin{Scaled}{\metaRuleScale}{\metaRuleScale}%{0.95}{0.95}
\ensuremath{%
\!\!\!\!\!\!\!\!\!\begin{array}{l}
\DisplayRuleName{#1}
\begin{array}{l}%
\frac{%
  \!\!\!\!\DisplayPremise{#2}\!\!\!\!%
}{%
  \!\!\!\!\DisplayConsequence{#3}\!\!\!\!%
}\\[-0.5ex]%
\scalebox{0.95}{\ensuremath{#4}}%
\end{array}
\end{array}
}%ensuremath
\end{Scaled}
}

\newsavebox{\saveScaled}
\newcommand*{\varScalableEnvironmentWidth}{1}
\newcommand*{\varScalableEnvironmentHeight}{1}
\newenvironment{Scaled}[2]{%
  \renewcommand*{\varScalableEnvironmentWidth}{#1}
  \renewcommand*{\varScalableEnvironmentHeight}{#2}
  \begin{lrbox}{\saveScaled}
  }{
  \end{lrbox}
  \noindent\scalebox{\varScalableEnvironmentWidth}[\varScalableEnvironmentHeight]{\usebox{\saveScaled}}
  }

\newif\ifsubmit
\submitfalse
\ifsubmit
 
\newcommand{\EZComm}[1]{} 
 
\newcommand{\MSComm}[1]{} 
 
\newcommand{\APComm}[1]{} 
 
\newcommand{\KJXComm}[1]{} 

\newcommand{\MAZComm}[1]{}

\ifdefined\AC%
  \renewcommand{\AC}[1]{#1}%
\else%
  \newcommand{\AC}[1]{#1}%
\fi%

\newcommand{\ACComm}[1]{} 
\else

\newcommand{\NoteText}[1]{{#1}}
\newcommand{\NoteComm}[3]{\NoteText{\scriptsize \textcolor{#1}{[#2{:} #3]}}}

\newcommand{\EZComm}[1]{\NoteComm{blue}{Elenna}{#1}}
 
\newcommand{\MSComm}[1]{\NoteComm{green}{Marco}{#1}}

\newcommand{\APComm}[1]{\NoteComm{purple}{Alex}{#1}}
 
\newcommand{\KJXComm}[1]{\NoteComm{brown}{kjx}{#1}}

\newcommand{\MAZComm}[1]{\NoteComm{purple}{Max}{#1}}
\ifdefined\AC%
  \renewcommand{\AC}[1]{\textcolor{orange}{#1}}%
\else%
  \newcommand{\AC}[1]{\textcolor{orange}{#1}}%
\fi%
\newcommand{\ACComm}[1]{\NoteComm{orange}{Andrea}{#1}}
\fi

\newenvironment{grammatica}{$\begin{array}[t]{lcll}}{\end{array}$}
\newcommand\produzioneDef{%
  \Coloneqq
}
\newcommand{\produzione}[3]{%
\!\!\!#1%
&%
\!\!\!\produzioneDef\!\!\!%
&%
#2%
&%
\!\!\!\mbox{{\small{#3}}}%
}
\newcommand{\seguitoProduzione}[2]{%
&&%
#1 
&%
\!\!\!\mbox{{\small{#2}}}%
}
\newcommand{\Terminale}[1]{%
\ensuremath{%
\mbox{%
  \textbf{\texttt{#1}}%
  }%
}%
}
\newcommand{\NonTerminale}[1]{\ensuremath{\mathit{#1}}\xspace}

\definecolor{metaVarColor}{rgb}{0.65,0.1,0.1}
\def\V#1{{#1}}

\def\Inverse#1{{\colorbox{black}{\color{white}\def\Inverse##1{{##1}}\def\color##1{}#1}}}

\newcommand*{\Aux}[1]{{\scalebox{0.8}[0.8]{\sf #1}}\xspace}
\newcommand{\Many}[1]{%\V{%
  \hspace{0.03ex}
  \overline{\hspace{-0.03ex}#1\hspace{-0.03ex}}%}%
  \hspace{0.03ex}
  }

\ifdefined\theorem%
  \relax%
\else%
  \newtheorem{theorem}{{\bf Theorem}}%
\fi%

\ifdefined\proof%
  \renewenvironment{proof}{{\noindent{\it Proof}.}}{}%
\else%
  \newenvironment{proof}{{\noindent{\it Proof}.}}{}%
\fi%

\ifdefined\qedhere%
  \relax
\else{%
  \ifdefined\qed%
    \newcommand{\qedhere}{\qed}%
  \else%
    \newcommand{\qedhere}{\ensuremath{\square}}%
\fi}%

  \newcommand{\qedhere}{\ensuremath{\square}}%
\fi%

\tikzstyle{every picture}+=[remember picture]

\pgfdeclarelayer{background}
\pgfsetlayers{background,main}

\newcommand{\NamePoint}[1]{\tikz\coordinate(#1);}

\newcommand{\RectangleOpen}[1]{%
\NamePoint{#1}%
  }

\definecolor{MyGrey}{rgb}{0.5,0.5,0.5}
\definecolor{MyRed}{rgb}{1,0.1,0.1}
\newcommand{\MiniRectangleClose}[3]{%
\NamePoint{#3}%
\DrawMiniRectangle{#1}{#2}{#3}%
}

\newcommand{\DrawMiniRectangle}[3]{%
\begin{tikzpicture}[overlay]%
\draw [#1,rounded corners=1pt,opacity=0.2,fill=#1]($(#2)+(-0.1ex,1.45ex)$)%
rectangle($(#3)+(-0.05ex,-0.15ex)$);%
\end{tikzpicture}%
}

\newcommand{\mO}{\RectangleOpen{privateMO}}
\newcommand{\mC}{\MiniRectangleClose{MyRed}{privateMO}{privateMC}}

\newcommand\sizeTitleBoxTitle{}
\newenvironment{SizeTitleBox}[2]{%
\renewcommand\sizeTitleBoxTitle{#2}%
\begin{Sbox}%
\begin{minipage}{#1}%
}{%
\end{minipage}%
\end{Sbox}%
\begin{tikzpicture}%
\node [mybox] (box){\TheSbox};%
\node[myTitle, right=10pt,rounded corners] at (box.north west) {\sizeTitleBoxTitle};%
\end{tikzpicture}%
}%

\newenvironment{SizeBox}[1]{%
\begin{Sbox}%
\begin{minipage}{#1}%
}{%
\end{minipage}%
\end{Sbox}%
\begin{tikzpicture}%
\node [mybox] (box){\TheSbox};%
\end{tikzpicture}%
}

\newenvironment{ensureMathEnv}{%
\begin{Sbox}%
}{%
\end{Sbox}%
\ensuremath{\TheSbox}%
}

\definecolor{myBlueL}{rgb}{1,1,1}
\definecolor{myBlueD}{rgb}{0.2,0.2,0.6}

\definecolor{myRedL}{rgb}{1,0.85,0.85}
\definecolor{myRedD}{rgb}{0.8,0.15,0.15}
\definecolor{myBackground}{rgb}{0.8,0.8,1}

{%
\begin{minipage}{#1}%
\begin{math}%
\!\!\!%
\begin{array}{l}%
}
{%
\end{array}
\end{math}%
\end{minipage}%
}%

\definecolor{greyCodeBg}{RGB}{245,245,245}
\definecolor{greyCodeLine}{RGB}{200,200,200}

\newcommand\HandoutOnly[1]{}

\newcommand{\TerminaleRed}[1]{
\ensuremath{
{\mbox{\mO\lstinline@#1@\mC}}
}\xspace}

\newcommand{\HSep}{\hbox to \columnwidth{\bf\hrulefill}}
\newcommand{\HSepLong}{\hbox to \textwidth{\bf\hrulefill}}

\newcommand{\myCalBig}[1]{
    {%
    \scalebox{1.1}[1.1]{%
      $\mathcal{#1}$%
    }%
    }%
}

\newcommand{\ctx}{\V{\myCalBig{E}}}

\providecommand\T{}
\renewcommand*{\T}{\V{\NonTerminale{T}}}

\newcommand*{\x}{\V{\NonTerminale{x}}}

\providecommand\f{}
\renewcommand*{\f}{\NonTerminale{f}}

\newcommand*{\es}{\Many\e}

\newcommand*{\e}{\V{\NonTerminale{e}}}

\providecommand\m{}
\renewcommand*{\m}
{\V{\NonTerminale{m}}}
\newcommand*{\mdf}{%
  \V\mu
}

\newcommand*{\semiColon}{\semicolon}
\newcommand*{\semicolon}{\Terminale{;}}

\newcommand*{\commaSign}{\Terminale{,}}
\ifdefined\comma\relax\else%
\newcommand\comma\commaSign
\fi%

\newcommand*{\oRound}{\Terminale{(}}
\newcommand*{\cRound}{\Terminale{)}}
\newcommand*{\oCurly}{\Terminale{\{}}
\newcommand*{\cCurly}{\Terminale{\}}}

\newcommand*{\oR}{\oRound}
\newcommand*{\cR}{\cRound}
\newcommand*{\oC}{\oCurly}
\newcommand*{\cC}{\cCurly}

\newcommand*{\singleDot}{{\Terminale{.}}}

\newcommand{\postApex}
  {\Terminale{\ensuremath{\raisebox{0.2em}{\mbox{\Q@^@}}}}}

\newcommand*{\PostMdf}[1]{\ensuremath{\mbox{\Q@`@}{}^{#1}}}
\newcommand*{\PreMdf}[1]{\ensuremath{\mbox{\Q@'@}{}^{#1}}}
\ifdefined\Case%
\renewcommand\Case[1]{\raisebox{0.15ex}{\scalebox{0.65}{\texttt{\textbf{\!\!\!\!\!(\!#1\!)}}}}}%
\else%
\newcommand\Case[1]{\raisebox{0.15ex}{\scalebox{0.65}{\texttt{\textbf{\!\!\!\!\!(\!#1\!)}}}}}%
\fi%

\newcommand*{\phfj}{\textsc{FJ}\ensuremath{{}^{\mbox{\Q@^@\!}}}}

\newcommand{\ExampleAndComment}[4]{
  \begin{minipage}{0.25\linewidth}
  \end{minipage}
  \begin{Scaled}{0.85}{0.85}
  \begin{minipage}{0.85\linewidth}
  \end{minipage}
  \end{Scaled}
  }

\newcommand\Opt[1]{\V{{%
  \hspace{0.03ex}
  \overbracket[0.10ex][0.3ex]{\hspace{-0.03ex}#1\hspace{-0.03ex}}}}%
\hspace{0.03ex}
}

\newlength\stageHeight
\newlength\stageWidth

\newcommand{\equals}{\Terminale{$\customColon{2}$}}
\renewcommand{\equals}{\Kw{ = }}

\newcommand\vs{\Many\Xv}

\newlength\colonwidth
\newlength\colonheight

\newcommand\scaleKw{0.6}
\newcommand\Kw[1]{\scalebox{\scaleKw}[1]{\Terminale{#1}}}
\newcommand\ExtractMType[2]
{\Aux{extractMType}_{#1}(#2)}

\newcommand*\ExtractMTypes\ExtractMType

\renewcommand\ldots{\scalebox{0.75}{$...$}}
\edef\myMapsto{\mapsto}
\renewcommand\mapsto{\scalebox{0.75}[1.2]{$\myMapsto\,$}}

\newcommand\FRestriction[2]{{
  \left.\kern-\nulldelimiterspace % automatically resize the bar with \right
  #1 % the function
  \vphantom{\big|} % pretend it's a little taller at normal size
  \right|_{#2} % this is the delimiter
  }}

\ifdefined\block%
  \renewcommand\block{\V{\cal{B}}}
\else%
  \newcommand\block{\V{\cal{B}}}
\fi%

\definecolor{darkRed}{RGB}{100,0,10}
\definecolor{darkBlue}{RGB}{10,0,100}
\newcommand*{\ttfamilywithbold}{\ttfamily}

\lstdefinelanguage{FortyTwo}[]{Java}{morekeywords={%
  M,
  exception,error,mut,imm,
  read,capsule,lent,assert
  with,in,immutable,trait,using,
  on,var,loop,reuse,method,is
  },
   basicstyle=\ttfamily,
   keywordstyle=\ttfamilywithbold\bfseries\color{darkRed},
   identifierstyle=\idstyle,
   showstringspaces=false,
   mathescape=true,
   xleftmargin=0pt,
   xrightmargin=0pt,
   breaklines=false,
   breakatwhitespace=false,
   breakautoindent=false,
   tabsize=2,
   commentstyle=\color{darkBlue}\ttfamily,
   stringstyle=\color{darkRed}\ttfamily,
   literate=
                 {\%}{{\mbox{\textbf{\%}}}}1
                 {~} {$\sim$}1
 }

\newcommand{\Comment}[1]{%
  \texttt{\textbf{\color{darkBlue}{//#1}}}%
}

\newcommand{\Q}{\lstinline}

\makeatletter
\newcommand*\idstyle{%
        \expandafter\id@style\the\lst@token\relax
}
\def\id@style#1#2\relax{%
        \ifcat#1\relax\else
                \ifnum`#1=\uccode`#1%
                        \ttfamilywithbold\bfseries
                \fi
        \fi
}
\makeatother

\newcommand{\saveSpace}{\vspace{-3px}}
\newcommand{\loseSpace}{\vspace{1ex}}

\newcommand{\subheading}[1]{%
	\loseSpace%
	\noindent\textsf{\textbf{\large#1\\\noindent}}
}

\title{Sound Invariant Checking Using Type Modifiers and Object Capabilities.}
\author{Isaac Oscar Gariano}{Victoria University of Wellington}{isaac@ecs.vuw.ac.nz}{}{}
\author{Marco Servetto}{Victoria University of Wellington}{marco.servetto@ecs.vuw.ac.nz}{}{}
\author{Alex Potanin}{Victoria University of Wellington}{alex@ecs.vuw.ac.nz}{}{}
\authorrunning{Isaac O.\,G., M. Servetto, and A. Potanin} %mandatory. First: Use abbreviated first/middle names. Second (only in severe cases): Use first author plus 'et al.'
\Copyright{Isaac Oscar Gariano, Marco Servetto, and Alex Potanin} %mandatory, please use full first names. LIPIcs license is "CC-BY";  http://creativecommons.org/licenses/by/3.0/

\subjclass{, 
	, 
	}% mandatory: Please choose ACM 2012 classifications from https://www.acm.org/publications/class-2012 or https://dl.acm.org/ccs/ccs_flat.cfm . E.g., cite as "General and reference $\rightarrow$ General literature" or \ccsdesc[100]{General and reference~General lwrapfigureiterature}. 

%mandatory
\EventEditors{John Q. Open and Joan R. Access}
\EventNoEds{2}
\EventLongTitle{42nd Conference on Very Important Topics (CVIT 2016)}
\EventShortTitle{CVIT 2016}
\EventAcronym{CVIT}
\EventYear{2016}
\EventDate{December 24--27, 2016}
\EventLocation{Little Whinging, United Kingdom}
\EventLogo{}
\SeriesVolume{42}
\ArticleNo{23}
\let\origSingleDot=\singleDot % reduce spacing arround the dot operator
\renewcommand{\singleDot}{\kern-1.2pt\origSingleDot\kern-1.8pt} 

\newcommand{\defActiveMathChar}[2]{%
	\begingroup\lccode`~=`#1\relax%
	\lowercase{\endgroup\def~}{#2}%
	\AtBeginDocument{\mathcode`#1="8000}%
}
\defActiveMathChar{.}{\singleDot}

\newcommand{\M}[3]{\ensuremath{\Kw{M}\oR{}#1\semiColon{}#2\semiColon{}#3\cR}}

\let\origMapsTo=\mapsto % put proper spacing arround the mapsto arrow
\renewcommand{\mapsto}{\mathrel{\origMapsTo}}
\newcommand{\invariant}{\Kw{invariant}\oR\cR}
\newcommand{\thm}[1]{\scalebox{0.9}[0.9]{\sf #1}}

\begin{document}
\maketitle
\begin{abstract}
In this paper we use pre existing language support for type modifiers and object capabilities to enable a system for sound runtime verification of invariants.
Our system guarantees that class invariants hold for all objects involved in execution.
Invariants are specified simply as methods whose execution is statically guaranteed to be deterministic and not access any externally mutable state.
We automatically call such invariant methods only when objects are created or the state they refer to may have been mutated.
Our design restricts the range of expressible invariants but improves upon the usability and performance of our system compared to prior work.
In addition, we soundly support mutation, dynamic dispatch, exceptions, and non determinism, while requiring only a modest amount of annotation.

We present a case study showing that our system requires a lower annotation burden compared to Spec\#, and  performs orders of magnitude less runtime invariant checks compared to the widely used `visible state semantics' protocols of D, Eiffel.
We also formalise our approach and prove that such pre existing type modifier and object capability support is sufficient to ensure its soundness.
\end{abstract}

\section{Introduction}
\label{s:intro}
%\newpage
%\LINE

Object oriented programming languages provide great flexibility through subtyping and dynamic dispatch: they
allow code to be adapted and specialised to behave differently in different contexts.
%%, which is made even more complex by dynamic class loading (supported by many mainstream OO languages).
However this flexibility hampers code reasoning, since object behaviour is usually nearly completely
unrestricted. This is further complicated with the support OO languages typically have for exceptions,
memory mutation, and I/O.
% Class invariants are a well known technique to help write correct code, however there are various different interpretations of when they should hold.
% invariant protocols, specifying when the invariant is expected to hold and when is checked. 
%% In the absence of on the fly static verification of dynamically loaded code, it is difficult for programmers to write code that is correct in a library setting.

Class invariants are an important concept when reasoning about software correctness.
They can be presented as documentation, checked as part of static verification, or, as we do in this paper, monitored for violations using runtime verification.
In our system, a class specifies its invariant by defining a boolean method called \Q@invariant@.
We say that an object's invariant holds when its \Q@invariant@ method would return \Q@true@. 
We do this, like Dafny~\cite{DBLP:conf/sigada/Leino12}, to minimise special syntactic and type-system treatment of invariants, making them easier to understand for users. Whereas most other approaches treat invariants as a special annotation with its own syntax.

An \emph{invariant protocol}~\cite{FlexibleInvariants} specifies when invariants need to be checked, and when they can be assumed; if such checks guarantee said assumptions, the protocol is sound.
The two main sound invariant protocols present in literature are \emph{visible state semantic} \cite{Meyer:1988:OSC:534929} and the \emph{Boogie/Pack-Unpack methodology}~\cite{DBLP:journals/jot/BarnettDFLS04}. The visible state semantics expect the invariants of receivers to hold before and after every public method call, and after constructors. Invariants are simply checked at all such points, thus this approach is obviously sound; however this can be incredibly inefficient, even in simple cases.
In contrast, the pack/unpack methodology marks all objects as either \emph{packed} or \emph{unpacked}, where a packed object is one whose invariant is expected to hold.
In this approach, an object's invariant is checked only by the pack operation.
In order for this to be sound, some form of aliasing and/or mutation control is necessary. For example, Spec\#, which follows the pack/unpack methodology, uses a theorem prover, together with source code annotations.
While Spec\# can be used for full static verification, it conveniently allows invariant checks to be performed
at runtime, 
whilst statically verifying aliasing, purity and other similar standard properties.
This allows us to closely compare our approach with Spec\#.

Instead of using automated theorem proving, 
it is becoming more popular to verify aliasing and immutability using a type system.
For example, three languages: L42~\cite{ServettoZucca15,ServettoEtAl13a,JOT:issue_2011_01/article1,GianniniEtAl16}, Pony~\cite{clebsch2015deny,clebsch2017orca}, and the language of Gordon et.~al.~\cite{GordonEtAl12} use Type Modifiers (TMs) and Object Capabilities (OCs) to ensure safe and deterministic parallelism.%
\footnote{TMs are called \emph{reference capabilities} in other works. We use the term TM here
to not confuse them with object capabilities, another technique we also use in this paper.}
While studying those languages, we discovered an elegant way to enforce invariants.

We use the guarantees provided by these systems to ensure that that at all times, if an object is usable in execution, its invariant holds. What this means is that if you can do anything with an object, such as by using it as an argument/receiver of a method call, we know that the invariant of it, and all objects reachable from it, holds. In order to achieve this, we use TMs and OCs to restrict how the result of invariant methods may change, this is done by restricting I/O as well as what state the invariant can refer to and what can alias/mutate such state.  We use these restrictions to reason as to when an object’s invariant could have been violated, and when such object can next be used, we then inject a runtime check between these two points. See Section \ref{s:protocol} for the exact details of our invariant protocol.

\subheading{Example}
Here we show an example illustrating our system in action. Suppose we have a \Q@Cage@ class which contains a \Q@Hamster@; the \Q@Cage@ will move its \Q@Hamster@ along a path. We would like to ensure that the \Q@Hamster@ does not deviate from the path. We can express this as the invariant of \Q@Cage@: the position of the \Q@Cage@'s \Q@Hamster@ must be within the path (stored as a field of \Q@Cage@).

%
% While Spec\# requires specialised \Q@Point@, \Q@Hamster@, and \Q@Cage@ declarations to be able to enforce the invariant, our version manages to capture the required information in just a few annotations on \Q@Cage@ and leaves \Q@Point@ and \Q@Hamster@ unmodified.
%	if(that==null || !(that instanceof Point)){return false;}
% 	return ((Point)that).x==this.x && ((Point)that).y==this.y; 
%  }
\begin{lstlisting}
class Point { Double x; Double y; Point(Double x, Double y) {..}
  @Override read method Bool equals(read Object that) {
    return that instanceof Point &&
      this.x == ((Point)that).x && this.y == ((Point)that).y; }
}
class Hamster {Point pos; //pos is imm by default
  Hamster(Point pos) {..} 
}
class Cage {
  capsule Hamster h;
  List<Point> path; //path is imm by default
  Cage(capsule Hamster h, List<Point> path) {..}
  read method Bool invariant() {
    return this.path.contains(this.h.pos); }
  mut method Void move() {
    Int index = 1 + this.path.indexOf(this.h.pos));
    this.moveTo(this.path.get(index % this.path.size())); }
  mut method Void moveTo(Point p) { this.h.pos = p; }
}
\end{lstlisting}

Many verification approaches take advantage of the separation between primitive/value types and objects, since the former are immutable and do not support reference equality.
However, our approach works in a pure OO setting without such a distinction. Hence we write all type names in \Q@BoldTitleCase@ to underline this. Note: to save space, here and in the rest of the paper we omit the bodies of constructors that simply initialise fields with the values of constructor parameters, but we show their signature in order to show any annotations.

We use the \Q@read@ annotation on \Q@equals@ to express that it does not modify either the
receiver or the parameter. In \Q@Cage@ we use 
the \Q@capsule@ annotation to ensure
that the \Q@Hamster@'s \emph{reachable object graph} (ROG) is fully under the control
of the containing \Q@Cage@. 
We annotated the \Q@move@
and \Q@moveTo@ methods with \Q@mut@, since they modify
their receivers ROG. The default annotation is always \Q@imm@, thus \Q@Cage@'s \Q@path@ field is a deeply immutable list of \Q@Point@s.
% Note how we just use \Q@List.contains()@ and \Q@List.indexOf()@
% to check if the hamster position is inside the list.
% The conventional syntax correctly instantiates a \Q@Cage@:
% \Q@new Cage(new Hamster(new Point(..)), List.of(new Point(...))@.
Our system performs runtime checks for the invariant
at the end of \Q@Cage@'s constructor, \Q@moveTo@ method, and after any update to one of its fields.
The \Q@moveTo@ method is the only one that may (directly) break the \Q@Cage@'s invariant. However, there is only a single occurrence of \Q@this@ and it is used to read the \Q@h@ field. We use the guarantees of TMs to ensure that no alias to \Q@this@ could be reachable from either \Q@h@ or the immutable \Q@Point@ parameter. Thus, the potentially broken \Q@this@ object is not visible while the \Q@Hamster@'s position is updated. 
The invariant is checked at the end of the \Q@moveTo@ method, just before \Q@this@ would become visible again.
This technique loosely corresponds to an implicit pack and unpack: we use \Q@this@ only to read the field value, then we work on its value while the invariant of \Q@this@ is not known to hold, finally we check the invariant before allowing the object to  be used again.

Note: since only \Q@Cage@ has an invariant,
 only \Q@Cage@ has special restrictions, allowing the code for \Q@Point@ and \Q@Hamster@ to be unremarkable.
 This is not the case in Spec\#: all code involved in  verification needs to be designed with verification in mind~\cite{barnett2011specification}.
% The best solution we found was to define our own equality for \Q@Point@ instead of relying on \Q@Object.Equals@,
% thus we could not use \Q@List.Contains@ and \Q@List.IndexOf@.

\subheading{Spec\# Example} Here we show the previous example in Spec\#, the system most similar to ours (see appendix \ref{s:hamster} for a more detailed discussion about this solution):
%or \small or \footnotesize etc.
\begin{lstlisting}[
language={[Sharp]C}, morekeywords={invariant,ensures,requires,expose,exists}]
// Note: assume everything is `public'
class Point { double x; double y; Point(double x, double y) {..}
  [Pure] bool Equal(double x, double y) {
    return x == this.x && y == this.y; }
}
class Hamster{[Peer]Point pos; 
  Hamster([Captured]Point pos){..}
}
class Cage {
  [Rep] Hamster h; [Rep, ElementsRep] List<Point> path;
  Cage([Captured] Hamster h, [Captured] List<Point> path)
    requires Owner.Same(Owner.ElementProxy(path), path); {
      this.h = h; this.path = path; base(); }
  invariant exists {int i in (0 : this.path.Count);
    this.path[i].Equal(this.h.pos.x, this.h.pos.y) };
  void Move() {
    int i = 0;
    while(i<path.Count && !path[i].Equal(h.pos.x,h.pos.y)){i++;}
    expose(this) {this.h.pos = this.path[i%this.path.Count];}}
}
\end{lstlisting}

In both versions, we designed \Q@Point@ and \Q@Hamster@ in a general way, and not solely to be used by classes with an invariant, in particular \Q@Point@ is not an immutable class. However, doing this in Spec\# proved difficult, in particular we were unable to override \Q@Object.Equals@, or even define a usable \Q@equals@ method that takes a \Q@Point@, as such we could not call either \Q@List<Point>.Contains@ or \Q@List<Point>.IndexOf@.
 
Even with all of the above annotations, we still needed special care in creating \Q@Cage@s:\vspace{-1.860px}% magic number that prevents the listings background going onto the next page
\begin{lstlisting}[
%basicstyle=\footnotesize,
language={[Sharp]C}, morekeywords={invariant,ensures,requires,expose,exists}]
List<Point> pl = new List<Point>{new Point(0,0),new Point(0,1)};
Owner.AssignSame(pl, Owner.ElementProxy(pl));
Cage c = new Cage(new Hamster(new Point(0, 0)), pl);
\end{lstlisting}

Whereas with our system we can simply write:
\begin{lstlisting}
List<Point> pl = List.of(new Point(0, 0), new Point(0, 1));
Cage c = new Cage(new Hamster(new Point(0, 0)), pl);
\end{lstlisting}

%3 read 2 capsule 3 mut extra method moveTo
%----
In Spec\# we had to add $10$ different annotations, of $8$ different kinds; some of which were quite involved. In comparison, our approach requires only $7$ simple keywords, of $3$ different kinds; however we needed to write 
a separate \Q@moveTo@ method, since we do not want to burden our language with extra constructs such as Spec\#'s \Q@expose@.
\subheading{Summary}
We have fully implemented our protocol in L42\footnote{A suitably anonymised, experimental version of L42, supporting the protocol described in this paper, together with the full code of our case studies, is available at \url{http://l42.is/EcoopArtifact.zip}.}\footnote{We also believe it would be easy to implement our protocol in Pony and Gordon et.~al.'s language.}, we used this implementation to implement and test an interactive GUI involving a class with an invariant. On a test case with $5$ objects with an invariant, 
our protocol performed only $77$ invariant checks, whereas the visible state semantic invariant protocols of D and Eiffel perform $53$ and $14$ million checks (respectively). See Section \ref{s:case-study} for an explanation of these result.
% \MS{The difference between D and Eiffel is an effect of Eiffel's `uniform access principle', where fields can be used to implements `interface methods', while not triggering invariant checks.}
We also compared with Spec\#, whose invariant protocol performs the same number of checks as ours, however the annotation burden was almost $4$ times higher than ours.

In this paper we argue that our protocol is not only more succinct than the pack/unpack approach, but is also easier and safer to use.
Moreover, our approach deals with more scenarios than most prior work: we allow sound catching of invariant failures and also carefully handle non deterministic operations like I/O.
%In our case study we show that
%we can still encode most of the examples explored in ~\cite{???} (including for example mutable collections of immutable objects) whilst having a significantly lower annotation-burden.
Section \ref{s:TMsAndOCs} explains the \emph{type modifier} and \emph{object capability} support we use for this work.
Section \ref{s:protocol} explains the details of our invariant protocol, and section \ref{s:formalism} formalises a language enforcing this protocol.
Sections \ref{s:immutable} and \ref{s:encapsulated}, respectively, explain and motivate how our protocol can handle invariants over immutable and encapsulated data.
Section \ref{s:case-study} presents our GUI case study and compares it against visible state semantics and Spec\#.
Sections \ref{s:related} and \ref{s:conclusion} provide related work and conclusions.

Appendix \ref{s:proof} provides a proof that our invariant protocol is sound. Appendices \ref{s:hamster} and \ref{s:MoreCaseStudies} provide further
case studies and comparisons against Spec\#, D and Eiffel.

%http://www.cs.cmu.edu/~NatProg/papers/p496-coblenz-Glacier-ICSE-2017.pdf

\section{Type Modifiers and Object Capabilities}
\label{s:TMsAndOCs}
Reasoning about imperative object oriented (OO) programs is a non trivial task,
made particularly difficult by mutation, aliasing, dynamic dispatch, I/O, and exceptions. There are many ways to perform such reasoning, here we use the type system to restrict, but not prevent such behaviour in order to be able to soundly enforce invariants with runtime verification (RV).
% [dynamic class loading],

\subheading{Type Modifiers (TMs)}
TMs, as used in this paper, are a type system feature that allows reasoning about aliasing and mutation. Recently a new design for them has emerged that radically improves their usability;
three different research languages are being independently developed relying on this new design: the language of Gordon et.~al.~\cite{GordonEtAl12}, Pony~\cite{clebsch2015deny,clebsch2017orca}, and L42~\cite{ServettoZucca15,ServettoEtAl13a,JOT:issue_2011_01/article1,GianniniEtAl16}.
These projects are quite large: several million lines of code are written in Gordon et.~al.'s language and are used by a large private Microsoft project; Pony and L42 have large libraries and are active open source projects. In particular the TMs of these languages are used to provide automatic and correct parallelism~\cite{GordonEtAl12,clebsch2015deny,clebsch2017orca,ServettoEtAl13a}.

While we focus on the specific TMs provided by L42, Pony, and Gordon et.~al., type modifiers
 are a well known language mechanism~\cite{TschantzErnst05,BirkaErnst04,OstlundEtAl08,clebsch2015deny,GianniniEtAl16,GordonEtAl12}
 that allow statically reasoning about mutability and aliasing properties of objects.
With slightly different names and semantics, the four most common modifiers for references to objects are:
\begin{itemize}
\item Mutable (\Q@mut@): the referenced object can be mutated, as in most imperative languages without modifiers.
If all types are \Q@mut@, there is no restriction on aliasing/mutation.
\item Readonly (\Q@read@): the referenced object cannot be mutated by such references, but there may be mutable aliases to such object, thus mutation can still be observed. 
\item Immutable (\Q@imm@): the referenced object can never mutate. Like \Q@read@ references, one cannot mutate through an \Q@imm@ reference, however \Q@imm@ references also guarantee that the referenced object will not mutate through any other alias.
\item Encapsulated (\Q@capsule@):
 everything in the reachable object graph (ROG) of a capsule reference (including itself) is mutable only through that reference; however immutable references can be freely shared across capsule boundaries.
\end{itemize}
%In the context of object-oriented programming, type modifiers may also apply to the implicit \Q@this@ parameter in method declarations, restricting the type of references the method can be called on. In addition, due to the deep meanings we type field access on object references to be the most restrictive of the object references modifier and the field’s. As \Q@read@ references impose no assumptions about aliasing, any \Q@imm@ or \Q@mut@ expression can be safely implicitly promoted to \Q@read@, whereas other conversions are not generally safe.
%\loseSpace

\noindent TMs are different to field or variable modifiers like Java's \Q@final@: TMs apply to references, whereas \Q@final@ applies to fields themselves. Unlike a variable/field of a \Q@read@ type, a \Q@final@ variable/field cannot be reassigned, it always refers to the same object, however the variable/field can still be used to mutate the referenced object.
On the other hand, an object cannot be mutated through a \Q@read@ reference, however a \Q@read@ variable can still be reassigned.\footnote{In C, this is similar to the difference between \Q@A* const@ (like \Q@final@) and \Q@const A*@ (like \Q@read@), where \Q@const A* const@ is like \Q@final read@.}
%\end{itemize}

Consider the following  example usage of \Q@mut@, \Q@imm@, and \Q@read@, where we can observe a change in \Q@rp@ caused by a mutation inside \Q@mp@.
\begin{lstlisting}
mut Point mp = new Point(1, 2);
mp.x = 3; // ok
imm Point ip = new Point(1, 2);
$\Comment{}$ip.x = 3; // type error
read Point rp = mp; // ok, read is a common supertype of imm/mut
$\Comment{}$rp.x = 3; // type error
mp.x = 5; // ok, now we can observe rp.x == 5
ip = new Point(3, 5); // ok, ip is not final
\end{lstlisting}

There are several possible interpretations of the semantics of type modifiers.
Here we assume the full/deep meaning~\cite{ZibinEtAl10,Potanin2013}:
\begin{itemize}
  \item the objects in the ROG of an immutable object are immutable,
  \item a mutable field accessed from a \Q@read@ reference produces a \Q@read@ reference,
%  \item no \emph{down}-casting is allowed between different type modifiers.
  \item no casting/promotion from \Q@read@ to \Q@mut@ is allowed.
%  \item promotion, is a type-system feature allowing implicit and safe casting from \Q@read@ and \Q@mut@ to \Q@imm@.
\end{itemize}

\noindent There are many different existing techniques and type systems that handle the modifiers above~\cite{ZibinEtAl10,ClarkeWrigstad03,HallerOdersky10,GordonEtAl12,ServettoZucca15}.
The main progress in the last few years is with the flexibility of such type systems:
 where the programmer should use \Q@imm@ when  representing immutable data
and \Q@mut@ nearly everywhere else. The system will be able to transparently promote/recover~\cite{GordonEtAl12,clebsch2015deny,ServettoZucca15} the type modifiers, adapting them to their use context.
To see a glimpse of this flexibility, consider the following example:
%//the same expression can create mut, imm or capsule
%\saveSpace
\begin{lstlisting}
    mut Circle mc = new Circle(new Point(0, 0), 7);
capsule Circle cc = new Circle(new Point(0, 0), 7);
    imm Circle ic = new Circle(new Point(0, 0), 7);
\end{lstlisting}

Here \Q@mc@, \Q@cc@, and \Q@ic@ are syntactically initialised with the same expression: \Q@new Circle(..)@.
The \Q@new@ expression returns a \Q@mut@, so \Q@mc@ is obviously ok.
%\footnote{Capsules must encapsulate their entire ROG, thus a \Q@new@ expression
%can not directly return \Q@capsule@ in the case of objects with \Q@mut@ fields.}
Moreover, the expression does not use any \Q@mut@ local variables, thus the flexible TM system
allows the \Q@mut@ result to be promoted to \Q@capsule@, thus \Q@cc@ is ok. 
Additionally, a \Q@capsule@ can be implicitly converted to \Q@imm@, thus \Q@ic@ is also ok.
We want to emphasise that this is not a special feature of \Q@new@ expressions:
any expression of a \Q@mut@ type that uses no free \Q@mut@ variables declared outside can be implicitly promoted to \Q@capsule@/\Q@imm@.\footnote{%
This requires some restrictions on \Q@read@ fields not discussed in detail for lack of space.
} This is the main improvement on the flexibility of TMs in recent literature~\cite{ServettoEtAl13a,ServettoZucca15,GordonEtAl12,clebsch2015deny,clebsch2017orca}.
Former work~\cite{Boyland10,boyland2003checking,Hogg91,Smith:2000:AT:645394.651903,DBLP:conf/pldi/AikenFKT03}, which eventually enabled the work of Gordon et.~al., does not consider promotion and 
infers uniqueness/isolation/immutability only when starting from references that have been tracked with restrictive annotations along their whole lifetime.
From a usability perspective, this improvement means that
these TMs are opt-in: a programmer can write large sections of code
mindlessly using \Q@mut@ types and be free to have rampant aliasing. 
Then, at a later stage, another programmer may still 
be able to encapsulate those data structures into an \Q@imm@ or \Q@capsule@ reference.

%\saveSpace
%\begin{lstlisting}
%mc.radius = 3; // ok
%imm Point ip = ic.center; // ok, ROG immutable
%read Circle rc = mc
%read Point rp = rc.center; // ok, fields of read Circle are read
%$\Comment{}$mut Point mp = rc.center; // type error
%\end{lstlisting}
%\saveSpace
%Such flexibility is also visible where \Q@rc.center@ returns a \Q@read@ but \Q@ic.center@ returns an \Q@imm@: any expression typed as \Q@read@ that only
%uses immutable variables can safely be promoted to \Q@imm@ or \Q@capsule@.

 %(since \Q@ic@ is \Q@imm@, and \Q@imm@ is a deep modifier).
%true fact but not sufficient?

%With this kind of type system, we can ensure immutable classes by just declaring all their fields as final and immutable.%
%\footnote{
%In Java,  to ensure a class is immutable we need:
%the class must be final, all the fields must be final of immutable
%classes (thus no interface fields, final classes all the way down),
%and the SecurityManager need to properly tame reflection.}

% Not sure about this paragraph:

The \Q@capsule@ modifier (sometimes called isolated/\Q@iso@) is possibly the one whose details differ the most in literature. Here we refer to the interpretation of~\cite{GordonEtAl12}, that introduced the concept of recovery/promotion.
This concept is the basis for L42, Pony, and Gordon et.~al.'s type systems~\cite{GordonEtAl12,ServettoEtAl13a,ServettoZucca15,ServettoEtAl13a,clebsch2015deny,clebsch2017orca}. 

%\begin{itemize}
%	\item A capsule local variable can only be used once. %as \Q@capsule@ or \Q@mut@.
%	\item Only a capsule expression can be used to initialize or update a \Q@capsule@ field.
%	\item A capsule field access has the same type modifier as the receiver.
%	\item An expression of a \Q@mut@ type that uses no \Q@mut@ variables declared outside can be implicitly promoted to \Q@capsule@. Promotion/recovery is the main improvement on the flexibility of TM in recent literature~\cite{ServettoEtAl13a,ServettoZucca15,GordonEtAl12,clebsch2015deny,clebsch2017orca}
%	(this requires some restrictions on \Q@read@ fields that we do not discuss in detail for lack of space).
% \end{itemize}

%This is to ensure the capsule doesn't `leak', potentially violating it's exclusivity,

The capsule/isolated fields of Gordon et.~al. and Pony rely on destructive reads~\cite{GordonEtAl12,clebsch2015deny}: in order to read them, a new value (such as \Q!null!) will be assigned to them. In contrast, L42~\cite{ServettoEtAl13a,ServettoZucca15} does not require such destructive reads, thus \Q@capsule@ fields can be accessed many times, and their content can be seen from outside; but only in controlled ways.
Both Gordon et.~al. and Pony restrict how \Q!capsule! local variables can be used by changing the type they are seen as, however both allow the local variable to be `consumed', allowing them to be used as normal capsule/isolated expressions, at the cost of being unable to use the variable again. L42 however uses a simpler approach where all accesses to \Q!capsule! local variables consume them: they are expressed using linear/affine types~\cite{boyland2001alias}, thus they can only be used once.
%Both the capsule/isolated fields and variables of Pony and Gordon et.~al. rely on destructive reads~\cite{GordonEtAl12,clebsch2015deny}: reading such fields replaces their values with \Q@null@.+++ and a static analysis ensures capsule local variables can be accessed only once after the initialization and every update.

%In contrast,  L42~\cite{ServettoEtAl13a,ServettoZucca15} do not require destructive reads, and treat %\Q@capsule@ local variables and \Q@capsule@ fields differently:
%\Q@capsule@ local variables are expressed using linear/affine types~\cite{boyland2001alias}, thus they can only be used once;
%\Q@capsule@ fields can be accessed many times,
%and thus their content can be seen from outside; but only in controlled ways.

%while M\# and Pony requires both capsule fields and capsule variables to be `balloons'~\cite{Almeida97,ServettoEtAl13a} in the object graph.

%Destructive reads would be a bad idea for validation as they would likely invalidate objects.

%x=this.#f()
%..do all you want with x...
%// invariant check here!!
%.
%this.f(x -> ..)
%this.f=transform(this.f)
%invariantCheck()
%transform(this.#f())
%invariantCheck()

%\loseSpace
\subheading{Exceptions}\label{s:exceptions}
In most languages exceptions may be thrown at any point; combined with mutation this complicates reasoning about the state of programs after exceptions are caught: if an exception was thrown whilst mutating an object, what state is that object in? Does its invariant hold?
The concept of \emph{strong exception safety} (SES)~\cite{Abrahams2000,JOT:issue_2011_01/article1} simplifies reasoning:
if a \Q@try@--\Q@catch@ block caught an exception, the state visible before execution of the \Q@try@ block is unchanged, and the exception object does not expose any object that was being mutated.
%\LINE
%\noindent{\textit{Exceptions:}}
% M\#, L42 and Pony rely on SES for all unchecked exceptions to ensure safe and transparent parallelism,
% They wish to ensure the code behave as if the execution was fully sequential.
% Exceptions create additional difficulties in such context: if two operations are running in parallel in
% a fork-join, and the first one produces an exception, it should be safe to cancel the other operation and
% to propagate the exception outwards. The system need to guarantee
% the progress the second operation accumulated is not observable.
% Pony avoids this problem simply by not supporting exceptions;
% while
%M\# and L42 will parallelize only expressions that do not leak checked exceptions,
%and they enforce Strong Exception Safety(SES)~\cite{Abrahams2000} for unchecked exceptions.
%Other authors have identified the concept of SES as
% a general issue when reasoning about objects state after catching an exception.
% while we need it to soundly capture invariant failures.
L42 already enforces SES for unchecked exceptions.\footnote{%
This is needed to support safe parallelism. Pony takes a more drastic approach and does not support exceptions in the first place. 
We are not aware of how Gordon et.~al. handles exceptions, however in order for it to have sound unobservable parallelism it must have some restrictions.%
%We do not know how M\# conciliate deterministic parallelism and unchecked exceptions, we suspect some variation of SES must be in place.
}
L42 enforces SES using TMs in the following way:\footnote{Transactions are another way of enforcing strong exception safety, but they require specialized and costly run time support.}\footnote{A formal proof of why these restriction are sufficient is presented in the work of Lagorio~\cite{JOT:issue_2011_01/article1}.}
\begin{itemize}
\item Code inside a \Q@try@ block capturing unchecked exceptions is typed as if all \Q@mut@ variables declared outside of the block are \Q@read@.
\item Only \Q@imm@ objects may be thrown as unchecked exceptions.
\end{itemize} 
%Of course this has the effect that even if no-exception is thrown, no mutation could have occured, which is an even stronger property than SES, other work is more flexible~\cite{?}, at the cost of more complicated typing rules.
%With SES we can soundly capture invariant-failures as an exception, since any mutation that caused the invariant failure cannot be observed. However, we also need to prevent a broke-object from being reachable from the exception object; since the only way a broken-object can be seen is within the \Q@read@ \Q@invariant@  method, it follows that if the exception-object contains no \Q@read@ references in its ROG it cannot leak a broken object. Preventing this in the-typsystem is non-trivial, so instead we simply require that:

\noindent This strategy does not restrict throwing exceptions, but only catching unchecked ones.
SES allows us to soundly capture invariant failures as unchecked exceptions: 
the broken object is guaranteed to be garbage collectable when the exception is captured. For the purposes of soundly catching invariant failures, it would be sufficient to enforce SES only when capturing exceptions caused by such failures.
%The ability to catch and recover from such failures is extremely useful as it allows the program to take corrective action.(DUPLICATED)

% We think this restriction is acceptable for run time verification, other works are much more restrictive,

%The above rules need only be enforced for catch blocks that could catch invariant-failures (including exceptions thrown within execution of \Q@invariant@) itself;
%, and since \Q@invariant@ declares no checked exceptions, this includes all exceptions throw-able by it.

% TMs are very useful in restricting the scope of mutation. 
% Any expression that does not use any \Q@mut@ 
% variable declared outside of such expression does not modify objects visible outside.
% With this observation in mind, we can use TMs to enforce SES in the following way:\footnote{
% 

% \begin{itemize}
% \item all thrown exceptions are immutable objects,
% \item 
% \end{itemize}

% For the aim of enforcing invariants, we could relax SES to hold only when capturing exceptions caused by invariant failures; but we are building on approaches that enforce SES on all unchecked exceptions .

% Intro to OCs

\subheading{Object Capabilities (OCs)}
OCs, which L42, Pony, and Gordon et.~al.'s work have, are a widely used~\cite{miller2003capability,
noble2016abstract,karger1988improving} programming style that allows associating resources with objects. When this style
is respected, code that does not possess an alias to such an object cannot use its associated resource.
%Object capabilities are programming style used to control and restrict use of operations such as access to external resources
Here, as in Gordon et.~al.'s work, we use OCs to reason about determinism and I/O. To properly enforce this, the OC style needs to be respected while implementing the primitives of the standard library and when performing foreign function calls that could be non deterministic, such as operations that read from files or generate random numbers. Such operations would not be provided by static methods, but instead instance methods of classes whose instantiation is kept under control. 
% \noindent\REVComm{\textit{Object Capabilities:}}{2}{Citations here?}
% While type modifiers are statically verified properties of references, object capabilities are run time characteristics of specific objects.

% Conceptually, an object capability is a communicable, unforgeable token of authority, a key to access special functionality: only certain objects with `special' powers can do `special' actions, and those objects are obtained in a controlled way. We call such objects `capability objects'.

% Their main use case is to allow for fine grained control over what sections of code are allowed to do. 

\lstset{language=Java}
 For example, in Java, \Q@System.in@
 \lstset{language=FortyTwo} 
  is a \emph{capability object} that provides access to the standard input resource, however, as it is globally accessible it completely prevents reasoning about determinism. 
 % In contrast, in the object capability style, one would not have-global variables but have the main por

% a capability object (it has the capability to read input); however it is globally accessible: thus any code could use it, preventing reasoning about determinism.
In contrast, if Java were to respect the object capability style, the \Q@main@ method could take a \Q@System@ parameter, as in
 \Q@main(mut System s)@
 \lstset{language=Java}
\Q@{.. s.in.read() ..}@. \lstset{language=FortyTwo}%
Calling methods on that \Q@System@ instance would be the only way to perform I/O;
moreover, the only \Q@System@ instance would be the one created by the runtime system before calling \Q@main@. % would have no usable constructor, and all its I/O methods would require a mutable (\Q@mut@) receiver.
% Other non deterministic operations would also work this.
%may just take a \Q@mut System@ object as a parameter.
% could also work this way.
This design has been explored by Joe-E~\cite{finifter2008verifiable}.
OCs are typically not part of the type system nor do they require runtime checks or special support beyond that provided by a memory safe language. However, since
L42 allows user code to perform foreign calls without going through a predefined standard library, its type system enforces the OC pattern over such calls:
%To reason about determinism, L42 connects TMs with the OC style as follows: % style by requiring:
\begin{itemize}
\item Foreign methods (which have not been whitelisted as deterministic) and methods whose names start with \texttt{\#\$} are \emph{capability methods}.%
\item Constructors of classes declared as \emph{capability classes} are also capability methods.
\item Capability methods can only be called by other capability-methods or \Q@mut@/\Q@capsule@ methods of capability classes.
\item In L42 there is no \Q@main@ method, rather it has several main expressions; such expressions can also call capability methods, thus they can instantiate capability objects and pass them around to the rest of the program.
% \item Any method that uses non deterministic primitive operations (such as native calls or access to global variables\footnote{ Even just allowing unrestricted access to \Q@imm@ global variables would prevent reasoning over determinism due to the possibility of global variable updates; however constant/final globals of an \Q@imm@ type would not cause such problems.
% }) must be an instance method requiring a \Q@mut@ receiver.
% Classes having such methods are \emph{capability classes}, and their instances are \emph{capability objects}.
% \item A capability object can only be created inside a \Q@mut@ method of a capability class; or
% by the runtime system, and passed to the main method.

% \item If the language has global variables, they should only be 
%\item There are no global variables.\footnote{}
\end{itemize}

\noindent L42 expects capability methods to be used mostly internally by capability classes, whereas user code would call normal methods on already existing capability objects.

For the purposes of invariant checking, we only care about the effects that methods could have on the running program and heap. As such, \emph{output} methods (such as a \Q@print@ method) can be whitelisted as `deterministic', provided they do not affect program execution, such as by non deterministically throwing I/O errors.

\subheading{Purity}\label{s:purity}
TMs and OCs together statically guarantee that any method with only \Q!read! or \Q!imm! parameters (including the receiver) is \emph{pure}; we define pure
as being deterministic and not mutating existing memory. Such methods are pure because:
\begin{itemize}
	\item the ROG of the parameters (including \Q!this!) is only accessible as \Q@read@ (or \Q@imm@), thus it cannot be mutated\footnote{This is even true in the concurrent environments of Pony and Gordon, since they ensure that no other thread/actor has access to a \Q@mut@/\Q@capsule@ alias of \Q@this@. 
	Thus, since such methods do not write to memory accessible by another thread, nor read memory that could be mutated by another thread, they are atomic.},
	\item if a capability object is in the ROG of any of the arguments (including the receiver), then it can only be accessed as \Q@read@, preventing calling any non deterministic (capability) methods,
	\item no other preexisting objects are accessible (as L42 does not have global variables).\footnote{%
		If L42 did have static variables, getters and setters for them would be capability methods.
		Even allowing unrestricted access to \Q@imm@
		static variables would prevent reasoning over
		determinism, due to the possibility of global variable
		updates; however constant/final globals of an 
		\Q@imm@ type would not cause such problems.%
	}
\end{itemize}

\section{Our Invariant Protocol}
\label{s:protocol}
Our invariant protocol guarantees that the whole ROG of any object involved in execution (formally, in a redex) is \emph{valid}: if you can call methods on an object, calling \Q@invariant@ on it is guaranteed to return \Q@true@ in a finite number of steps. However, calls to \Q!invariant! that are generated by our runtime monitoring (see below) can access the fields of a potentially invalid \Q!this!. This is necessary to allow for the \Q!invariant! method to do its job: namely distinguish between valid and invalid objects. However, as for any other method, calls to \Q!invariant! written explicitly by users are guaranteed to have a valid receiver. 

% However, the \Q!invariant! method itself needs to be able to operate on a potentially invalid \Q!this!, this will only happen when it is automatically called by the language itself, not by explicit calls present in the source code.
% Clearly the \Q@invariant@ method must be able to take an invalid \Q@this@, since the purpose of such method is to distinguish valid and invalid objects. On a first look this may seem an open contradiction
% with the aim of this work, however only calls to \Q@invariant@ inserted by the language semantics can take an invalid \Q@this@. As for any other method, when the application code can call \Q@invariant@, \Q@this@ is guaranteed to be valid.
%

%Logically, there are two reasons to access a field: we may wish to read the information stored in such object or we wish to mutate the object contained in the field.
%For the first case, we can type the field access as \Q@read@, but in the second case we
%need to type it as \Q@mut@. 
%We call `capsule mutators' a method accessing as \Q@mut@ a capsule field referenced in the invariant.
%We will show how capsule mutators are analogous of the pack/unpack/expose~\cite{???}.
%In order for a class to have an invariant under our protocol,
%\Q@invariant@ method the form 
% Can  a program write say mut method invariant or is it syntactically [???]
For simplicity, in the following explanation and in our formalism
we require 
%the \Q@this@ receiver is 
receivers to always be specified explicitly, and require that the receivers of field accesses and updates are always \Q!this!; that is, all fields are instance private.
We also do not allow explicit constructor definitions, instead we assume constructors are of the standard form \Q@$C$($T_1 x_1$,$\ldots$,$T_n x_n$) {this.$f_1$=$x_1$;$\ldots$;this.$f_n$=$x_n$;}@, where the fields of $C$ are $T_1 f_1;\ldots; T_n f_n;$. This ensures that partially uninitialised (and likely invalid) objects are not passed around or used. 
These restrictions only apply to our formalism; our code examples and the L42 implementation soundly relax these, see below for a discussion.% of why this is sound.
%We will later explain how these and other restrictions can be partially relaxed, as in the code examples.

\subheading{Invariants}
We require that all classes contain a \Q@read method Bool invariant() {..}@, if no \Q!invariant! method is present, a trivial one returning \Q!true! will be assumed. As this method only takes a \Q!read! parameter (the receiver), we can be sure that it is pure \footnote{If the invariant were not pure, it would be nearly impossible to ensure that it would return \Q@true@ at any point.}, as discussed in Section \ref{s:purity}.
The bodies of \Q@invariant@ methods are limited in their usage of \Q@this@: \Q!this! can only be used to access \Q@imm@ and \Q@capsule@ fields. This restriction ensures that 
an invalid \Q@this@ cannot be passed around.
We prevent accessing \Q@mut@ fields since their ROG could be changed by unrelated code (see Section \ref{s:immutable}).
Note that we do not require such fields to be \Q@final@: when a field is updated, we simply check the invariant of the receiver of the update.
% In order to prevent passing an invalid \Q@this@ to other methods.
%and unrelated code cannot break the invariant
%s of arbitrary objects,
%since a \Q!read! or \Q!mut! field could be modified through arbitrary aliases 
%(see Section \ref{s:immutable}). 

%To ensure that invariants cannot be broken by unrelated code (see Section \ref{s:immutable})  %
\subheading{Capsule mutators}
In order to allow complex mutations of objects with invariants we introduce the notion of \emph{capsule mutator}. A \emph{capsule mutator} can perform an arbitrarily complex mutation of the ROG of a capsule field. We use TMs to ensure that the object containing the capsule field is not usable whilst the fields ROG is mutated, and it's invariant is checked immediately afterwards. 

Formally, \emph{capsule mutators} are \Q@mut@ methods whose body accesses a \Q@capsule@ field mentioned in the invariant of the class containing the field. 
Capsule mutators must use \Q@this@ exactly once in their body, since fields are instance private, this will be to access the \Q!capsule! field.
%By construction, such use is the access of the \Q@capsule@ field.
Excluding the \Q!mut! receiver, such methods cannot have any \Q!mut! or \Q!read! parameters, their return type must not be \Q!mut!, and their \Q!throws! clause must be empty.%
\footnote{%
To allow capsule mutators to leak checked exceptions,
we would need check the invariant
when such exceptions are leaked. However, this would make the runtime semantics of checked exceptions inconsistent with unchecked ones.}.

%\end{itemize}
%Our type system will ensure that such methods are \Q!mut method!s, and the \Q!capsule! field will be seen as \Q!mut!.
As capsule mutators use \Q!this! only once, and have no \Q!read! or \Q!mut! parameters, \Q!this! will not be accessible during execution. This is important, as it allows the invariant to be violated part way through the capsule mutator, but re established by the end.
Preventing \Q!mut! return types ensures that such methods cannot leak out a mutable alias to the \Q!capsule! field, which could then be used to break the invariant.
Note that these restrictions do not apply when the receiver of the field access is \Q!capsule!, since we guarantee that the receiver is not in the ROG of any of its \Q!capsule! fields, and hence it can never be seen afterwards.

\subheading{Monitoring}
The language runtime will insert automatic calls to \Q!invariant!, if such a call returns \Q!false!, an unchecked exception will be thrown. Such calls are inserted in the following points:
\begin{itemize}
	\item After a constructor call, on the newly created object.
	\item After a field update, on the receiver.
	\item After a capsule mutator method returns, on the receiver of the method\footnote{The invariant is not checked if the call was terminated via an an unchecked exception, since strong exception safety guarantees the object will be unreachable anyway.}.
\end{itemize}

\noindent In Appendix \ref{s:proof}, we show that these checks, together with our aforementioned restrictions, are sufficient to ensure our guarantee that all objects involved in execution (except as part of an invariant check) are valid.

\subheading{Relaxations}
The above restrictions can be partially relaxed without breaking soundness, however this would not make the proof more interesting. In particular:
\begin{itemize}
	\item \Q!invariant! methods can be allowed to call instance methods that in turn only use \Q@this@ to read \Q!imm! or \Q!capsule!, or call other such instance methods. With this relaxation, the semantics of \Q@invariant@ needs to be understood with the body of those methods inlined; thus the semantics of the inlined code needs to be logically reinterpreted in the context of \Q@invariant@, where \Q@this@ may be invalid. In some sense, those inlined methods and field accesses can be thought of as macro expanded, and hence are not dynamically dispatched. Such inlining has been implemented in L42.
	\item We could allow all fields to be public, however \Q!capsule! fields, mentioned in the invariant of their containing class, should not be accessible over a \Q!mut! receiver other than \Q!this!. Even without this relaxation, however, getters and setters could be used to simulate public fields.
	\item Unrestricted readonly access to \Q!capsule! fields can be allowed by automatically generated getters of the form \Q!read method read C f() { return this.f; }!. Such getters are already a fundamental part of the L42 language.
	
	\item Java style constructors could be allowed, provided that \Q!this! is only used as the receiver of field initialisations. L42 does not provide such constructors, but one can always write a static factory method that behaves equivalently.
\end{itemize}
Both L42, and our formal language (see Section~\ref{s:formalism}) do not have traditional subclassing, rather all `classes' are either interfaces (which only have abstract methods), or are final (which cannot be subtyped). In a language with traditional subclassing, invariant methods would implicitly start with a check that \Q@super.invariant()@ returns \Q@true@. Note that invariant checks would not be performed at the end of \Q@super(..)@ constructor calls, but only at the end of \Q@new@ expressions, as happens in~\cite{feldman2006jose}.

\section{Formal Language Model}
\label{s:formalism}

%----------------------------------
In order to model our system, we need to formalise an imperative object oriented language
with exceptions, object capabilities, and rich type system
support for TMs and strong exception safety.
Formal models of the runtime semantics of such languages are simple, but 
defining and proving, such a type system would require a paper
of its own, and indeed many such papers exist in literature%
~\cite{ServettoEtAl13a,ServettoZucca15,GordonEtAl12,clebsch2015deny,JOT:issue_2011_01/article1}.
Thus we are going to assume that we already have an expressive and sound type system enforcing the properties we need, and instead focus on invariant checking.
We clearly list in Appendix \ref{s:proof} the assumptions we make on such a type system, so that any language satisfying them, such as L42, can soundly support our invariant protocol.

To keep our small step semantics as conventional as possible, we follow Pierce~\cite{pierce2002types} and Featherweight Java~\cite{IgarashiEtAl01}, and assume:
\begin{itemize}
	\item An implicit program/class table.
	\item Memory, $\sigma : l\rightarrow C\{\Many{v}\}${, is} a finite map from locations, $l$, to annotated tuples, $C\{\Many{v}\}$, representing objects; where $C$ is the class name and $\Many{v}$ are the field values.
	We use the notation $\sigma[l.f=v]$ to update a field of an object, and $\sigma[l.f]$ to access one.
	\item A main expression that is reduced in the context of such a memory and program.
	\item A typing relation, $\Sigma;\Gamma\vdash\e:T$, where 
	the expression $\e$ can contain locations and free variables. The types of locations are encoded in 
a memory environment, 
$\Sigma : l\rightarrow C$,
	while the types of free variables are encoded in
a variable environment, $\Gamma : x\rightarrow T$.
	\item We use $\Sigma^\sigma$ to trivially extract the corresponding $\Sigma$ from a $\sigma$.
\end{itemize}
To encode object capabilities and I/O, we assume a special location  $c$ of class \Q@Cap@. This location would refer to an object whose fields model things like the content of files. In order to simplify our proof, we assume that:
\begin{itemize}
	\item instances of \Q@Cap@ cannot be created with a \Q@new@ expression,
	\item all methods in the \Q@Cap@ class must require a \Q@mut@ receiver, and will mutate its ROG,
	\item \Q@Cap@ can only have \Q@mut@ fields, and
	\item \Q@Cap@'s \Q@invariant@ method is defined to return \Q@true@.
\end{itemize}
For simplicity, we do not formalise actual exception objects, rather we have \emph{error}s, which correspond to expressions which are currently  `throwing' an exception; 
in this way there is no value associated with the \emph{error}.
Our L42 implementation instead models exceptions as throwing an \Q@imm@ value, formalising exceptions in this way would not cause any interesting variation of our proof.

\newcommand{\ctxG}{\myCalBig{G}}
\renewcommand{\vs}{\Many{v}}
\renewcommand{\Opt}[1]{#1?}
\begin{figure}
	\!\!\!\!
	\begin{grammatica}
		\produzione{\e}{\x\mid l\mid\Kw{true}\mid\Kw{false}\mid \e\singleDot\m\oR\es\cR\mid \e\singleDot\f 
			\mid\e\singleDot\f\equals\e 
			\mid\Kw{new}\ C\oR\es\cR
			\mid\Kw{try}\ \oC\e_1\cC\ \Kw{catch}\ \oC\e_2\cC
		}{expression}\\
		\seguitoProduzione{
			\mid \M{l}{\e_1}{\e_2}\mid\Kw{try}^{\sigma}\oC\e_1\cC\ \Kw{catch}\ \oC\e_2\cC
		}{runtime expr.}\\
		\produzione{v}{l}{value}\\
		\produzione{\ctx_v}{\square
			\mid \ctx_v\singleDot m\oR\es\cR
			\mid v\singleDot\m\oR\Many{v}_1,\ctx_v,\es_2\cR
			%\mid \ctx_v\singleDot\f 
			%\mid \ctx_v\singleDot\f\equals\e
			\mid v\singleDot\f\equals\ctx_v
		}{evaluation context}\\
		\seguitoProduzione{
			\mid \Kw{new}\ C\oR\Many{v}_1,\ctx_v,\es_2\cR
			\mid \M{l}{\ctx_v}{\e}
			\mid \M{l}{v}{\ctx_v}
			\mid \Kw{try}^\sigma\oC\ctx_v\cC\ \Kw{catch}\ \oC\e\cC}{}\\
		
		\produzione{\ctx}{\square\mid\ctx\singleDot m\oR\es\cR\mid\e\singleDot\m\oR\es_1,\ctx,\es_2\cR
			%\mid \ctx\singleDot\f 
			%\mid \ctx\singleDot\f\equals\e
			\mid \e\singleDot\f\equals\ctx
			\mid \Kw{new}\ C\oR\es_1,\ctx,\es_2\cR
		}{full context}\\
		\seguitoProduzione{
			\mid
			\M{l}{\ctx}{\e}\mid
			\M{l}{\e}{\ctx}\mid
			\Kw{try}^{\sigma?}\oC\ctx\cC\ \Kw{catch}\ \oC\e\cC\mid
			\Kw{try}^{\sigma?}\oC\e\cC\ \Kw{catch}\ \oC\ctx\cC
			
		}{}\\

		%\produzione{M_l}{\ctx[M\oR l,\e\cR]}{}\\
		%\produzione{\ctxG_l}{
		%  M_l\singleDot\m\oR\es_1,\ctx,\es_2\cR
		% |\e\singleDot\m\oR\es_1, M_l, \es_2, \ctx, \es_3\cR
		% |M_l\singleDot\f\equals\ctx
		% |\Kw{new}\ C\oR\es_1,M_l,\es_2,\ctx,\es_3\cR
		% |\Kw{try}\oC\ctx\cC\ \Kw{catch}\ \oC\e\cC
		% |\ctx[\ctxG_l]}{}\\
		\produzione{\mathit{CD}}{\Kw{class}\ C\ \Kw{implements}\ \Many{C}\oC\Many{F}\,\Many{M}\cC\mid 
			\Kw{interface}\ C\ \Kw{implements}\ \Many{C}\oC\Many{M}\cC
		}{class declaration}\\
		\produzione{F}{\T\ \f\semiColon}{field}\\
		\produzione{M}{\mdf\, \Kw{method}\, \T\ \m\oR\T_1\,\x_1,\ldots,\T_n\,\x_n\cR\ \Opt\e}{method}\\
		\produzione{\mdf}{\Kw{mut}\mid\Kw{imm}\mid\Kw{capsule}\mid\Kw{read}}{type modifier}\\
		\produzione{\T}{\mdf\,C}{type}\\
		\produzione{r_l}{
			v\singleDot\m\oR\Many{v}\cR
			\mid v\singleDot\f
			\mid v_1\singleDot\f\equals v_2
			\mid \Kw{new}\,C\oR\Many{v}\cR
			,\quad\text{where }l\in \{v,v_1,v_2,\Many{v}\}
		}{redex containing $l$}\\
		\produzione{\mathit{error}}{
			\ctx_v[\M{l}{v}{\Kw{false}}]
			,\quad\text{where }
			\ctx_v \text{ not of form}\ \ctx_v'[\Kw{try}^{\sigma?}\oC\ctx_v''\cC\ \Kw{catch}\ \oC\_\cC]
		}{validation error}
	\end{grammatica}
	\caption{Grammar}\label{f:grammar}
\end{figure}

\subheading{Grammar}
The detailed grammar is defined in Figure \ref{f:grammar}. 
Most of our expressions are standard.
\emph{Monitor expressions}
 are of the form \M{l}{\e_1}{\e_2}, they 
are run time expressions and thus are not present in method bodies, rather they are generated by our reduction rules inside the main expression. Here, $l$ refers to the object being monitored, $e_1$ is the expression which is being monitored, and $e_2$ denotes the evaluation of $l.\invariant$. If, at any point in execution, $\e_2$ is \Q!false!, then $l$'s invariant failed to hold; such a monitor expression corresponds to the throwing of an unchecked exception.

In addition, our reduction rules will annotate \Q@try@ expressions with
the original state of memory. This is used to model the guarantee of strong exception safety, that is, the annotated memory will not be mutated by executing the body of the \Q@try@.

\subheading{Well Formedness Criteria}
We additionally restrict the grammar with the following well formedness criteria:
\begin{itemize}
	\item \Q@invariant@ methods and capsule mutators satisfy the restrictions in Section \ref{s:protocol}.
	\item Field accesses and updates in methods are of the form $\Kw{this}.f$ or $\Kw{this}.f\equals\e$, respectively.
	\item Field accesses and updates in the main expression are of the form $l.f$ or $l.f\equals\e$, respectively.
	\item Locations that are preserved by \Q@try@ blocks are
	never monitored, that is, for $\Kw{try}^\sigma\oC\e\cC\ \_$, if $\e$ is of the form $\ctx[\M{l}{\_}{\_}]$, then $l\notin\sigma$.
\end{itemize}
\newcommand{\rowSpace}{\\\vspace{2.5ex}}
\begin{figure}
	\!\!
	$\!\!\!\!\!\begin{array}{l}
	\inferrule[(update)]{{}_{}}{
		\sigma|l.f\equals{}v\rightarrow \sigma[l.f=v]|
		\M{l}{l}{l\singleDot\invariant}
	}{}
	\quad
	\inferrule[(new)]{{}_{}}{
		\sigma|\Kw{new}\ C\oR\vs\cR\rightarrow \sigma,l\mapsto C\{\vs\}|
		\M{l}{l}{l\singleDot\invariant}
	}{}
	\\
	\rowSpace
	\inferrule[(mcall)]{{}_{}}{
		\sigma|l\singleDot\m\oR v_1,\ldots,v_n\cR\rightarrow \sigma|
		\e'[\Kw{this}=l,\x_1=v_1,\ldots,x_n=v_n]
	}{
		\begin{array}{l}
		\sigma(l)=C\{\_\}\\
		C.m=\mdf\,\Kw{method}\,\T\,\m\oR\T_1\,\x_1\ldots\T_n\x_n\cR\,\e\\
		
		\text{if }\ \exists \f\text{ such that } C.f=\Kw{capsule}\,\_,
		\mdf=\Kw{mut},
		\\*\quad\f\, \text{inside}\, C\singleDot\m
		\text{, and }
		\f\,\text{inside}\, C\singleDot\Kw{invariant}
		
		\\*
		\text{then }\e'=\M{l}{e}{l\singleDot\invariant}\\*
		\text{otherwise }\e'= \e
	\end{array}
}
\rowSpace
\inferrule[(monitor exit)]{{}_{}}{
	\sigma|\M{l}{v}{\Kw{true}}\rightarrow \sigma|v
}{}
\quad

\inferrule[(ctxv)]{\sigma_0|\e_0\rightarrow\sigma_1|\e_1}{
	\sigma_0|\ctx_v[\e_0]\rightarrow \sigma_1|\ctx_v[\e_1]
}{}

\quad
\inferrule[(try enter)]{{}_{}}{
	\sigma|\Kw{try}\ \oC \e_1\cC\ \Kw{catch}\ \oC\e_2\cC\rightarrow 
	\sigma|\Kw{try}^\sigma\oC\e_1\cC\ \Kw{catch}\ \oC\e_2\cC
}{}
\quad

\rowSpace

\inferrule[(try ok)]{{}_{}}{
	\sigma,\sigma'|\Kw{try}^{\sigma}\oC v\cC\ \Kw{catch}\ \oC\_\cC\rightarrow \sigma,\sigma'|v
}{}
\quad

\inferrule[(try error)]{{}_{}}{
	\sigma,\_|\Kw{try}^\sigma\oC \mathit{error}\cC\ \Kw{catch}\ \oC\e\cC\rightarrow \sigma|\e
}
\quad
\inferrule[(access)]{{}_{}}{
	\sigma|l.f\rightarrow \sigma|\sigma[l.f]
}{}
%\quad
\end{array}$
\caption{Reduction rules}\label{f:reductions}
\end{figure}

\subheading{Reduction rules}
Our reduction rules are defined in Figure \ref{f:reductions}.
They are pretty standard, except for our handling of monitor expressions.
We define the relation \emph{inside} as follows:\par
$%\begin{array}{l}
\f\, \textit{inside}\, C\singleDot\m\text{ iff }
C\singleDot\m=\_\,\Kw{method}\_\,\ctx[\Kw{this}\singleDot\f]
%\end{array}
$

%\noindent Inserting the monitor expressions during reduction is convenient for the proof,
%but it could instead be done ahead of time.

\noindent Monitor expressions are added after all field updates, \Q@new@ expressions, and calls to capsule mutators.
%Our formalism of monitor expressions are only a proof device, they need not be part of the language itself, for example L42 implements our invariant protocol by generating wrapper functions over primitive setters and factory methods.
%Monitor expressions are only a proof device, and an execution on a real hardware 
Monitor expressions are only a proof device, they need not be implemented directly as presented.
For example, in L42 we implement them by statically injecting calls to \Q!invariant! at the end of setters, factory methods and capsule mutators; this works as L42 does not have primitive expression forms for field updates and constructors, rather they are uniformly represented as method calls.
% do not need to represent them.  In L42 field updates are always performed throughout a setters, thus we can just inject calls to \Q@invariant@ on setters, at the end of constructor bodies and at the end of  capsule mutators.

Our \textsc{ctxv} rule evaluates monitor expressions, \M{l}{\e_1}{\e_2}, by first evaluating $\e_1$ and then $\e_2$. If $\e_2$ evaluates to \Q@true@, then the monitor succeeded, and will yield the result of $\e_1$. If however $\e_2$ evaluated to \Q!false!, then the monitor failure will be caught by our \textsc{try error} rule, as will any other uncaught monitor failure in $e_1$ or $e_2$.

\subheading{Statement of Soundness}
We define a deterministic reduction to mean that exactly one reduction is possible:\\*
\indent$\ \sigma_0|e_0\Rightarrow \sigma_1|e_1$ iff $\{\sigma_1|\e_1\}=\{\sigma|\e \text{, where } \sigma_0|e_0\rightarrow \sigma|e\}$

\noindent An object is \emph{valid} iff calling its \Q@invariant@ method would
deterministically produce \Q@true@ in a finite number of steps, i.e. it does not evaluate to \Q@false@, fail to terminate, or produce an \emph{error}.
We also require evaluating \Q@invariant@ to preserve existing memory ($\sigma$), however new objects ($\sigma'$) can be created and freely mutated.

\indent$\mathit{valid}(\sigma,l)$ iff $\sigma | l.\invariant {\Rightarrow^+} \sigma,\sigma' | \Kw{true}$.%\loseSpace

\noindent 
To allow the invariant method to be called on an invalid object, and access fields on such object, we define the set of trusted execution steps as the the call to \Q@invariant@ itself, and any field accesses inside its evaluation. Note that this only applies to single small step reductions, and not the entire evaluation of \Q!invariant!.

%\loseSpace
\noindent $\mathit{trusted}(\ctx_v,r_l)$ iff\\*
\indent either
$r_l=l.\invariant$ and
$\ctx_v=\ctx_v'[\M{l}{v}{\square}]$,\\*
\indent or
$r_l=l$\Q@.f@ and
$\ctx_v=\ctx_v'[\M{l}{v}{\ctx_v''}]$.
%\loseSpace

\noindent Finally, we define what it means to soundly enforce our invariant protocol: every object referenced by any untrusted redex is valid.

\begin{theorem}[Soundness]\rm
if $c:\Kw{Cap};\emptyset\vdash \e: \T$ and
$c\mapsto\Kw{Cap}\{\_\}|\e\rightarrow^+ \sigma|\ctx_v[r_l]$, then
either $\mathit{valid}(\sigma,l)$ or $\mathit{trusted}(\ctx_v,r_l)$.
\end{theorem}

%We believe this property captures very precisely our statements in section~\ref{s:protocol}.
\section{Invariants Over Immutable State}
\label{s:immutable}
In this section we consider validation over fields of \Q@imm@ types.
%\footnote{
%In a real language, for conciseness one could make the \Q@imm@ modifier the default, allowing it to be omitted and our \Q@Person@ example class would only use 3 type modifiers; however we explicitly use it here for clarity.
%}
In the next section we detail our technique for \Q@capsule@ fields.

In the following code \Q@Person@ has a single immutable (non final) field \Q@name@:
\begin{lstlisting}
class Person {
  read method Bool invariant() { return !name.isEmpty(); }
  private String name;//the default modifier imm is applied here
  read method String name() { return this.name; }
  mut method String name(String name) { this.name = name; }
  Person(String name) { this.name = name; }
}
\end{lstlisting}
\Q@Person@ only has immutable fields and its constructor only uses \Q@this@ to initialise them.
%; we say such a class is \emph{simple}.
%\Q@Person@, only has immutable fields and the constructor 
%uses the parameters to directly initialize (all) the fields.
% We say such a class is \emph{simple}.%
%\footnote{
%We consider only standard contractors for simplicity of exposition.
%More complex constructors could be supported, provided that \Q@this@ is only used to access fields, we do discuss them for simplicity.}
% The difference with respect to UML DataTypes 
%immutable types (like UML DataTypes)
%UML datatypes are aclass property. immutable types are often an instance one (so no final fields) 
Note that the \Q@name@ field is not final, thus \Q@Person@ objects can change state during their lifetime. This means that the ROGs of all \Q@Person@s fields are immutable, but \Q@Person@s themselves may be mutable.
%Of course UML DataTypes
%immutable types
%No, a type is not a class
% are just a special case of simple classes.
We can easily enforce \Q@Person@'s invariant by generating checks on the result of \Q@this.invariant()@: immediately after each field update, and at the end of the constructor.%
%\footnote{Since the constructor only initialises fields; as with the \Q@invariant@ method itself, we allow field uses since \Q@this@ is not directly reached.}
%would require the initial/default value of \Q@this@ to be valid.}

% If a simple class provides a \Q@invariant@ method, then validation will be enforced.
% For \Q@Person@, intuitively, the code would behave as follow:

%\Comment{if we made this public, all users who update the field need to call validate}%
%There are many interpretations for your comment
%why you deleted my code comments?
\begin{lstlisting}
class Person { .. // Same as before
  mut method String name(String name) {
    this.name = name; // check after field update
    if (!this.invariant()) { throw new Error(...); }}
  Person(String name) {
    this.name = name; // check at end of constructor
    if (!this.invariant()) { throw new Error(...); }}
}
\end{lstlisting}
%... $\MComment{validation error}$ 
%
% Many programmers attempted to write similar code in mainstream languages like Java to ensure  that some property always holds. Indeed, at first look, this code seems to correctly enforce validation. Sadly, without relying on TM and OC, the former code would be broken: just making the fields private and checking the \Q@invariant@ method at the \textbf{end of the constructor} and at the \textbf{end of mutator methods} is not enough to enforce validation.
% The trick is that our intuition relies not on statically verified properties, or on the semantics of the language, but on the expectations about `correct' behaviour of \Q@String@. We need to enforce Validation without assuming the behaviour of other objects.

Such checks will be generated/injected, and not directly written by the programmer. If we were to relax (as in Rust), or even eliminate (as in Java), the support for TMs or OCs, the enforcement of our invariant protocol for the \Q@Person@ class would become harder, or even impossible. 

\subheading{Unrestricted use of non determinism} Allowing the \Q@invariant@ method to (indirectly) perform a non deterministic operation, such as by creating new capability objects, could break our guarantee that (manually) calling it always returns \Q!true!.
%\Q@invariant@ to be non-deterministic.}%
% 
For example consider this simple and contrived (mis)use of person:
\begin{lstlisting}[morekeywords={assert}]
class EvilString extends String {
  @Override read method Bool isEmpty() {
    // Create a new capability object out of thin air
    return new Random().bool(); }
} ..
method mut Person createPersons(String name) {
  // we can not be sure that name is not an EvilString
  mut Person schrodinger = new Person(name); // exception here?
  assert schrodinger.invariant(); // will this fail?
  ..}
\end{lstlisting}
%//  mut Person schrodinger2 = new Person(name); // what about here?

Despite the code for \Q@Person.invariant@ intuitively looking correct and deterministic, the above call to it is not. Obviously this breaks any reasoning and would make our protocol unsound. 
In particular, note how in the presence of dynamic class loading, we have no way of knowing what the type of \Q@name@ could be. Since our system allows non determinism only through capability objects, and 
restricts their creation, the above example would be prevented.
 %Languages like Java, Rust and Pony, which do not require the user of object-capabilities to perform non-deterministic operations, suffer from . 
%???
%Even if we disallow subtyping the same problem could still occur if we had a strange implementing of \Q@String@, or \Q@Person.validate()@ itself.

\subheading{Allowing Internal Mutation Through Back Doors}
Suppose we relax our rules by allowing interior mutability
as in Rust and Javari, where sneaky mutation
of the ROG of an `immutable' object is allowed.
Those back doors are usually motivated by performance reasons, however in~\cite{GordonEtAl12} they
briefly discuss how a few trusted language primitives can be used to perform caching and other needed optimisations,
without the need for back doors.

Our example shows that such back doors can be used to break determinism of \Q@invariant@ methods, by allowing the invariant to store and read information about previous calls. In the following example we use \Q@MagicCounter@ as a back door to remotely break the invariant of \Q@person@ without any interaction with the \Q@person@ object itself:
\begin{lstlisting}
class MagicCounter {
  method Int increment(){
    //Magic mutation through an imm receiver, equivalent to i++
}}
class NastyS extends String {..
  MagicCounter evil = new MagicCounter(0);
  @Override read method Bool isEmpty() {
    return this.evil.increment() != 2; }
} ..
NastyS name = new NastyS("bob"); //TMs believe name's ROG is imm
Person person = new Person(name); // person is valid, counter=1
name.increment(); // counter == 2, person is now broken
person.invariant(); // returns false!, counter == 3
person.invariant(); // returns true, counter == 4
\end{lstlisting}

\subheading{Strong Exception Safety}
The ability to catch and recover from invariant failures is extremely useful as it allows programs to take corrective action.
Since we represent invariant failures by throwing unchecked exceptions, programs can recover from them with a conventional \Q@try@--\Q@catch@.
%\REVComm{
%	Due to the guarantees of strong exception safety, the only trace that an invalid object existed is the exception thrown; any object that has been mutated/created during the \Q@try@ block is now unreachable (as happens in alias burying~\cite{boyland2001alias}).
	Due to the guarantees of strong exception safety, any object that has been mutated during a \Q@try@ block is now unreachable (as happens in alias burying~\cite{boyland2001alias}). In addition, since unchecked exceptions are immutable, they can not contain a \Q@read@ reference to any object (such as the \Q@this@ reference seen by \Q@invariant@ methods). These two properties ensure that an object whose invariant fails will be unreachable after the invariant failure has been captured. %in a \Q@catch@.	
%}{3}{\label{SES2} [see footnote \ref{SES1}]}.
If instead we were to not enforce strong exception safety, an invalid object could be made reachable:
%\saveSpace
\begin{lstlisting}[morekeywords={assert}, escapechar=\%]
mut Person bob = new Person("bob");
// Catch and ignore invariant failure:
try { bob.name(""); } catch (Error t) { } // ill typed in L42
assert bob.invariant(); // bob is invalid!
\end{lstlisting}

As you can see, recovering from an invariant failure in this way is unsound and would break our protocol.

\section{Invariants over encapsulated state}
\label{s:encapsulated}
Consider managing the shipment of items, where there is a maximum combined weight:
\begin{lstlisting}
class ShippingList {
  capsule Items items;
  read method Bool invariant() {
    return this.items.weight() <= 300; }
  ShippingList(capsule Items items) {
    this.items = items;
    if (!this.invariant()) {throw Error(...);}}// injected check
  mut method Void addItem(Item item) {
    this.items.add(item);
    if (!this.invariant()) {throw Error(...);}}// injected check
}
\end{lstlisting}

To handle this class we just inject calls to \Q@invariant@ at the end of the constructor and the \Q@addItem@ method.
This is safe since the \Q@items@ field is declared \Q@capsule@.
Relaxing our system to allow a \Q@mut@ modifier for
the \Q@items@ field and the corresponding constructor parameter 
breaks the code:
the cargo we received in the constructor may already be compromised:
%\saveSpace
\begin{lstlisting}
mut Items items = ...;
mut ShippingList l = new ShippingList(items); // l is valid
items.addItem(new HeavyItem()); // l is now invalid!
\end{lstlisting}

As you can see it would be possible for external code with no knowledge of the \Q@ShippingList@ to mutate its items.%
\footnote{%
Conventional ownership solves these problems by requiring a deep clone of all the data the constructor takes as input, as well as all exposed data (possibly through getters).
In order to write correct library code in mainstream languages like Java and C++, defensive cloning~\cite{Bloch08} is needed.
%\REVComm{
For performance reasons, this is hardly done in practice and is a continuous source of bugs and unexpected behaviour~\cite{Bloch08}.}
%}{2}{citation to support this?}

Our restrictions on capsule mutators ensure that capsule fields are essentially an exclusive mutable reference.
Removing these restrictions would break our invariant protocol.
If we were to allow \Q@x.items@ to be seen as \Q@mut@, where \Q@x@ is not \Q@this@, then  even if the \Q@ShippingList@ has full control at initialisation time, such control may be lost later, and code unaware of the \Q@ShippingList@ could break it:
%\saveSpace
\begin{lstlisting}
mut ShippingList l = new ShippingList(new Items()); // l is ok
mut Items evilAlias = l.items // here l loses control
evilAlias.addItem(new HeavyItem()); // now l is invalid!
\end{lstlisting}

If we allowed a \Q@mut@ return type the following would be accepted:
\begin{lstlisting}
mut method mut Items expose(C c) {return c.foo(this.items);}
\end{lstlisting}

Depending on dynamic dispatch, \Q@c.foo()@ may just be the identity function, thus
we would get in the same situation as the former example.
%Static analysis is usually unable/unwilling to track precise behaviour of dynamic dispatch.

%In addition to the above we put restrictions on any \Q@mut@ and \Q@capsule@ methods using a \Q@capsule@ field (we call such methods `capsule mutators'):
%\begin{itemize}
%\item only a single use of \Q@this@ is allowed (and is the one that uses the field),
%\item no \Q@mut@ or \Q@read@ parameters are allowed (apart from the implicit \Q@this@ parameter)
%\item and the return type cannot be \Q@mut@.
%\end{itemize}
%\noindent  Moreover, if the used capsule field is referenced in \validate, a \Q@this.validate()@ call is injected at the end of the method.

Allowing \Q@this@ to be used more than once can also cause problems:
\begin{lstlisting}
mut method imm Void multiThis(C c) {
  read Foo f = c.foo(this);
  this.items.add(new HeavyItem());
  f.hi(); } // Can `this' be observed here?
\end{lstlisting}

If the former code were accepted, 
\Q@this@ may be reachable from \Q@f@, thus \Q@f.hi()@ may observe an invalid object.

In order to ensure that a second reference to \Q@this@ is not reachable through the parameters, we only accept \Q@imm@ and \Q@capsule@ parameters.
If we were however to accept a \Q@read@ parameter, as in the example below,
we would be in the same situation as before, where \Q@f@ may contain
a reference to \Q@this@:
\begin{lstlisting}
mut method imm Void addHeavy(read Foo f) {
  this.items.add(new HeavyItem())
  f.hi() } // Can 'this' be observed here?
...
mut ShippingList l = new ShippingList();
read Foo f = new Foo(l);
l.addHeavy(f); // We pass another reference to `l' through f
\end{lstlisting}

\section{GUI Case study}
\label{s:case-study}

%interface Foo{ma mb}
%
%class Raw implements Foo{
%  no validate 
%}
%class ValidFoo1 implemens Foo{
%  private capsule Foo inner;
%  ma(){
%    this.inner.ma();
%  }
%  validate
%}
%
%
%class Raw {
%  mut method mut A stuff(read A a) {
%      if a.bar() {
%         x = ...
%      }
%      return new A(x)
%  }
%
%  read method imm Object stuffPre(read A a) {
%     return a.bar()
%  }
%  read method mut A stuffPost(imm Object o) {
%     return new A(x)
%  }
%  mut method imm Object stuffInner(imm Object) {
%      x = ...
%  }
%}
%class ValidFoo1 {
%	cpasule Raw r;
%  mut method mut A stuff(read A a) {
%     imm Object p = this.r.stuffPre(a);
%     imm Object res = this.transformR(r -> r.stuffInner(p));
%      return this.r.stuffPost(res)
%  }
%}
%
%
%
%class ValidFoo2 implemens Foo{
%  private capsule Foo inner;
%  validate
%}
Here we show that we are able to verify classes with circular mutable object graphs, that interact with the real world using I/O.
%Here we discuss how to use conventional OO programming patterns to take advantage of our system. Using those
%patterns allows to circumvent some of the apparent limitations of our system.
% At first glance it would seem that by only being able to validate immutable and encapsulated state one could not create validated classes with complicated, mutable interconnected object graphs.
%We show that this is not the case by encoding
%Our invariant is that everything in every movable container (and the top level component) should 
%-be inside the container 
%-not overlap with anyhing else inside
%Our containers represent boxes that completley contain other non-overlapping boxes.
Our case study involves a GUI with containers (\Q@SafeMovable@s) and \Q@Button@s;
the \Q@SafeMovable@ class has an invariant to ensure that its children are completely contained within it and do not overlap. The \Q@Button@s move their \Q@SafeMovable@ when pressed. We have a \Q@Widget@ interface which provides methods to get \Q!Widget!s' size and position as well as children (a list of \Q@Widget@s). Both \Q@SafeMovable@s and \Q@Button@s implement \Q@Widget@. Crucially, since the children of \Q@SafeMovable@ is a list of \Q@Widget@s it can contain other \Q@SafeMovable@s, and all queries to their size and position are dynamically dispatched, such queries are also used in \Q@SafeMovable@'s invariant.
Here we show a simplified version\footnote{The full version, written in L42, which uses a different syntax, is available in our artifact at\\ \url{http://l42.is/EcoopArtifact.zip}}, where  \Q@SafeMovable@ has just one \Q@Button@, and certain sizes and positions are fixed. Note that \Q@Widgets@ is a class representing a mutable list of \Q@mut@ \Q@Widget@s.
\begin{lstlisting}[mathescape=false]
class SafeMovable implements Widget { capsule Box box;
  @Override read method Int left() { return this.box.l; }
  @Override read method Int top() { return this.box.t; }
  @Override read method Int width() { return 300; }
  @Override read method Int height() { return 300; }
  @Override read method read Widgets children() {
    return this.box.c; }
  @Override  mut method Void dispatch(Event e) {
    for (Widget w:this.box.c) { w.dispatch(e); }}
  read method Bool invariant() {..}
  SafeMovable(capsule Widgets cs) { this.box = makeBox(c); }
  static method capsule Box makeBox(capsule Widgets c) {
    mut Box b = new Box(5, 5, cs);
    b.c.add(new Button(0, 0, 10, 10, new MoveAction(b));
    return b; } //mut b is soundly promoted to capsule
}
class Box { Int l; Int t; mut Widgets c;
  Box(Int l, Int t, mut Widgets c) {..}
}
class MoveAction implements Action { mut Box outer;
  MoveAction(mut Box outer) { this.outer = outer; }
  mut method Void process(Event event) { this.outer.l += 1; }
}
..
//main expression; #$ is a capability method making a Gui object
Gui.#$().display(new SafeMovable(..));
\end{lstlisting}

As you can see, \Q@Box@es encapsulate the state of the \Q@SafeMovable@s that can change over time:
\Q@left@, \Q@top@, and \Q@children@. Also note how the ROG of \Q@Box@ is circular: since
the \Q@MoveAction@s inside \Q@Button@s need a reference to the containing \Q@Box@ in order to move it.
Even though the children of \Q@SafeMovable@s are fully encapsulated, we can still easily dispatch events to them using \Q@dispatch@. Once a \Q@Button@ receives an \Q@Event@ with a matching ID, it will call its \Q@Action@'s \Q@process@ method. 

%Our main function uses a capability-object to display the top-level \Q@Widget@ and its \Q@children@, as well as dispatch events to it. 
Our example shows that the restrictions of TMs and OCs are flexible enough to encode interactive GUI programs, where widgets may circularly reference other widgets.
In order to perform this case study we had to first implement a simple GUI Library in L42. This library uses object capabilities to draw the widgets on screen, as well as fetch and dispatch the events. Importantly, neither our application, nor the underlying GUI library require back doors into either our type modifier or capability system to function, demonstrating the practical usability of our restrictions.

\subheading{The Invariant}
\Q@SafeMovable@ is the only class in our GUI that has an invariant, our system automatically checks it in two places: the end of its constructor and the end of its \Q@dispatch@ method (is a capsule mutator). There are no other checks inserted since we never do a field update on a \Q@SafeMovable@. The code for the invariant is just a couple of simple nested loops:
\begin{lstlisting}
read method Bool invariant() {
  for(Widget w1 : this.box.c) {
    if(!this.inside(w1)) { return false; }
    for(Widget w2 : this.box.c) {
      if(w1!=w2 && SafeMovable.overlap(w1, w2)){return false;}}}
  return true;}
\end{lstlisting}

Here \Q@SafeMovable.overlap@ is a static method that simply checks that the bounds of the widgets don't overlap. The call to \Q@this.inside(w1)@ similarly checks that the widget is not outside the bounds of \Q@this@; this instance method call is allowed as \Q@inside@ only uses \Q@this@ to access its fields.%its \Q@width@ and \Q@height@ fields.

% This code is a simplified version of our first case study. In the full code the \Q@SafeMovable@ constructors take \Q@left@, \Q@top@, \Q@width@ and \Q@height@ parameters and can either take a \Q@box@ directly or will generate one with $4$ \Q@Button@s}
% we we have $4$ buttons, each button moves in one of the $4$ cardinal directions.
\newlength{\imglength}%
\setlength{\imglength}{0.4\textwidth}%
\columnsep=0.5em%
\begin{wrapfigure}{l}{\imglength}%
	\setbox0=\hbox{\strut}%
	\vspace{-1.3\ht0}%
    \includegraphics[width=\imglength]{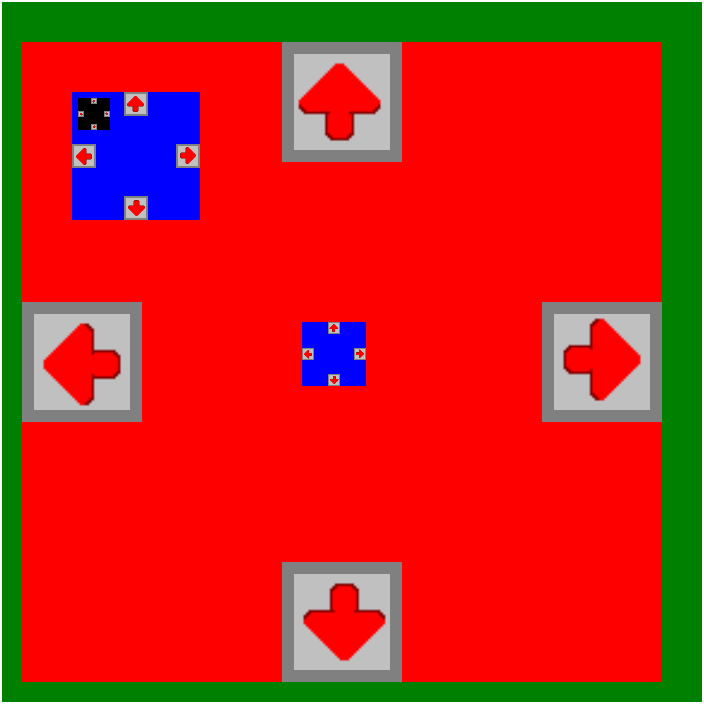}%
	\vspace{-1.5\ht0}%
\end{wrapfigure}
\subheading{Our Experiment}
As shown in the figure to the left, counting both \Q@SafeMovable@s and \Q@Button@s, our main method creates $21$ widgets: a top level (green) \Q@SafeMovable@ without buttons, containing $4$ (red, blue, and black) \Q@SafeMovable@s with
$4$ (gray) buttons each. When a button is pressed it moves the containing \Q@SafeMovable@ a small amount in the corresponding direction.
%Each container has 4 gray buttons one for each cardinal direction.
% In our set up
% the top level \Q@SafeMovable@
% contains a big red \Q@SafeMovable@
%  containing $2$ smaller blue \Q@SafeMovable@. One of those contains a tiny black \Q@SafeMovable@.
%Our invariant is that everything in every movable container (and the top level component) should 
%-be inside the container 
%-not overlap with anyhing else inside
%Our containers represent boxes that completley contain other non-overlapping boxes.
This set up is not overly complicated, the maximum nesting level of \Q@Widget@s is $5$.
Our main method automatically presses each of the $16$ buttons once. In L42, using the approach of this paper, this resulted in $77$ calls to \Q@SafeMovable@'s invariant.

\subheading{Comparison With Visible State Semantics}
As an experiment, we set our implementation to generate invariant checks following the  visible state semantics approaches of D and Eiffel~\cite{Alexandrescu:2010:DPL:1875434,DRef},
where the invariant of the receiver is instead checked at the start and end of \emph{every} 
public (in D) and qualified\footnote{That is, the receiver is not \Q!this!.} (in Eiffel) method calls.
In our \Q!SafeMovable! class, all methods are public, and all calls are qualified, thus this difference is irrelevant. Neither protocol performs invariant checks on field accesses or updates,
however due to the `uniform access principle'%\footnote{%
%L42 also follows the Eiffel uniform access principle: field accesses are the same as method calls.%
%}
, Eiffel allows fields to directly implement methods, allowing the \Q!width! and \Q!height! \emph{fields} to directly implement \Q!Widget!s \Q!width! and \Q!height! \emph{methods}. On the other hand in D, one would have to write getter \emph{methods}, which would invoke invariant checks.
%though 42 represents field accesses as method calls, for a fair comparison with conventional OO approach, we do not treat field accesses on \Q@SafeMovable@s within the \Q@SafeMovable@ class itself as public method calls
% D and Eiffel have slightly different interpretation of
% visible state semantic when qualified/unqualified method calls or field accesses are performed.
% However, in our GUI example those corner cases are
% not relevant. To be sure, we implemented both of their strategies and obtained the same results.
When we ran our test case following the D approach, the \Q!invariant! method was called $52,734,053$ times, whereas the Eiffel approach `only' called it $14,816,207$ times;%\footnote{%
%We expect all (sound) runtime approaches based on visible state semantics ~\cite{feldman2006jose,fahndrich2010embedded,abercrombie2002jcontractor,tran2003design} to produce similar results.}
in comparison our invariant protocol only performed $77$ calls. The number of checks is exponential in the depth of the GUI: the invariant of a \Q@SafeMovable@ will call the \Q@width@, \Q@height@, \Q@left@, and \Q@top@ methods of its children, which may themselves be \Q@SafeMovable@s, and hence such calls may invoke further invariant checks. Note that \Q!width! and \Q!height! are simply getters for fields, whereas the other two are non trivial \emph{methods}. % Since we only use \Q@this@ to perform method calls, the invariant is not recursively checked on \Q@this@, if it were we would get an infinite recursion.

\subheading{Spec\# Comparison}
We also encoded our example in Spec\#\footnote{We compiled Spec\# using the latest available source (from 19/9/2014). The verifier available online at \url{rise4fun.com/SpecSharp} behalves differently.}, which like L42, statically verifies aliasing/ownership properties, as well as the admissibility of invariants.
% To keep the Spec\# code aligned with the L42 one, 
% we do not use a .NET GUI library, but we just simulate the behaviour of L42, without actually opening a window.
The backend of the L42 GUI library is written in Java, we did not port it to Spec\#, rather we just simulate the backend, and don't actually display a GUI in Spec\#.

We ran our code through the Spec\# verifier (powered by Boogie~\cite{DBLP:conf/fmco/BarnettCDJL05}), which only gave us $2$ warnings\footnote{We used \Q@assume@ statements, equivalent to Java's \Q@assert@, to dynamically check array bounds. %. and value presence.
This aligns the code with L42, which also performs such checks at runtime.}: that the invariant of \Q@SafeMovable@ was not known to hold at the end of its constructor and \Q@dispatch@ method. Like our system however, Spec\# checks the invariant
at those two points at runtime. Thus the code is equivalently verified in both Spec\# and L42; in particular it performed exactly the same number ($77$) of runtime invariant checks.\footnote{%
We also encoded our GUI in Microsoft Code Contracts~\cite{DBLP:conf/sac/FahndrichBL10}, whose unsound heuristic also calls the invariant $77$ times; however Code Contract does not enforce the
encapsulation of \Q@children@, thus their approach would not be sound in our context.}

% Assuming that Spec\#'s verifier is sound, this means that our code is equally verified in Spec\# and L42, providing a reasonable comparison.
% Both Spec\# and L42 did the same thing:
% they statically verify ownership/aliasing annotations,
% they check the admissibility/valididty of a the invariant code and finally 
% they perform sufficient runtime-checking of the invariant.
% We used the Boogie static checker to verify all the aliasing/ownership properties needed to
% ensure that the $77$ run-time invariant checks soundly enforce that the invariant holds when is expected.
% This of course includes preventing the Gui to ever display two overlapping Widget. 

%\end{itemize}
%The result is the same of L42: the invariant is checked $77$ times, and in exactly the same locations of L42.

We found it quite difficult to encode the GUI in Spec\#, due to its unintuitive and rigid ownership discipline. In particular we needed to use many more annotations, which were larger and had greater variety. In the following table we summarise the annotation burden,
for the \emph{program} that defines and displays the \Q@SafeMovable@s and our GUI; as well as the \emph{library} which defines \Q@Button@s, \Q@Widget@, and event handling.\footnote{We only count constructs Spec\# adds over C\# as annotations, we also do not count annotations related to array bounds or null checks.}:
\begin{center}\saveSpace\saveSpace
\begin{tabular}{ c  c  c  c  c}
 & Spec\# & Spec\# & L42 & L42 \\ 
 & \!\!program\!\! & library & \!\!program\!\! & library \\
\hline
 
\!\!\!Total number of annotations 
 	& $40$ & $19$ & $19$ & $18$ \\ \hline
% Totals 	 $59$ $37
\!\!\!Tokens (except \Q@.,;(){}[]@ and whitespace)\!\!\!
%(ecluding \Q@;,@ characters, white-space, or parentheses/brackets.) 
	& $106$ & $34$ & $18$ & $18$  \\  \hline
% $140$  & $37$
Characters (with minimal whitespace) 
	& $619$ & $207$ & $74$ & $60$ \\ \hline
%  $826$ $134$
\end{tabular}
\end{center}
\lstset{morekeywords={requires}}

To encode the GUI example in L42, the only annotations we needed were the 3 type modifiers: \Q@mut@, \Q@read@, and \Q@capsule@.
% , for a total of 
% $19$ annotations (one token each, $74$ characters in total).
Our Spec\# code requires things such as, purity, immutability, ownership, method pre/post conditions and method modification annotations. In addition, it requires the use of 4 different ownership functions including explicit ownership assignments. In total we used 18 different kinds of annotations in Spec\#.
Together these annotations can get quite long, such as the following precondition on \Q@SafeMovable@'s constructor: \\
\indent\Q@requires Owner.Same(Owner.ElementProxy(children), children);@

% methods and field attributes as well as requires, ensures and modifies clauses, and finally also %explicit ownership assignment statements.
% Spec\# annotation can be involved, as for example \Q@requires Owner.Same(Owner.ElementProxy(children), children);@}

% To estimate the annotation burden we count the number of tokens (excluding \Q@.;,@ and parenthesis).
% This gave us $113$ tokens, that is more then $5$ times the amount needed in L42.
% The total annotation character count is $830$; $10$ times more then L42.

% 40   106+17=113 619+211
%main 11 annotations, 28 + 11 tokens, 118+65 characters
%safemovable 29 annotations, 78+6 tokens, 501+146 characters

%19   34   207+44
% auxLib // 14 annotations, 25 tokens, 155+44 characters
%guiLib// 5 annoations, 9 tokens, 52 characters

% Moreover, in L42 we only use $3$ different kinds of annotations, while in Spec\# we use $15$ kinds of annotations.

\noindent The Spec\# code also required us to deviate from the style of code we showed in our simplified version: we could not write a usable \Q@children@ method in \Q@Widget@ that returns a list of children, instead we had to write \Q@children_count()@ and \Q@children(int i)@ methods; we also needed to create a trivial class with a \Q@[Pure]@ constructor (since \Q@Object@'s one is not marked as such). In contrast, the only strange thing we had to in L42 was creating \Q@Box@es by using 
an additional variable in a nested scope.
This is needed to delineate scopes for promotions.
Based on these results, we believe our system is significantly simpler and easier to use.
% On the basis of these results we believe  that our system is easier to use for programmers that are not experts in software verification.

\subheading{The Box Pattern}
Our design, using an inner \Q@Box@ object, is a common pattern in static verification: where one encapsulates all relevant mutating state into an encapsulated sub object which is not exposed to users.

% We also found natural to use the box pattern also in Spec\# requires.
Both our L42 and Spec\# code required us to use the box pattern for our \Q@SafeMovable@, due to the circular object graph caused by the \Q@Action@s of \Q@Button@s needing to change their enclosing \Q@SafeMovable@'s position.
% Also the implementation of the minimal GUI library (that our \Q@SaveMovable@ builds upon) has a much lower annotation burden.
%\LINE

\subheading{The Transform Pattern}
% A capsule mutator method is essentially a mutation of a field, which is guaranteed to not see the \Q@this@ object.
% Thus, if \Q@this@ is made invalid during  the method's execution, we could not observe it until after the method completes.
Suppose we want to scale a \Q@Widget@, we could add \Q@mut@ setters for \Q@width@, \Q@height@, \Q@left@, and \Q@top@ in the \Q@Widget@ interface. However, if we also wish to scale its children we have a problem, since \Q@Widget.children@ returns a \Q@read Widgets@, which does not allow mutation. We could of course add a \Q@mut@ method \Q@zoom@ to the \Q@Widget@ interface, however this does not scale if more operations are desired. If instead \Q@Widget.children@ returned a \Q@mut Widgets@, it would be difficult for \Q@Widget@ implementations, such as \Q@SafeMovable@, to keep control of their ROGs.

% In the above \Q@SafeMovable@ we only had one capsule mutator: \Q@dispatch@. However suppose a \Q@Widget@ wants to directly mutate it's descendents, however it can't do that since \Q@Widget.children@ returns a \Q@read Widgets@, if it returned a \Q@mut Widgets@ then \Q@SafeMovable@ could not be implement, as it's children are contained inside a capsule-field. 
% At first glance, it may seem that capsule mutators allow only very limited kinds %of mutation.
% This is however not the case. 

% Consider the following
% simple pattern to allow flexible use of encapsulated content: define a

A simple and practical solution would be to define a \Q@transform@ method in \Q@Widget@, and a \Q@Transformer@ interface 
like so:\footnote{A more general transformer could return a generic \Q@read R@.}
\begin{lstlisting}
interface Transformer<T> { method Void apply(mut T elem); }
interface Widget { ..
  mut method Void top(Int that); // setters for immutable data
  mut method read Void transform(Transformer<Widgets> t);
} // transformers for possibly encapsulated data
class SafeMovable { ..
  mut method Void transform(Transformer<Widgets> t) {
    return t.apply(this.box.c); }} // Well typed capsule mutator
\end{lstlisting}\saveSpace
% Note that the code above does not access a capsule field but merely calls a method that does; thus  it is \emph{not} a capsule mutator method, so it is not constrained by the restrictions on them. Code like the above would also allow one to mutate multiple capsule fields in one method.
%Our pattern cooperates with the language’s restrictions to ensure each mutation is completed as a separate operation, that is perceived by the rest of the system %as if it was atomic.%
%,  i.e. they can't see or update the other capsule fields.
The \Q@transform@ method offers an expressive power similar to \Q@mut@ getters, but prevents \Q@Widgets@ from leaking out.  With a \Q@Transformer@, a \Q@zoom@ function could be simply:
\begin{lstlisting}
static method Void zoom(mut Widget w) {
  w.transform(ws -> { for (wi : ws) { zoom(wi,scale); }});
  w.width(w.width() / 2); ..; w.top(w.top() / 2); }
\end{lstlisting}

% One of the advantages of this approach is that a the \@zoom@ method can be written by anyone anywhere

% \begin{lstlisting}[escapechar=\%]
%// Lambda Expression that creates a new Transformer<...>
%this.transform(l -> l.add(new MyWidget(..)))
%\end{lstlisting}
%//`i' is captured by the closure.
%// `imm' and `capsule' varaibles can be captured.

%    %\Comment{}%this.items.add(i);
%    // Cant instead capture `this': it can't be typed %as `imm'
%    // since `ItemTransformer.transform()' is an %`imm' method
%  })
%}
%  // instead of:
  %\Comment{}%this.exposeItems().add(i)

%Note that the code above does not access a capsule field but merely calls a method that does; thus
%it is \emph{not} a capsule mutator method, so it is not constrained by the restrictions on them. Code like the above would also allow one to mutate multiple capsule fields in one method.
%Our pattern cooperates with the language’s restrictions to ensure each mutation is completed as a separate operation, that is perceived by the rest of the system
%as if it was atomic.%
%,  i.e. they can't see or update the other capsule fields. %s:case-study
\section{Related work}
\label{s:related}

\subheading{Type Modifiers}
We rely on a combination of modifiers that are supported by at least 3 languages/lines of research:
L42~\cite{ServettoZucca15,ServettoEtAl13a,JOT:issue_2011_01/article1,GianniniEtAl16},
Pony~\cite{clebsch2015deny,clebsch2017orca}, and Gordon et.~al.~\cite{GordonEtAl12}; 
each of these works is accompanied by proofs about the properties of those modifiers.
Since such proofs have already been done, in this work we just assume the required properties.
Those approaches all support deep/strong interpretation, without back doors.

TM approaches like Javari~\cite{TschantzErnst05,Boyland06} and Rust~\cite{matsakis2014rust} are unsuitable since they introduce back doors which are not easily verifiable as being used properly.
Many approaches just try to preserve purity (as for example~\cite{pearce2011jpure}), but here we also need aliasing control.
Ownership~\cite{ClarkeEtAl13,ZibinEtAl10,DietlEtAl07} is another popular form of aliasing control that can be used as a building block for static verification~\cite{%
muller2002modular,%
barnett2011specification%
}.
Capsule/isolated local variables are affine in that they can be used only once, however this linearity is a property of variables, not expressions or fields. Linear/affine types extend this idea further, however they usually do not consider the ROGs of such types, or work in an OO setting~\cite{ahmed20073,fahndrich2002adoption}.
% The interpretation of capsule/isolated local variable is linear/affine: they can be used at most 1 time. See~\cite{Smith:2000:AT:645394.651903,ahmed20073,fahndrich2002adoption} for foundational work on linear/affine types, where...

%On ownership verification
%Peter Mueller and Arnd Peotzsch Heffter,  eg Müller, P.: Modular Specification and Verification of Object Oriented Programs, 2002.
%M. Barnett and M. Fähndrich and K. R. M. Leino and P. Müller and W. Schulte and H. Venter: Specification and Verification: The Spec# Experience. Communications of the ACM, 2011.
%\noindent\textit{Strong Exception Safety:}
%Exception safety seems at first glance a smaller issue with respect 
%to the other two, but is the final piece that lets the whole system work in a real world setting.
%Note that state of the art type systems to enforce exception safety
% do not restrict code that do not capture errors, and
%only the point of error capturing is constrained.

\subheading{Object Capabilities}
In literature, OCs are used to provide a wide range of guarantees, and many variations are present.
Object capabilities~\cite{RobustComposition}, in conjunction with type modifiers, are able to
 enforce purity of code in a modular way, without requiring the use of monads.
L42 and Gordon use OCs simply to reason about I/O and non determinism. This approach is best exemplified by Joe-E~\cite{finifter2008verifiable}, which is a self contained and minimalistic language using OCs over a subset of Java in order to reason about determinism.
However, in order for Joe-E to be a subset of Java, they leverage on a simplified model of immutability:
immutable classes must be final with only final fields that refer to immutable classes.
%Instances of immutable classes are immutable objects.
In Joe-E, every method that only takes instances of immutable classes is pure.
Thus their model would not allow the verification of purity for invariant methods of mutable objects.
In contrast our model has a more fine grained representation of mutability: it is \emph{reference based} instead of \emph{class based}. In our work, every method taking only \Q@read@ or \Q@imm@ \emph{references} is pure, regardless of their class type.
%;both in the sense that no object visible outside of the method is mutated, but also that it is deterministic.

\subheading{Class invariant protocols}
Class invariants are a fundamental part of the design by contract methodology. 
Invariant protocols differ wildly and can be unsound or complicated, particular due to re entrancy and aliasing~\cite{leino2004object,drossopoulou2008unified, meyer2016class}. 

While invariant protocols all seem to 
check and assume the invariant of an object after its construction, they handle invariants differently across object lifetimes; popular sound approaches include:
%literature on class invariant accepts that sometime the object invariant may not hold,
%and that is exacerbated because of 
%Leino, K. R. M. and Müller, P.: Object Invariants in Dynamic Contexts (ECOOP), 2004.
%S. Drossopoulou and A. Francalanza and  P. Müller and A. J. Summers: A Unified Framework for Verification Techniques for Object Invariants ECOOP 2008. 
%There are different options as to what object-invariants are known to hold:
\begin{itemize}
\item The invariants of objects in a \textit{steady} state are known to hold: that is when execution is not inside any of the objects public methods~\cite{Gopinathan:2008:RMO:1483018.1483028}. Invariants need to be constantly maintained between calls to public methods~\cite{WikiInvariant}.
\item 
%\LINE
The invariant of the receiver before a public method call and at the end of every public method body needs to be ensured. The invariant of the receiver at the beginning of a public method body and after a public method call can be assumed~\cite{Burdy2005,drossopoulou2008unified}.  
Some approaches ensure the invariant of the receiver of the \emph{calling} method, rather than the \emph{called} method~\cite{DBLP:journals/scp/MullerPL06}.
JML~\cite{JML} relaxes these requirements for helper methods, whose semantic is the same as if they were inlined.
%\LINE
%The invariant of the receiver (some approaches require the invariant of 'this' instead~\cite{?}) before a public (or non-helper~\cite{JML}) method call, and at the end of every method body needs to be ensured. The invariant of the receiver at the beginning of a public method body, and after a public method call can be assumed~\cite{Burdy2005,drossopoulou2008unified}.  

\item The same as above, but only for the bodies of `selectively exported' (i.e. non instance private) methods, and only for `qualified' (i.e. not \Q@this@) calls~\cite{meyer2016class}.
\item The invariant of an object is assumed only when a contract requires the object be `packed'. It is checked after an explicit `pack' operation, and objects can later be `unpacked'~\cite{DBLP:journals/jot/BarnettDFLS04}.
%\url{https://en.wikipedia.org/wiki/Class_invariant}}; %\item
%constantly maintained when the object is \textit{closed};
%invariant can be manually opened and closed by using special operations; % Add cite here!
\item Or, as in this work, the invariant of any object which could be \emph{involved} in execution is assumed to hold. It is checked after every modification of the object or its encapsulated ROG.
\end{itemize}
\noindent These different protocols can be deceivingly similar, and 
some approaches like JML suggest verifying a simpler approach (that method calls preserve the invariant of the \emph{receiver}) but assume a stronger one (the invariant of \emph{every} object, except \Q@this@, holds).

% use the unsound option of assuming one protocol, but actually checking another.

%DONE IN INTRO breaking class invariants = bug in class code
%braking validation= DEPEND.

%To encode this range of invariant semantics
%in our approach we can add a boolean \Q@isOpen@ field and add \Q@this.isOpen || ..@
%in front of the validity condition.
%Validation can be used to manually encode complex scenarios,
%for example if a method called on an object needs to break the invariant of another object,
%it can do so by manually setting the \Q@isOpen@ flag on the other object.

%On ownership verification
%Peter Mueller and Arnd Peotzsch Heffter,  eg Müller, P.: Modular Specification and Verification of Object-Oriented Programs, 2002.
%M. Barnett and M. Fähndrich and K. R. M. Leino and P. Müller and W. Schulte and H. Venter: Specification and Verification: The Spec# Experience. Communications of the ACM, 2011.

\newcommand\sepItems{\saveSpace\saveSpace\saveSpace\\*${}_{}$\\*${}_{}\,\bullet\,$}

\subheading{Runtime Verification Tools}
Many languages and tools support some form of runtime invariant checking (e.g. Eiffel~\cite{Meyer:1992:EL:129093}, D~\cite{Alexandrescu:2010:DPL:1875434},
and JML~\cite{Burdy2005}).
By looking to a survey by Voigt et al.~\cite{Voigt2013} and the extensive MOP project~\cite{meredith2012overview},
it seems that most runtime verification tools (RV) empower users
to implement the kind of monitoring they see fit for their specific problem at hand. This means that users are responsible for deciding, designing, and encoding both the logical properties and the instrumentation criteria~\cite{meredith2012overview}.
In the context of class invariants, this means the user defines the invariant protocol and the soundness of such protocol is not checked by the tool.

In practice, this means that the logic, instrumentation, and implementation end up connected:
a specific instrumentation strategy is only good to test certain logic properties in certain applications.
No guarantee is given that the implemented instrumentation strategy is able to support the required logic in the monitored application.
Some of these tools are designed to support class invariants: for example InvTS~\cite{gorbovitski08efficient} lets you write Python conditions that are verified on a set of Python objects, but the programmer needs to be able
to predict which objects are in need of being checked and to use a simpler domain specific language to target them. Hence if a programmer makes a mistake while using this domain specific language, invariant checking
will not be triggered.
Some tools are intentionally unsound and just perform invariant checking following some heuristic that is expected to catch most failures: such as jmlrac~\cite{Burdy2005} and Microsoft Code Contracts~\cite{fahndrich2010embedded}.

%In particular, the heuristic of 
%We encoded our GUI example also on Microsoft Code Contract; their system also ran the invariant checking $77$ times. Their system is easy to use, but it is unsound since it is built over an unsound/incomplete static verifier~\cite{?}.

%
%In this work we define a language where a minimal, standardized,
%efficient and completely general purpose instrumentation strategy can soundly verify conditions
%expressible as a\\* \Q@read method imm Bool invariant()@, for any well-typed program; with open world assumption
%and possible Byzantine behaviour of any object in the system.
%
%By seeing class invariant as a part of the type of the object,
%the `RV tool' philosophy is akin to letting the programmer customize the behaviour of the
%type system: the programmer implementation may be unsound; while our philosophy is
%to give the user a way to represent complex and expressive types (in the form of arbitrary code in 
%the \Q@invariant()@ method), but 
%the type system implementation is fixed in stone by the language designer.

Many works attempt to move out of the `RV tool' philosophy to ensure RV monitors work as expected, as for example
%\sepItems
%In avionics, where memory allocation is disallowed, making reasoning about aliasing much simpler~\cite{laurent2015assuring}:
%``\emph{Runtime Verification (RV) can act as the last line of defense to
%protect the public safety, but only if the RV system itself is trusted.}''.
%\sepItems
%In domain specific languages~\cite{ferrari2002guardians}:
%``\emph{Proof techniques for establishing security properties}''.
%\sepItems
%On assertions over restrictive domain specific languages, to tame some of the C/C++
%undefined behaviour~\cite{agten2015sound}:
%``\emph{no verified assertion in the verified
%module will ever fail at runtime, even if the module runs as part of
%a vulnerable application thSound and Unsound monitorsat is subject to code injection attacks}''.
the study of contracts as refinements of types~\cite{findler2001contract}.
However, such work is only interested in pre and post conditions, not class invariants.

Our invariant protocol is much stronger then
visible state semantics, and keeps the invariant under tight control.
Gopinathan et.~al.'s.~\cite{Gopinathan:2008:RMO:1483018.1483028} approach keeps
a similar level of control:
relying on powerful aspect oriented support, they detect any field update in the whole ROG of any object, and check all the invariants that such update may have violated.
We agree with their criticism of visible state semantics, where  methods still have to assume that any object may be broken; in such case calling any public method would trigger an error, but while the object is just passed around (and for example stored in collections), the broken state will not be detected; Gopinathan et.~al. says ``\emph{there are many instances where $o$'s invariant is violated by the programmer inadvertently changing the state of $p$ when $o$ is in a steady state. Typically, $o$ and $p$ are objects exposed by the API, and the programmer (who is the user of the API), unaware of the dependency between $o$ and $p$, calls a method of $p$ in such a way that $o$'s invariant is violated. The fact that the violation occurred is detected much later, when a method of $o$ is called again, and it is difficult to determine exactly where such violations occur.}''

However, their approach addresses neither exceptions nor non determinism caused by I/O, so their work is unsound if those aspects are taken into consideration.

Their approach is very computationally intensive, but we think it is powerful enough that it could even be used to roll back the very field update that caused the invariant to fail, making the object valid again.
We considered a roll back approach for our work, however rolling back a single field update is likely to be completely unexpected, rather we should roll back more meaningful operations, similarly to what happens
with transactional memory, and so is likely to be very hard to support efficiently.
%However we think roll-back this would be a 
%\REVComm{\REVComm{terrible}{2}{It seems in poor taste to complain of ``terrible'' ideas, especially without attempting to demonstrate the improvements of the proposed approach.}}{3}{Nontechnical term. It is not a great idea to label previous work as ``terrible''}
% ideally not only the field-update breaking the invariant should be reverted, %the roll-back should 
Using TMs to enforce strong exception safety is a much simpler alternative, providing the same level of safety, albeit being more restrictive (namely that if the operation did succeed it is still effectively rolled back).

%: for example
%assume that we are moving object between two boats:
%the overflowing object may be removed from the \Q@cargo@ of the second boat, but it would not
%be placed back in the first boat. It would look like the object has disappeared.
%The important pTheir approach is very computationally intensive, but we think it is powerful enough that it could even be used to roll-back the very field update that caused oint here is that the program would be in an unexpected state
%even if no object invariants are violated, and this would happen \textbf{because} of the 
%invariant checking/fixing behaviour, not because of code written by the programmer.
%We believe that the only viable option is to detect violations after the fact.

%\LINE
%Another approach used in the dynamic language Racket is to interpose on primitive operations like procedure-calls and field updates; this allows one to enforce visible-state semantics by wrapping invariant operations around such operations (as is done in aspect-oriented systems like Jose~{?}). This technique can be used with gradual typing to dynamically enforce `types' of mutable structures in a safe way.

%\LINE
%\subheading{Chaperones and impersonators}
Chaperones and impersonators~\cite{DBLP:conf/oopsla/StricklandTFF12} lifts the techniques of gradual typing~\cite{takikawa2015towards,DBLP:conf/oopsla/TakikawaSDTF12,DBLP:conf/popl/WrigstadNLOV10}
to work on general purpose predicates, where
values can be wrapped to ensure an invariant holds.
This technique is very powerful and can be used to enforce pre and post conditions by wrapping function arguments and return values.
This technique however does not monitor the effects of aliasing, as such they may notice if a contract has broken, but not when or why. In addition, due to the difficulty of performing static analysis in weakly typed languages, they need to inject runtime checking code around every user facing operation.
Aspect oriented systems like Jose~\cite{feldman2006jose}, similarly wrap invariant checks around method bodies.

%One of the advantages of their system is that is transparent to the user. 
%\LINE

%\LINE

%\noindent\textit{Performance}
%Our case study shows that our sound approach can monitor programs
%for a fraction of the cost of many other approaches.
%Many other works%
%~\cite{feldman2006jose,fahndrich2010embedded,abercrombie2002jcontractor,tran2003design}
% check/run
%the invariant code at the start and end of every public
%method; this even include trivial getters.
%In  our approach, we call the \Q@invariant@ method
%one time at the end of each setter, capsule mutator method and constructor.
%We do not inject it at the end of other methods, which are usually more numerous and invoked much more often.
%Of course, \Q@invariant@ can still be called indirectly, for example by calling a setter.
%We expect our approach to result in a dramatic reduction over the number of required checks,
%except for cases when public methods just update many fields directly (without using setters).

%\LINE
\subheading{Security and Scalability}
Our approach allows verifying an object's invariant independently of the actual invariants of other objects.
This is in contrast with the main strategy of static verification: to verify a method, the system assumes the contracts of other methods, and the content of those contracts is the starting point for their proof.
Thus, static verification proceeds like a mathematical proof: a program is valid if it is all correct, but a single error invalidates all claims. This makes it hard to perform verification on large programs, or when independently maintained third party libraries are involved.
This is less problematic with a type system, since its properties are more coarse grained, simpler and easier to check.
 Static verification has more flexible and fine grained annotations and often relies on a fragile theorem prover as a backend.

%\REVComm{
%To solve this issue, static verification systems are %starting to
%}{2}{[is this correct?] verification of reference %monitors, gradual typing, and contracts have been %explored for longer}
To soundly verify code embedded in an untrusted environment, as in gradual typing~\cite{DBLP:conf/oopsla/TakikawaSDTF12,DBLP:conf/popl/WrigstadNLOV10}, it is possible to 
consider a verified core and a runtime verified boundary.
You can see our approach as an extremely modularized	version of such system: every class is its
own verified core, and the rest of the code could have Byzantine behaviour. Our formal proofs show that every class that compiles/type checks is 
soundly handled by our protocol, independently of the code that uses such class or any other surrounding code.

Our approach works both  in a library setting and with the open world assumption.
Consider for example the work of Parkinson~\cite{parkinson2007class}: in his short paper he verified a property of the \Q@Subject/Observer@ pattern. However, the proof relies on (any override of) the \Q@Subject.register(Observer)@ method respecting its contract. Such assumption is unrealistic in a real world system with dynamic class loading, and could trivially be broken by a user defined \Q@EvilSubject@.

\subheading{Static Verification}
Spec\#~\cite{Barnett:2004:SPS:2131546.2131549} is a language built on top of C\#, it adds various annotations such as method contracts and class invariants. 
It primarily follows the Boogie methodology~\cite{DBLP:journals/tcs/NaumannB06} where (implicit) annotations are used to specify and modify the owner of objects and whether their invariants are required to hold. Invariants can be \emph{ownership} based~\cite{DBLP:journals/jot/BarnettDFLS04}, where an invariant only depends on objects it owns; or \emph{visibility} based~\cite{DBLP:conf/mpc/BarnettN04,DBLP:conf/ecoop/LeinoM04}, where an invariant may depend on objects it doesn't own, provided that the class of such objects know about this dependence. Unlike our approach, Spec\# does not restrict the aliases that may exist for an object, rather it restricts object mutation: an object cannot be modified if the invariant of its owner is required to hold. This is more flexible than our approach as it also allows only part of an object's ROG to be owned/encapsulated. However as we showed in Section \ref{s:case-study}, it can become much more difficult to work with and requires significant annotation since merely having an alias to an object
is insufficient to modify it or call methods on it.
%tells you nothing about it, hindering
%modification and method calls.
Spec\# also works with existing .NET libraries by annotating them with contracts, however such annotations are not verified. Spec\#, like us, does perform runtime checks for invariants and throws unchecked exceptions on failure.  However Spec\# does not allow soundly recovering from an invariant failure, since catching unchecked exceptions in Spec\# is intentionally unsound.~\cite{Leino2004ExceptionSF}

Another system is AutoProof~\cite{DBLP:conf/fm/PolikarpovaTFM14}, a static verifier for Eiffel that also follows the Boogie methodology, but extends it with \emph{semantic collaboration} where objects keep track of their invariants' dependencies using ghost state.
Dafny~\cite{DBLP:conf/sigada/Leino12} is a new language where all code is statically verified, it supports invariants by injecting pre and post conditions following visible state semantics;
however it requires objects to be newly allocated (or cloned) before another object's invariant may depend on it.
Dafny is also generally highly restrictive with its rules for mutation, and object construction, it also does not provide any means of performing non deterministic I/O.

%Spec\# is statically verified wheras we rely on a type system: we have 4 type-modifiers that can be applied anywhere a type can be used (like a variable declaration) and the type-system uses a small set of fixed deterministic rules.
%Wheras the static-verification aproach has much more flexible and fine-grained annotations (with meth pre/post conditions) and uses a theoreom prover as a back-end, this can make it harder for users to program as it is not obvious whether the theorom prover will accept a program or not.
% In addition, the approach of a static-verifier can also be non-modular: changes to the body of one method could affect whether another is verified.

%many works on static verification, such as thoser Spec\#~\cite{??}

\subheading{Specification languages}
Using a specification language based on the mathematical metalanguage and different from the program language's semantics may seem attractive, since it can express uncomputable concepts, has no mutation or non determinism, and is often easier to formally reason about.

However, a study~\cite{chalin2007logical} discovered that developers expect specification languages to follow the semantics of the underling language, including short circuit semantics and arithmetic exceptions; thus for example \Q@1/0@\,\Q@||@\,\Q@2>1@ should not hold, while \Q@2>1@\,\Q@||@\,\Q@1/0@ should, thanks to short circuiting.
This study was influential enough to convince JML to change its interpretation of logical expressions
accordingly~\cite{chalin2008jml}.
Dafny~\cite{DBLP:conf/sigada/Leino12} uses a hybrid approach: it has mostly the same language for both specification and execution. Specification (`ghost') contexts can use uncomputable constructs such as universal quantification over infinite sets. Whereas runtime contexts allow mutation, object allocation and print statements. The semantics of shared constructs (such as short circuiting logic operators) is the same in both contexts.

Most runtime verification systems, such as ours, use a metacircular approach: specifications are simply code in the underlying language. Since specifications are checked at runtime, they are unable to verify uncomputable contracts.
 Ensuring determinism in a non functional language is challenging. Spec\# recognizes the need for purity/determinism when method calls are allowed in contracts~\cite{barnett200499} `\emph{There are three main current approaches: a) forbid the use of functions in specifications, b) allow only provably pure functions, or c) allow programmers free use
	of functions. The first approach is not scalable, the second overly restrictive and
	the third unsound.}'.

They recognize that many tools unsoundly use option (c), such as AsmL~\cite{barnett2003runtime}.
Spec\# aims to follow (b) but only considers non determinism caused by memory mutation, and allows other non deterministic operations, such as I/O and random number generation. For example, the following method verifies:
\begin{lstlisting}
[Pure] bool uncertain() {return new Random().Next() %$$ 2 == 0;}
\end{lstlisting}

And so \Q@assert uncertain() == uncertain();@ also verifies, but randomly fails with an exception at runtime.
As you can see failing to handle non determinism jeopardises reasoning.

A simpler and more restrictive solution to these problems is to prevent `pure' functions from reading or writing to any non final fields, or calling any impure functions. This is the approach used by~\cite{Flanagan06hybridtypes}, one advantage of their approach is that invariants (which must be `pure') can read from a chain of final fields, even when they are contained in otherwise mutable objects. However their approach completely prevents invariants from mutating newly allocated objects, thus greatly restricting how computations can be performed.
\saveSpace

\section{Conclusions and Future Work}
\label{s:conclusion}
Our approach follows the principles of \emph{offensive programming}
~\cite{stephens2015beginning}, where 
no attempt to fix or recover an invalid object is performed and
%	\begin{itemize}
%\item
 failures (unchecked exceptions)
		are raised close to their cause: at the end of constructors creating invalid objects and immediately after field updates and instance methods that invalidate their receivers.

%}{3}{[meaning] is not clear} (the operation creating an invalid object), i.e. we ``fail-fast''.    
%		\item
%	\end{itemize}

%The aim of our work is only to enforce object invariants, so we do not present complexities unnecessary for this purpose.
Our work builds on a specific form of TMs and OCs, whose
popularity is growing, and we expect future languages to support some variation of these.
Crucially, any language already designed with such TMs and OCs
can also support our invariant protocol with minimal added complexity.

We demonstrated the applicability and simplicity of our approach with a GUI example.
Our invariant protocol performs several orders of magnitude less checks than visible state semantics, and requires much less annotation 
than Spec\#, (the system with the most comparable goals). In Section~\ref{s:formalism} we formalised our invariant protocol and in Appendix~\ref{s:proof} we prove it sound.
%In appendix~\ref{s:formalism} we formalise our invariant protocol and prove it sound. 
To stay parametric over the various existing type systems which provably enforce the properties we require for our proof (and much more), we do not formalise any specific type system.

One interesting avenue for future work would be to
use invariants to encode pre and post conditions,
as done by~\cite{Flanagan06hybridtypes}: where pre and post conditions are encoded as the invariants of the parameter and return types (respectively).
Without good syntax sugar, such an approach could be quite verbose, however it would ensure that a methods precondition holds during the entire execution of a method, and not just the beginning. In addition this could be more efficient than traditional runtime checking when the same argument is used in the invocations of methods with the same pre condition, as happens often in practice for recursive methods: where many parameters are simply parsed unmodified in recursive calls.

% a method could be declared as taking a class whose invariant corresponds to the method's pre-condition,  and returning a class whose invariant corresponds to the pos-condition.

% Such approach may be quite verbose, but would ensure that the precondition on the argument holds for the whole execution of the method, instead of just holding at the beginning.

%It could be worthwhile formalising the minimal type system required by validation.

%However the restrictions we make to ensures that \Q@validate@ is deterministic, namely those the type-system enforces due to its signature, seem quite flexible and reasonable;

%%%%%examples of things that future work may investigate allowing are deterministic I/O and multi-threading. 

The language we presented here restricts the forms of \Q@invariant@ and capsule mutator methods;
such strong restrictions allow for sound and efficient injection of invariant checks. 
These restrictions do not get in the way of writing invariants over immutable data, but the box pattern is required for verifying complex mutable data structures. We believe this pattern, although verbose, is simple and understandable. While it may be possible for a more complex and fragile type system to reduce the need for the pattern whilst still ensuring our desired semantics, we prioritize simplicity and generality. 

In order to obtain safety, simplicity, and efficiency we traded some expressive power:
the \Q@invariant@ method can only refer to immutable and encapsulated state.
This means that while we can easily verify that a doubly linked list of immutable elements
is correctly linked up,
we can not do the same for a doubly linked lists of mutable elements. Our approach does not prevent correctly implementing such data structures, but the \Q@invariant@ method would be unable to access the list's nodes, since they would contain \Q@mut@ references to shared objects.
In order to verify such data structures we could add a special kind of field which cannot be (transitively) accessed by invariants; such fields could freely refer to any object. We are however unsure if such complexity would be justified.

% To verify those data-structures, in future work
% we may investigate a special kind of field that
% could be accessed only using a \Q@mut@ receiver.
% Such fields would be allowed to refer to not encapsulated state, 
% and they would be unreachable from the invariant code,that starts from a \Q@read this@.

% \LINE
% The language we presented here restricts the form of \Q@invariant@ and capsule mutator methods. 
% We have shown that such restrictions, albeit strong, allow sound and efficient injection of invariant checks. 
% While our restrictions do not hamper writing invariants over immutable data, invariants over complex mutable % data require the box pattern.  We believe the box pattern, although verbose, is simple and understandable, but % could be improved with syntax sugar.

% Our goals of simplicity and efficiency come at the cost of expressivity: we are unable to express invariants % over non-encapsulated mutable structures, 
% though a more complex and fragile type system may reduce such limitations. However we believe we have % demonstrated that our limitations are not too severe and that we have achieved our goals.
% \LINE

%, however such a language is unlikely to be easily understood by programmers;
%being able to predict whether code would be well typed allows programmers
%to better take advantage of the language.

For an implementation of our work to be sound, catching exceptions like stack overflows or out of memory
cannot be allowed in \Q@invariant@ methods, since they are not deterministically thrown.
%For an implementation of our work to be sound, non-deterministic exceptions like stack overflows or out of memory
%errors cannot be caught in invariants.
%this
%use exception catching as a non deterministic conditional choice, 
%allowing non deterministic behaviour.
Currently L42 never allows catching them, however we could also write a (native) capability method (which can't be used inside an invariant) that enables catching them. Another option worth exploring would be to make such exceptions deterministic, perhaps by giving invariants fixed stack and heap sizes.

Other directions that could be investigated to improve our work include the addition of syntax sugar to ease the burden of the box and the transform patterns; type modifier inference, and support for flexible ownership types.

\bibliography{main} % The template has this at the end...

\begin{thebibliography}{10}

\bibitem{Abrahams2000}
David Abrahams.
\newblock {\em Exception-Safety in Generic Components}, pages 69--79.
\newblock Springer Berlin Heidelberg, Berlin, Heidelberg, 2000.
\newblock \href {http://dx.doi.org/10.1007/3-540-39953-4_6}
  {\path{doi:10.1007/3-540-39953-4_6}}.

\bibitem{DBLP:conf/pldi/AikenFKT03}
Alexander Aiken, Jeffrey~S. Foster, John Kodumal, and Tachio Terauchi.
\newblock Checking and inferring local non-aliasing.
\newblock In {\em Proceedings of the {ACM} {SIGPLAN} 2003 Conference on
  Programming Language Design and Implementation 2003, San Diego, California,
  USA, June 9-11, 2003}, pages 129--140, 2003.
\newblock \href {http://dx.doi.org/10.1145/781131.781146}
  {\path{doi:10.1145/781131.781146}}.

\bibitem{Alexandrescu:2010:DPL:1875434}
Andrei Alexandrescu.
\newblock {\em The D Programming Language}.
\newblock Addison-Wesley Professional, 1st edition, 2010.

\bibitem{DBLP:conf/fmco/BarnettCDJL05}
Michael Barnett, Bor{-}Yuh~Evan Chang, Robert DeLine, Bart Jacobs, and
  K.~Rustan~M. Leino.
\newblock Boogie: {A} modular reusable verifier for object-oriented programs.
\newblock In {\em Formal Methods for Components and Objects, 4th International
  Symposium, {FMCO} 2005, Amsterdam, The Netherlands, November 1-4, 2005,
  Revised Lectures}, pages 364--387, 2005.
\newblock \href {http://dx.doi.org/10.1007/11804192\_17}
  {\path{doi:10.1007/11804192\_17}}.

\bibitem{DBLP:journals/jot/BarnettDFLS04}
Michael Barnett, Robert DeLine, Manuel F{\"{a}}hndrich, K.~Rustan~M. Leino, and
  Wolfram Schulte.
\newblock Verification of object-oriented programs with invariants.
\newblock {\em Journal of Object Technology}, 3(6):27--56, 2004.
\newblock \href {http://dx.doi.org/10.5381/jot.2004.3.6.a2}
  {\path{doi:10.5381/jot.2004.3.6.a2}}.

\bibitem{DBLP:conf/mpc/BarnettN04}
Michael Barnett and David~A. Naumann.
\newblock Friends need a bit more: Maintaining invariants over shared state.
\newblock In {\em Mathematics of Program Construction, 7th International
  Conference, {MPC} 2004, Stirling, Scotland, UK, July 12-14, 2004,
  Proceedings}, pages 54--84, 2004.
\newblock \href {http://dx.doi.org/10.1007/978-3-540-27764-4\_5}
  {\path{doi:10.1007/978-3-540-27764-4\_5}}.

\bibitem{barnett2011specification}
Mike Barnett, Manuel F{\"a}hndrich, K~Rustan~M Leino, Peter M{\"u}ller, Wolfram
  Schulte, and Herman Venter.
\newblock Specification and verification: the spec\# experience.
\newblock {\em Communications of the ACM}, 54(6):81--91, 2011.
\newblock \href {http://dx.doi.org/10.1145/1953122.1953145}
  {\path{doi:10.1145/1953122.1953145}}.

\bibitem{Barnett:2004:SPS:2131546.2131549}
Mike Barnett, K.~Rustan~M. Leino, and Wolfram Schulte.
\newblock The spec\# programming system: An overview.
\newblock In {\em Proceedings of the 2004 International Conference on
  Construction and Analysis of Safe, Secure, and Interoperable Smart Devices},
  CASSIS'04, pages 49--69, Berlin, Heidelberg, 2005. Springer-Verlag.
\newblock \href {http://dx.doi.org/10.1007/978-3-540-30569-9_3}
  {\path{doi:10.1007/978-3-540-30569-9_3}}.

\bibitem{barnett200499}
Mike Barnett, David~A Naumann, Wolfram Schulte, and Qi~Sun.
\newblock 99.44\% pure: Useful abstractions in specifications.
\newblock In {\em ECOOP workshop on Formal Techniques for Java-like Programs
  (FTfJP)}, 2004.
\newblock \href {http://dx.doi.org/10.1.1.72.3429} {\path{doi:10.1.1.72.3429}}.

\bibitem{barnett2003runtime}
Mike Barnett and Wolfram Schulte.
\newblock Runtime verification of .net contracts.
\newblock {\em Journal of Systems and Software}, 65(3):199--208, 2003.
\newblock \href {http://dx.doi.org/10.1016/S0164-1212(02)00041-9}
  {\path{doi:10.1016/S0164-1212(02)00041-9}}.

\bibitem{BirkaErnst04}
Adrian Birka and Michael~D. Ernst.
\newblock A practical type system and language for reference immutability.
\newblock In {\em ACM SIGPLAN Conference on Object-Oriented Programming,
  Systems, Languages and Applications (OOPSLA 2004)}, pages 35--49, 2004.
\newblock \href {http://dx.doi.org/10.1145/1035292.1028980}
  {\path{doi:10.1145/1035292.1028980}}.

\bibitem{Bloch08}
Joshua Bloch.
\newblock {\em Effective {J}ava (2Nd Edition) (The Java Series)}.
\newblock Prentice Hall PTR, 2 edition, 2008.

\bibitem{boyland2001alias}
John Boyland.
\newblock Alias burying: Unique variables without destructive reads.
\newblock {\em Software: Practice and Experience}, 31(6):533--553, 2001.
\newblock \href {http://dx.doi.org/10.1002/spe.370}
  {\path{doi:10.1002/spe.370}}.

\bibitem{boyland2003checking}
John Boyland.
\newblock Checking interference with fractional permissions.
\newblock In {\em International Static Analysis Symposium}, pages 55--72.
  Springer, 2003.

\bibitem{Boyland06}
John Boyland.
\newblock Why we should not add readonly to {J}ava (yet).
\newblock {\em Journal of Object Technology}, 5(5):5--29, 2006.
\newblock \href {http://dx.doi.org/10.5381/jot.2006.5.5.a1}
  {\path{doi:10.5381/jot.2006.5.5.a1}}.

\bibitem{Boyland10}
John Boyland.
\newblock Semantics of fractional permissions with nesting.
\newblock {\em ACM Transactions on Programming Languages and Systems}, 32(6),
  2010.
\newblock \href {http://dx.doi.org/10.1145/1749608.1749611}
  {\path{doi:10.1145/1749608.1749611}}.

\bibitem{Burdy2005}
Lilian Burdy, Yoonsik Cheon, David~R. Cok, Michael~D. Ernst, Joseph~R. Kiniry,
  Gary~T. Leavens, K.~Rustan~M. Leino, and Erik Poll.
\newblock An overview of jml tools and applications.
\newblock {\em International Journal on Software Tools for Technology
  Transfer}, 7(3):212--232, Jun 2005.
\newblock \href {http://dx.doi.org/10.1007/s10009-004-0167-4}
  {\path{doi:10.1007/s10009-004-0167-4}}.

\bibitem{chalin2007logical}
Patrice Chalin.
\newblock Are the logical foundations of verifying compiler prototypes matching
  user expectations?
\newblock {\em Formal Aspects of Computing}, 19(2):139--158, 2007.
\newblock \href {http://dx.doi.org/10.1007/s00165-006-0016-1}
  {\path{doi:10.1007/s00165-006-0016-1}}.

\bibitem{chalin2008jml}
Patrice Chalin and Fr{\'e}d{\'e}ric Rioux.
\newblock Jml runtime assertion checking: Improved error reporting and
  efficiency using strong validity.
\newblock {\em FM 2008: Formal Methods}, pages 246--261, 2008.
\newblock \href {http://dx.doi.org/10.1007/978-3-540-68237-0_18}
  {\path{doi:10.1007/978-3-540-68237-0_18}}.

\bibitem{ClarkeEtAl13}
Dave Clarke, Johan {\"{O}}stlund, Ilya Sergey, and Tobias Wrigstad.
\newblock Ownership types: {A} survey.
\newblock In Dave Clarke, James Noble, and Tobias Wrigstad, editors, {\em
  Aliasing in Object-Oriented Programming. Types, Analysis and Verification},
  volume 7850 of {\em Lecture Notes in Computer Science}, pages 15--58.
  Springer, 2013.
\newblock \href {http://dx.doi.org/10.1007/978-3-642-36946-9_3}
  {\path{doi:10.1007/978-3-642-36946-9_3}}.

\bibitem{ClarkeWrigstad03}
David Clarke and Tobias Wrigstad.
\newblock External uniqueness is unique enough.
\newblock In {\em ECOOP'03 - Object-Oriented Programming}, volume 2473 of {\em
  Lecture Notes in Computer Science}, pages 176--200. Springer, 2003.
\newblock \href {http://dx.doi.org/10.1007/978-3-540-45070-2_9}
  {\path{doi:10.1007/978-3-540-45070-2_9}}.

\bibitem{clebsch2015deny}
Sylvan Clebsch, Sophia Drossopoulou, Sebastian Blessing, and Andy McNeil.
\newblock Deny capabilities for safe, fast actors.
\newblock In {\em Proceedings of the 5th International Workshop on Programming
  Based on Actors, Agents, and Decentralized Control}, pages 1--12. ACM, 2015.
\newblock \href {http://dx.doi.org/10.1145/2824815.2824816}
  {\path{doi:10.1145/2824815.2824816}}.

\bibitem{clebsch2017orca}
Sylvan Clebsch, Juliana Franco, Sophia Drossopoulou, Albert~Mingkun Yang,
  Tobias Wrigstad, and Jan Vitek.
\newblock Orca: Gc and type system co-design for actor languages.
\newblock {\em Proceedings of the ACM on Programming Languages}, 1(OOPSLA):72,
  2017.
\newblock \href {http://dx.doi.org/10.1145/3133896}
  {\path{doi:10.1145/3133896}}.

\bibitem{DRef}
{D Language Foundation}.
\newblock {D Programming Language Specification}, 2018.
\newblock URL: \url{https://dlang.org/dlangspec.pdf}.

\bibitem{DietlEtAl07}
Werner Dietl, Sophia Drossopoulou, and Peter M{\"u}ller.
\newblock Generic universe types.
\newblock In {\em ECOOP'07 - Object-Oriented Programming}, volume 4609 of {\em
  Lecture Notes in Computer Science}, pages 28--53. Springer, 2007.
\newblock \href {http://dx.doi.org/10.1007/978-3-540-73589-2_3}
  {\path{doi:10.1007/978-3-540-73589-2_3}}.

\bibitem{drossopoulou2008unified}
Sophia Drossopoulou, Adrian Francalanza, Peter M{\"u}ller, and Alexander~J
  Summers.
\newblock A unified framework for verification techniques for object
  invariants.
\newblock In {\em European Conference on Object-Oriented Programming}, pages
  412--437. Springer, 2008.
\newblock \href {http://dx.doi.org/10.1007/978-3-540-70592-5_18}
  {\path{doi:10.1007/978-3-540-70592-5_18}}.

\bibitem{DBLP:conf/sac/FahndrichBL10}
Manuel F{\"{a}}hndrich, Michael Barnett, and Francesco Logozzo.
\newblock Embedded contract languages.
\newblock In {\em Proceedings of the 2010 {ACM} Symposium on Applied Computing
  (SAC), Sierre, Switzerland, March 22-26, 2010}, pages 2103--2110, 2010.
\newblock \href {http://dx.doi.org/10.1145/1774088.1774531}
  {\path{doi:10.1145/1774088.1774531}}.

\bibitem{fahndrich2010embedded}
Manuel F{\"a}hndrich, Michael Barnett, and Francesco Logozzo.
\newblock Embedded contract languages.
\newblock In {\em Proceedings of the 2010 ACM Symposium on Applied Computing},
  pages 2103--2110. ACM, 2010.
\newblock \href {http://dx.doi.org/10.1145/1774088.1774531}
  {\path{doi:10.1145/1774088.1774531}}.

\bibitem{fahndrich2002adoption}
Manuel Fahndrich and Robert DeLine.
\newblock Adoption and focus: Practical linear types for imperative
  programming.
\newblock In {\em ACM SIGPLAN Notices}, volume~37, pages 13--24. ACM, 2002.
\newblock \href {http://dx.doi.org/10.1145/543552.512532}
  {\path{doi:10.1145/543552.512532}}.

\bibitem{feldman2006jose}
Yishai~A Feldman, Ohad Barzilay, and Shmuel Tyszberowicz.
\newblock Jose: Aspects for design by contract.
\newblock In {\em Software Engineering and Formal Methods, 2006. SEFM 2006.
  Fourth IEEE International Conference on}, pages 80--89. IEEE, 2006.
\newblock \href {http://dx.doi.org/10.1109/SEFM.2006.26}
  {\path{doi:10.1109/SEFM.2006.26}}.

\bibitem{findler2001contract}
Robert~Bruce Findler and Matthias Felleisen.
\newblock Contract soundness for object-oriented languages.
\newblock In {\em ACM SIGPLAN Notices}, volume~36, pages 1--15. ACM, 2001.
\newblock \href {http://dx.doi.org/10.1145/504311.504283}
  {\path{doi:10.1145/504311.504283}}.

\bibitem{finifter2008verifiable}
Matthew Finifter, Adrian Mettler, Naveen Sastry, and David Wagner.
\newblock Verifiable functional purity in java.
\newblock In {\em Proceedings of the 15th ACM conference on Computer and
  communications security}, pages 161--174. ACM, 2008.
\newblock \href {http://dx.doi.org/10.1145/1455770.1455793}
  {\path{doi:10.1145/1455770.1455793}}.

\bibitem{Flanagan06hybridtypes}
Cormac Flanagan.
\newblock Hybrid types, invariants, and refinements for imperative objects.
\newblock In {\em In International Workshop on Foundations and Developments of
  Object-Oriented Languages}, 2006.

\bibitem{JML}
{Gary T. Leavens, Erik Poll, Curtis Clifton, Yoonsik Cheon, Clyde Ruby, David
  Cok, Peter M¨uller, Joseph Kiniry, Patrice Chalin, Daniel M. Zimmerman,
  Werner Dietl}.
\newblock {JML Reference Manual}, 2013.
\newblock URL:
  \url{http://www.eecs.ucf.edu/~leavens/JML//refman/jmlrefman.pdf}.

\bibitem{GianniniEtAl16}
Paola Giannini, Marco Servetto, and Elena Zucca.
\newblock Types for immutability and aliasing control.
\newblock In {\em ICTCS'16 - Italian Conf. on Theoretical Computer Science},
  volume 1720 of {\em {CEUR} Workshop Proceedings}, pages 62--74. CEUR-WS.org,
  2016.
\newblock URL: \url{http://ceur-ws.org/Vol-1720/full5.pdf}.

\bibitem{Gopinathan:2008:RMO:1483018.1483028}
Madhu Gopinathan and Sriram~K. Rajamani.
\newblock Runtime verification.
\newblock chapter Runtime Monitoring of Object Invariants with Guarantee, pages
  158--172. Springer-Verlag, Berlin, Heidelberg, 2008.
\newblock \href {http://dx.doi.org/10.1007/978-3-540-89247-2_10}
  {\path{doi:10.1007/978-3-540-89247-2_10}}.

\bibitem{gorbovitski08efficient}
Michael Gorbovitski, Tom Rothamel, Yanhong~A. Liu, and Scott~D. Stoller.
\newblock Efficient runtime invariant checking: A framework and case study.
\newblock In {\em Proceedings of the 6th International Workshop on Dynamic
  Analysis (WODA 2008)}. ACM Press, 2008.
\newblock \href {http://dx.doi.org/10.1145/1401827.1401837}
  {\path{doi:10.1145/1401827.1401837}}.

\bibitem{GordonEtAl12}
Colin~S. Gordon, Matthew~J. Parkinson, Jared Parsons, Aleks Bromfield, and Joe
  Duffy.
\newblock Uniqueness and reference immutability for safe parallelism.
\newblock In {\em ACM SIGPLAN Conference on Object-Oriented Programming,
  Systems, Languages and Applications (OOPSLA 2012)}, pages 21--40. ACM Press,
  2012.
\newblock \href {http://dx.doi.org/10.1145/2384616.2384619}
  {\path{doi:10.1145/2384616.2384619}}.

\bibitem{HallerOdersky10}
Philipp Haller and Martin Odersky.
\newblock Capabilities for uniqueness and borrowing.
\newblock In Theo D'Hondt, editor, {\em ECOOP'10 - Object-Oriented
  Programming}, volume 6183 of {\em Lecture Notes in Computer Science}, pages
  354--378. Springer, 2010.
\newblock \href {http://dx.doi.org/10.1007/978-3-642-14107-2_17}
  {\path{doi:10.1007/978-3-642-14107-2_17}}.

\bibitem{Hogg91}
John Hogg.
\newblock Islands: Aliasing protection in object-oriented languages.
\newblock In {\em ACM Symp. on Object-Oriented Programming: Systems, Languages
  and Applications 1991}, pages 271--285. ACM Press, 1991.

\bibitem{IgarashiEtAl01}
Atsushi Igarashi, Benjamin~C. Pierce, and Philip Wadler.
\newblock Featherweight {J}ava: a minimal core calculus for {J}ava and {GJ}.
\newblock {\em ACM Transactions on Programming Languages and Systems},
  23(3):396--450, 2001.

\bibitem{karger1988improving}
Paul~Ashley Karger.
\newblock {\em Improving security and performance for capability systems}.
\newblock PhD thesis, Citeseer, 1988.

\bibitem{JOT:issue_2011_01/article1}
Giovanni Lagorio and Marco Servetto.
\newblock Strong exception-safety for checked and unchecked exceptions.
\newblock {\em Journal of Object Technology}, 10:1:1--20, 2011.
\newblock \href {http://dx.doi.org/10.5381/jot.2011.10.1.a1}
  {\path{doi:10.5381/jot.2011.10.1.a1}}.

\bibitem{DBLP:conf/sigada/Leino12}
K.~Rustan~M. Leino.
\newblock Developing verified programs with dafny.
\newblock In {\em Proceedings of the 2012 {ACM} Conference on High Integrity
  Language Technology, {HILT} '12, December 2-6, 2012, Boston, Massachusetts,
  {USA}}, pages 9--10, 2012.
\newblock \href {http://dx.doi.org/10.1145/2402676.2402682}
  {\path{doi:10.1145/2402676.2402682}}.

\bibitem{leino2004object}
K~Rustan~M Leino and Peter M{\"u}ller.
\newblock Object invariants in dynamic contexts.
\newblock In {\em European Conference on Object-Oriented Programming}, pages
  491--515. Springer, 2004.
\newblock \href {http://dx.doi.org/10.1007/978-3-540-24851-4_22}
  {\path{doi:10.1007/978-3-540-24851-4_22}}.

\bibitem{DBLP:conf/ecoop/LeinoM04}
K.~Rustan~M. Leino and Peter M{\"{u}}ller.
\newblock Object invariants in dynamic contexts.
\newblock In {\em {ECOOP} 2004 - Object-Oriented Programming, 18th European
  Conference, Oslo, Norway, June 14-18, 2004, Proceedings}, pages 491--516,
  2004.
\newblock \href {http://dx.doi.org/10.1007/978-3-540-24851-4\_22}
  {\path{doi:10.1007/978-3-540-24851-4\_22}}.

\bibitem{DBLP:conf/vstte/LeinoMW08}
K.~Rustan~M. Leino, Peter M{\"{u}}ller, and Angela Wallenburg.
\newblock Flexible immutability with frozen objects.
\newblock In {\em Verified Software: Theories, Tools, Experiments, Second
  International Conference, {VSTTE} 2008, Toronto, Canada, October 6-9, 2008.
  Proceedings}, pages 192--208, 2008.
\newblock \href {http://dx.doi.org/10.1007/978-3-540-87873-5\_17}
  {\path{doi:10.1007/978-3-540-87873-5\_17}}.

\bibitem{Leino2004ExceptionSF}
K.~Rustan~M. Leino and Wolfram Schulte.
\newblock Exception safety for c\#.
\newblock {\em Proceedings of the Second International Conference on Software
  Engineering and Formal Methods, 2004. SEFM 2004.}, pages 218--227, 2004.

\bibitem{matsakis2014rust}
Nicholas~D Matsakis and Felix~S Klock~II.
\newblock The rust language.
\newblock In {\em ACM SIGAda Ada Letters}, volume~34, pages 103--104. ACM,
  2014.
\newblock \href {http://dx.doi.org/10.1145/2663171.2663188}
  {\path{doi:10.1145/2663171.2663188}}.

\bibitem{meredith2012overview}
Patrick~O'Neil Meredith, Dongyun Jin, Dennis Griffith, Feng Chen, and Grigore
  Ro{\c{s}}u.
\newblock An overview of the mop runtime verification framework.
\newblock {\em International Journal on Software Tools for Technology
  Transfer}, 14(3):249--289, 2012.
\newblock \href {http://dx.doi.org/10.1007/s10009-011-0198-6}
  {\path{doi:10.1007/s10009-011-0198-6}}.

\bibitem{Meyer:1988:OSC:534929}
Bertrand Meyer.
\newblock {\em Object-Oriented Software Construction}.
\newblock Prentice-Hall, Inc., Upper Saddle River, NJ, USA, 1st edition, 1988.

\bibitem{Meyer:1992:EL:129093}
Bertrand Meyer.
\newblock {\em Eiffel: The Language}.
\newblock Prentice-Hall, Inc., Upper Saddle River, NJ, USA, 1992.

\bibitem{meyer2016class}
Bertrand Meyer.
\newblock Class invariants: Concepts, problems, solutions.
\newblock {\em arXiv preprint arXiv:1608.07637}, 2016.

\bibitem{miller2003capability}
Mark~S Miller, Ka-Ping Yee, Jonathan Shapiro, et~al.
\newblock Capability myths demolished.
\newblock Technical report, Technical Report SRL2003-02, Johns Hopkins
  University Systems Research Laboratory, 2003. http://www. erights.
  org/elib/capability/duals, 2003.

\bibitem{RobustComposition}
Mark~Samuel Miller.
\newblock {\em Robust Composition: Towards a Unified Approach to Access Control
  and Concurrency Control}.
\newblock PhD thesis, Johns Hopkins University, Baltimore, Maryland, USA, May
  2006.

\bibitem{ahmed20073}
Greg Morrisett, Amal Ahmed, and Matthew Fluet.
\newblock L 3: a linear language with locations.
\newblock In {\em International Conference on Typed Lambda Calculi and
  Applications}, pages 293--307. Springer, 2005.
\newblock \href {http://dx.doi.org/10.1007/11417170_2}
  {\path{doi:10.1007/11417170_2}}.

\bibitem{muller2002modular}
Peter M{\"u}ller.
\newblock {\em Modular specification and verification of object-oriented
  programs}.
\newblock Springer-Verlag, 2002.
\newblock \href {http://dx.doi.org/10.1007/3-540-45651-1}
  {\path{doi:10.1007/3-540-45651-1}}.

\bibitem{DBLP:journals/scp/MullerPL06}
Peter M{\"{u}}ller, Arnd Poetzsch{-}Heffter, and Gary~T. Leavens.
\newblock Modular invariants for layered object structures.
\newblock {\em Sci. Comput. Program.}, 62(3):253--286, 2006.
\newblock URL: \url{https://doi.org/10.1016/j.scico.2006.03.001}, \href
  {http://dx.doi.org/10.1016/j.scico.2006.03.001}
  {\path{doi:10.1016/j.scico.2006.03.001}}.

\bibitem{DBLP:journals/tcs/NaumannB06}
David~A. Naumann and Michael Barnett.
\newblock Towards imperative modules: Reasoning about invariants and sharing of
  mutable state.
\newblock {\em Theor. Comput. Sci.}, 365(1-2):143--168, 2006.
\newblock URL: \url{https://doi.org/10.1016/j.tcs.2006.07.035}, \href
  {http://dx.doi.org/10.1016/j.tcs.2006.07.035}
  {\path{doi:10.1016/j.tcs.2006.07.035}}.

\bibitem{noble2016abstract}
James Noble, Sophia Drossopoulou, Mark~S Miller, Toby Murray, and Alex Potanin.
\newblock Abstract data types in object-capability systems.
\newblock 2016.

\bibitem{OstlundEtAl08}
Johan {\"{O}}stlund, Tobias Wrigstad, Dave Clarke, and Beatrice {\AA}kerblom.
\newblock Ownership, uniqueness, and immutability.
\newblock In Richard~F. Paige and Bertrand Meyer, editors, {\em International
  Conference on Objects, Components, Models and Patterns}, volume~11 of {\em
  Lecture Notes in Computer Science}, pages 178--197. Springer, 2008.
\newblock \href {http://dx.doi.org/10.1007/978-3-540-69824-1_11}
  {\path{doi:10.1007/978-3-540-69824-1_11}}.

\bibitem{parkinson2007class}
Matthew Parkinson.
\newblock Class invariants: The end of the road?
\newblock {\em Aliasing, Confinement and Ownership in Object-oriented
  Programming (IWACO)}, page~9, 2007.

\bibitem{pearce2011jpure}
David Pearce.
\newblock Jpure: a modular purity system for java.
\newblock In {\em Compiler construction}, pages 104--123. Springer, 2011.
\newblock \href {http://dx.doi.org/10.1007/978-3-642-19861-8_7}
  {\path{doi:10.1007/978-3-642-19861-8_7}}.

\bibitem{pierce2002types}
Benjamin~C Pierce.
\newblock {\em Types and programming languages}.
\newblock MIT press, 2002.

\bibitem{DBLP:conf/fm/PolikarpovaTFM14}
Nadia Polikarpova, Julian Tschannen, Carlo~A. Furia, and Bertrand Meyer.
\newblock Flexible invariants through semantic collaboration.
\newblock In {\em {FM} 2014: Formal Methods - 19th International Symposium,
  Singapore, May 12-16, 2014. Proceedings}, pages 514--530, 2014.
\newblock \href {http://dx.doi.org/10.1007/978-3-319-06410-9\_35}
  {\path{doi:10.1007/978-3-319-06410-9\_35}}.

\bibitem{Potanin2013}
Alex Potanin, Johan {\"O}stlund, Yoav Zibin, and Michael~D. Ernst.
\newblock {\em Immutability}, pages 233--269.
\newblock Springer Berlin Heidelberg, Berlin, Heidelberg, 2013.
\newblock \href {http://dx.doi.org/10.1007/978-3-642-36946-9_9}
  {\path{doi:10.1007/978-3-642-36946-9_9}}.

\bibitem{ServettoEtAl13a}
Marco Servetto, David~J. Pearce, Lindsay Groves, and Alex Potanin.
\newblock Balloon types for safe parallelisation over arbitrary object graphs.
\newblock In {\em WODET 2014 - Workshop on Determinism and Correctness in
  Parallel Programming}, 2013.
\newblock \href {http://dx.doi.org/doi=10.1.1.353.2449}
  {\path{doi:doi=10.1.1.353.2449}}.

\bibitem{ServettoZucca15}
Marco Servetto and Elena Zucca.
\newblock Aliasing control in an imperative pure calculus.
\newblock In Xinyu Feng and Sungwoo Park, editors, {\em Programming Languages
  and Systems - 13th Asian Symposium (APLAS)}, volume 9458 of {\em Lecture
  Notes in Computer Science}, pages 208--228. Springer, 2015.
\newblock \href {http://dx.doi.org/10.1007/978-3-319-26529-2_12}
  {\path{doi:10.1007/978-3-319-26529-2_12}}.

\bibitem{Smith:2000:AT:645394.651903}
Frederick Smith, David Walker, and J.~Gregory Morrisett.
\newblock Alias types.
\newblock In {\em Proceedings of the 9th European Symposium on Programming
  Languages and Systems}, ESOP '00, pages 366--381, London, UK, UK, 2000.
  Springer-Verlag.
\newblock URL: \url{http://dl.acm.org/citation.cfm?id=645394.651903}.

\bibitem{stephens2015beginning}
R.~Stephens.
\newblock {\em Beginning Software Engineering}.
\newblock Wiley, 2015.

\bibitem{DBLP:conf/oopsla/StricklandTFF12}
T.~Stephen Strickland, Sam Tobin{-}Hochstadt, Robert~Bruce Findler, and Matthew
  Flatt.
\newblock Chaperones and impersonators: run-time support for reasonable
  interposition.
\newblock In {\em Proceedings of the 27th Annual {ACM} {SIGPLAN} Conference on
  Object-Oriented Programming, Systems, Languages, and Applications, {OOPSLA}
  2012, part of {SPLASH} 2012, Tucson, AZ, USA, October 21-25, 2012}, pages
  943--962, 2012.
\newblock \href {http://dx.doi.org/10.1145/2384616.2384685}
  {\path{doi:10.1145/2384616.2384685}}.

\bibitem{FlexibleInvariants}
Alexander~J. Summers, Sophia Drossopoulou, and Peter M\"{u}ller.
\newblock The need for flexible object invariants.
\newblock In {\em International Workshop on Aliasing, Confinement and Ownership
  in Object-Oriented Programming}, IWACO '09, pages 6:1--6:9, New York, NY,
  USA, 2009. ACM.
\newblock \href {http://dx.doi.org/10.1145/1562154.1562160}
  {\path{doi:10.1145/1562154.1562160}}.

\bibitem{takikawa2015towards}
Asumu Takikawa, Daniel Feltey, Earl Dean, Matthew Flatt, Robert~Bruce Findler,
  Sam Tobin-Hochstadt, and Matthias Felleisen.
\newblock Towards practical gradual typing.
\newblock In {\em LIPIcs-Leibniz International Proceedings in Informatics},
  volume~37. Schloss Dagstuhl-Leibniz-Zentrum fuer Informatik, 2015.
\newblock \href {http://dx.doi.org/10.4230/LIPIcs.ECOOP.2015.4}
  {\path{doi:10.4230/LIPIcs.ECOOP.2015.4}}.

\bibitem{DBLP:conf/oopsla/TakikawaSDTF12}
Asumu Takikawa, T.~Stephen Strickland, Christos Dimoulas, Sam
  Tobin{-}Hochstadt, and Matthias Felleisen.
\newblock Gradual typing for first-class classes.
\newblock In {\em Proceedings of the 27th Annual {ACM} {SIGPLAN} Conference on
  Object-Oriented Programming, Systems, Languages, and Applications, {OOPSLA}
  2012, part of {SPLASH} 2012, Tucson, AZ, USA, October 21-25, 2012}, pages
  793--810, 2012.
\newblock \href {http://dx.doi.org/10.1145/2384616.2384674}
  {\path{doi:10.1145/2384616.2384674}}.

\bibitem{TschantzErnst05}
Matthew~S. Tschantz and Michael~D. Ernst.
\newblock Javari: Adding reference immutability to {J}ava.
\newblock In {\em ACM SIGPLAN Conference on Object-Oriented Programming,
  Systems, Languages and Applications (OOPSLA 2005)}, pages 211--230. ACM
  Press, 2005.
\newblock \href {http://dx.doi.org/10.1145/1094811.1094828}
  {\path{doi:10.1145/1094811.1094828}}.

\bibitem{Voigt2013}
Janina Voigt, Warwick Irwin, and Neville Churcher.
\newblock {\em Comparing and Evaluating Existing Software Contract Tools},
  pages 49--63.
\newblock Springer Berlin Heidelberg, Berlin, Heidelberg, 2013.
\newblock \href {http://dx.doi.org/10.1007/978-3-642-32341-6_4}
  {\path{doi:10.1007/978-3-642-32341-6_4}}.

\bibitem{WikiInvariant}
{Wikipedia contributors}.
\newblock Class invariant, 2018.
\newblock URL: \url{https://en.wikipedia.org/wiki/Class_invariant}.

\bibitem{DBLP:conf/popl/WrigstadNLOV10}
Tobias Wrigstad, Francesco~Zappa Nardelli, Sylvain Lebresne, Johan
  {\"{O}}stlund, and Jan Vitek.
\newblock Integrating typed and untyped code in a scripting language.
\newblock In {\em Proceedings of the 37th {ACM} {SIGPLAN-SIGACT} Symposium on
  Principles of Programming Languages, {POPL} 2010, Madrid, Spain, January
  17-23, 2010}, pages 377--388, 2010.
\newblock \href {http://dx.doi.org/10.1145/1706299.1706343}
  {\path{doi:10.1145/1706299.1706343}}.

\bibitem{ZibinEtAl10}
Yoav Zibin, Alex Potanin, Paley Li, Mahmood Ali, and Michael~D. Ernst.
\newblock Ownership and immutability in generic {J}ava.
\newblock In {\em ACM SIGPLAN Conference on Object-Oriented
  Programming,Systems, Languages and Applications (OOPSLA 2010)}, pages
  598--617, 2010.
\newblock \href {http://dx.doi.org/10.1145/1869459.1869509}
  {\path{doi:10.1145/1869459.1869509}}.

\end{thebibliography}
\appendix
\clearpage
\section{Proof and Axioms}
\label{s:proof}
\lstset{morekeywords={fwd}}

\subheading{Axiomatic Type Properties}
As previously discussed, instead of providing a concrete set of type rules, we provide a set of properties
that the type system needs to respect.
To express these properties, we first need some auxiliary definitions.

%\noindent\textbf{Define}
%$\mathit{encapsulatedObj}(C)$:\\*
%${}_{}$\quad\quad \Q@class @$C$\,\Q@implements @$\Many{C}$\Q@{@$\,\Many{F}\,\Many{M}$\Q@}@
% and $\forall \mdf\,C\,\f \in \Many{F},\ \mdf \in \{\Kw{imm},\Kw{capsule}\}$\\*
%\noindent As we discussed, only encapsulated objects can support invariants;
%their class declarations only have immutable or capsule fields. Note how here we see immutable
%and simple objects as special cases of encapsulated ones.

The encapsulated ROG of $l_0$ is composed of all the objects
in the ROG of its immutable and capsule fields:\\*
\indent $l \in \mathit{erog}(\sigma,l_0)
\text{ iff } \exists f, \Sigma^\sigma(l_0).f \in \{\Kw{imm}\,\_,\Kw{capsule}\,\_\}
\text{ and } l \in \mathit{rog}(\sigma,\sigma(l_0).f)$% \loseSpace

\noindent An object is \emph{mutatable} in a $\sigma$ and  $\e$ if there is an occurrence of 
$l$ in $e$, that when seen as \Q@imm@ makes the expression ill typed:\\*
\indent $\mathit{mutatable}(l,\sigma,\e)$ iff for some $T=\Kw{imm}\,\Sigma^\sigma(l)$ and $\ctx[l]=\e$,\\*
\indent \indent $\Sigma^\sigma;\x:T\vdash\ctx[\x]:T'$ does not hold for any $T'$.%\loseSpace

%if $\ \sigma_0|e_0\rightarrow \sigma|e$ then $\sigma_1|\e_1=\sigma|\e$
% $\exists! \sigma_1|\e_1$ such that $\sigma_0|\e_0\rightarrow \sigma_1|\e_1$\\*

%We can now assume the following properties over the type system:

Here we assume the usual \thm{Progress} and \thm{Subject Reduction Base}. Note that \thm{Subject Reduction Base} only ensures properties about type checking, not invariant checking.\saveSpace\saveSpace
\begin{Assumption}[Progress]\rm
	if $\Sigma^{\sigma_0};\emptyset\vdash e_0: T_0$,
	and $e_0$ is not form $l$ or $\mathit{error}$, then
	$\sigma_0|e_0\rightarrow \sigma_1|e_1$.
\end{Assumption}

\begin{Assumption}[Subject Reduction Base]\rm
	if $\Sigma^{\sigma_0};\emptyset\vdash e_0: T_0$,
	$\sigma_0|e_0\rightarrow \sigma_1|e_1$,
	then
	$\Sigma^{\sigma_1};\emptyset\vdash e_1: T_1$.
\end{Assumption}

If the result of a field access is \Q!mut!,
the receiver is also \Q!mut!; field updates are only allowed on \Q!mut! receivers:\saveSpace\saveSpace
\begin{Assumption}[Mut Field]\rm
	\ \\
	\indent(1)\ if $\Sigma;\Gamma\vdash\e\singleDot\f:\Kw{mut}\,\_$
	then $\Sigma;\Gamma\vdash\e:\Kw{mut}\,\_$
	and 
	\\*\indent(2)
	if $\Sigma;\Gamma\vdash\e_0\singleDot\f\equals\e_1:T$
	then $\Sigma;\Gamma\vdash\e_0:\Kw{mut}\,\_$.
\end{Assumption}

An object is not part of the ROG of its immutable or capsule fields\footnote{This is not strictly true in L42, as 42 allows circular objects with \Q!fwd imm! fields, however such objects cannot have an invariant.}:\saveSpace\saveSpace
\begin{Assumption}[Head Not Circular]\rm
	if
	$\Sigma^\sigma;\Gamma\vdash l:T$,
	then $l\notin\text{erog}(\sigma,l)$.
\end{Assumption}

In a well typed $\sigma$ and $e$, if mutatable $l_2$ is reachable  through the \emph{erog} of
$l_1$, and $l_1$ is reachable through the \emph{erog} of $l_0$,
then all the paths connecting $l_0$ and $l_2$ pass trough $l_1$; thus
if we were to remove $l_1$ from the object graph, $l_0$ would no longer reach $l_2$:
\saveSpace\saveSpace
\begin{Assumption}[Capsule Tree]\rm
	If   $\Sigma^\sigma;\Gamma\vdash \e:\T$,
	$l_2\in\text{erog}(\sigma,l_1)$,
	$l_1\in\text{erog}(\sigma,l_0)$,\\*
	and
	$\mathit{mutatable}(l_2,\sigma,\e)$
	then 
	$l_2\notin\text{erog}(\sigma\setminus l_1,l_0)$.
\end{Assumption}

\thm{Capsule Tree} and \thm{Head Not Circular} together 
imply that capsule fields section the object graph into a tree of nested `balloons',
where nodes are mutable encapsulated objects and
edges are given by reachability between those objects in the original memory: if
$l_2$ is in the encapsulated ROG of $l_1$, and
$l_2$ is mutatable and reachable through $l_1$, then $l_2$ must be reachable by a \Q@capsule@ field.
Thanks to \thm{Head Not Circular} and $l_1\in\text{erog}(\sigma,l_0)$ we can derive that
$l_0\notin\text{erog}(\sigma,l_1)$.

The execution of an expression
with no \Q@mut@ free variables is deterministic and does not
mutate pre existing memory (and thus does not not perform I/O by mutating the pre existing $c$):\saveSpace\saveSpace
\begin{Assumption}[Determinism]\rm
	if $\emptyset;\Gamma\vdash \e:\T$, 
	$\forall x\, ( \Gamma(x)\neq\Kw{mut}\,\_$), and
	$\sigma | \e'\rightarrow^+ \sigma' | \e''$
	then 
	$\sigma | \e'\Rightarrow^+ \sigma,\_ | \e''$,
	where $\e'=\e[x_1=l_1,\ldots,x_n=l_n]$ and $\Sigma^\sigma;\emptyset\vdash \e':\T$
\end{Assumption}

For each \Q@try@--\Q@catch@, execution preserves the memory needed to continue the execution in case of an error (the memory visible outside of the \Q@try@).\saveSpace\saveSpace
\begin{Assumption}[Strong Exception Safety]\rm
	if $\Sigma^{\sigma,\sigma'};\emptyset\vdash \ctx[\Kw{try}^\sigma\oC\e_0\cC\ \Kw{catch}\ \oC\e_1\cC]:\T$
	and\\*
	$
	\sigma,\sigma'|\ctx[\Kw{try}^\sigma\oC\e_0\cC\ \Kw{catch}\ \oC\e_1\cC]\rightarrow 
	\sigma''|\ctx[\Kw{try}^\sigma\oC\e'\cC\ \Kw{catch}\ \oC\e_1\cC]
	$
	then 
	$\sigma''=\sigma,\_$
	and
	$\Sigma^\sigma;\emptyset\vdash \ctx[\e_1]:\T$
\end{Assumption}

%Thanks to how our reduction rules are designed, especially \textsc{try error},
%@Progress will need to rely on @StrongExceptionSafety internally.

Note that our last well formedness rule requires 
\textsc{update} and \textsc{mcall} to introduce
monitor expressions only over locations
that are not preserved by \Q@try@ blocks.
This can be achieved, since monitors are introduced
around mutation operations
(and \Q@new@ expression),
and \thm{Strong Exception Safety} ensures no mutation happens on preserved memory.

% To the best of our knowledge, only the type system of 42~\cite{ServettoEtAl13a,ServettoZucca15}
%  supports all these assumptions out of the box,
% while both Gordon~\cite{GordonEtAl12} and Pony~\cite{clebsch2015deny,clebsch2017orca} supports all except StrongExceptionSafety,
% however it should be trivial to modify them to support it:
% the \Q@try-catch@ rule could be modified to
% $\emptyset;\Gamma\vdash\Kw{try}\ \oC\e_0\cC\ \Kw{catch}\ \oC\e_1\cC:\T$
% if\\* $\emptyset;
% \Gamma,\{x:\Kw{read}\,C | x:\Kw{mut}\,C\,\in\Gamma\}
% \vdash\e_0:\T$ and $\emptyset;\Gamma\vdash\e_1:\T$,
% i.e. $e_0$ can be typed when seeing all externally defined mutable references as \Q@read@.

\subheading{Proof of Soundness}
It is hard to prove \thm{Soundness} directly,
so we first define a stronger property,
called \thm{Stronger Soundness}, and
show that it is preserved during reductions by means of conventional
\thm{Progress} and \thm{Subject Reduction} (\thm{Progress} is one of our assumptions,
while \thm{Subject Reduction} relies heavily upon \thm{Subject Reduction Base}).
That is:
\begin{itemize}
\item \thm{Progress} $\wedge$ \thm{Subject Reduction} $\Rightarrow$ \thm{Stronger Soundness}, and
\item \thm{Stronger Soundness} $\Rightarrow$ \thm{Soundness}.
\end{itemize}
%The structure of the proof is interesting:
%It is hard to prove Sound Validation directly,
%so we first define a stronger property,
%called Stronger Sound Validation and
%we show that it is preserved during reduction by mean of conventional Progress and Subject Reduction.
%That is,
%Progress+Subject Reduction $\Rightarrow$ Stronger Sound Validation
%and Stronger Sound Validation $\Rightarrow$ Sound Validation.

\subheading{Stronger Soundness $\Rightarrow$ Soundness}
\thm{Stronger Soundness} depends on
$\mathit{wellEncapsulated}$, $\mathit{monitored}$
and \emph{OK}:\\
%\loseSpace%
\indent $\mathit{wellEncapsulated}(\sigma,\e,l_0)$ iff
$\forall l \in \mathit{erog}(\sigma,l_0), \text{not}\ \mathit{mutatable}(l,\sigma,\e)$.%\loseSpace

\noindent The main idea is that an object is well encapsulated if its encapsulated state cannot be
modified by $e$.

An object is \emph{monitored} if execution
is currently inside of a monitor for that object, and
the monitored expression $\e_1$ does not
contain $l$ as a \emph{proper} subexpression:

\indent $\mathit{monitored}(\e,l)$ iff
$\e=\ctx_v[\M{l}{\e_1}{\e_2}]$ and either $\e_1=l$, or $l$ is not inside $\e_1$.%\loseSpace

\noindent A monitored object is associated with an expression that can not observe it, but may
reference its internal representation directly.
In this way, we can safely modify its representation before checking its invariant.

The idea is that at the start the object will be valid and $\e_1$ will reference $l$;
but during reduction, $l$ will be used to
modify the object; only after that moment, the object may become invalid.

\noindent\textbf{Define} $\mathit{OK}(\sigma,e)$:\\
\indent $\forall l\in\textit{dom}(\sigma)$
  either\\
\indent\indent 1. $\mathit{garbage}(l,\sigma,\e)$,\\
\indent\indent 2. $\mathit{valid}(\sigma,l)$ and $\mathit{wellEncapsulated}(\sigma,\e,l)$, or\\
\indent\indent 3. $\mathit{monitored}(\e,l)$.

\noindent Finally, the system is in an \emph{OK} state
if all objects in memory, are either
%the class of the object has no invariant method;
not (transitively) reachable from the expression (thus can be garbage collected),
valid and encapsulated,
or currently monitored.

\begin{theorem}[Stronger Soundness]\rm
if $c:\Kw{Cap};\emptyset\vdash \e_0: \T_0$ and
$c\mapsto\Kw{Cap}\{\_\}|\e_0\rightarrow^+ \sigma|\e$, then
$\emph{OK}(\sigma,\e)$.
\end{theorem}
\noindent Starting from only the capability object,
any well typed expression $\e_0$ can be reduced in an arbitrary number of steps,
and \emph{OK} will always hold.

\begin{theorem}\rm \thm{Stronger Soundness} $\Rightarrow$ \thm{Soundness}
\end{theorem}
\begin{proof}
\noindent By \thm{Stronger Soundness}, each $l$ in the current redex must be \emph{OK}:
\begin{enumerate}
	\item If $l$ is \emph{garbage}, it cannot be in the current redex, a contradiction.
	\item If $\mathit{valid}(\sigma,l)$, then $l$ is valid, so thanks to \thm{Determinism}
	no invalid object could be observed.
	\item Otherwise, if $\mathit{monitored}(\e,l)$ then either:
	\begin{itemize}
	 \item we are executing inside of $\e_1$, thus the current redex is inside of a sub expression of the monitor that does not contain $l$, a contradiction.
	 \item or we are executing inside $\e_2$:
	 by our reduction rules, all monitor expressions start with
	 $\e_2=l.\invariant$, thus the first execution step
	 of $\e_2$ is trusted. Further execution steps are also trusted, since by well formedness the body of invariant methods only use \Q@this@ (now replaced with $l$) to read fields.
	\end{itemize}
\end{enumerate}
In any of the possible cases above, \thm{Soundness} holds for $l$, and so it holds for all redexes.
\end{proof}

\subheading{Subject Reduction}
\noindent\textbf{Define} $\mathit{fieldGuarded}(\sigma,\e)$:\\*
\indent$\forall \ctx$ such that $\e=\ctx[l\singleDot\f] $
and $\Sigma^\sigma(l).f=\Kw{capsule}\,\_$, and $\f\mathrel{\mathit{inside}} \Sigma^\sigma(l).\invariant$\\*
\indent\indent either
$\forall T, \forall C, \Sigma^\sigma;\x:\Kw{mut}\,C\,\not\vdash\ctx[\x]:T$, or\\*
\indent\indent $\ctx=\ctx'[\M{l}{\ctx''}{\e}]$ and $l$ is contained exactly once in $\ctx''$.

That is, all \Q@mut@ capsule field accesses are individually guarded by monitors.
Note how we use $C$ in $\x:\Kw{mut}\,C$ to guess the type of the accessed field,
and that we use the full context $\ctx$, instead of the evaluation context $\ctx_v$,
to refer to field accesses everywhere in the expression $\e$.

\begin{theorem}[Subject Reduction]\rm
if $\Sigma^{\sigma_0};\emptyset\vdash e_0: T_0$,
$\sigma_0|e_0\rightarrow \sigma_1|e_1$,
$\mathit{OK}(\sigma_0,\e_0)$
and
$\mathit{fieldGuarded}(\sigma_0,\e_0)$
then
$\Sigma^{\sigma_1};\emptyset\vdash e_1: T_1$,
$\mathit{OK}(\sigma_1,e_1)$ and
$\mathit{fieldGuarded}(\sigma_1,\e_1)$
\end{theorem}

\begin{theorem}\rm
	\thm{Progress} + \thm{Subject Reduction} $\Rightarrow$ \thm{Stronger Soundness}
\end{theorem}
\begin{proof}
This proof proceeds by induction in the usual manner.

\emph{Base case}: At the start of execution, memory only contains $c$: since $c$ is defined to always be $\mathit{valid}$, and has only \Q@mut@ fields, it is trivially $\mathit{wellEncapsulated}$, thus $\mathit{OK}(c\mapsto\Kw{Cap},e)$.

\emph{Induction}: By \thm{Progress}, we always have another evaluation step to take, by \thm{Subject Reduction} such a step will preserve $\mathit{OK}$, and so by induction, $\mathit{OK}$ holds after any number of steps.

Note how for the proof garbage collectability is important:
when the \Q@invariant()@ method evaluates to \Q@false@,
execution can continue only if the offending object is classified as \emph{garbage}.
\end{proof}

\subheading{Exposer Instrumentation}
We first introduce a lemma derived from our well formedness criteria and the type system:
\begin{Lemma}[Exposer Instrumentation]\rm
If $\sigma_0 | \e_0\rightarrow \sigma_1 |\e_1$ and
$\text{fieldGuarded}(\sigma_0,\e_0)$
\\*
then $\text{fieldGuarded}(\sigma_1,\e_1)$.
\end{Lemma}
\begin{proof}
The only rule that can
introduce a new field access is \textsc{mcall}.
In that case, \thm{Exposer Instrumentation} holds
by well formedness (all field accesses in methods are of the form \Q@this.f@),
since \textsc{mcall} inserts a monitor while invoking capsule mutator methods, and not field accesses themselves. If however the method is not a \Q@mut@ method but still accesses a capsule field, by \thm{Mut Field} such a field access expression cannot be typed as \Q@mut@ and so no monitor is needed.

Note that \textsc{monitor exit} is fine because monitors are removed only when
 $e_1$ is a value.
\end{proof}

\subheading{Proof of Subject Reduction}%\saveSpace
Any reduction step can be obtained
by exactly one application of the \textsc{ctxv} rule and one other rule. Thus the proof can simply proceed by cases on the other applied rule.

By \thm{Subject Reduction Base} and \thm{Exposer Instrumentation},
$\Sigma^{\sigma_1};\emptyset\vdash e_1: T_1$ and  $\mathit{fieldGuarded}(\sigma_1,\e_1)$. So we just need to proceed by cases on the reduction rule applied to verify that $\mathit{OK}(\sigma_1,\e_1)$ holds:
\begin{enumerate}
\item (\textsc{update}) $\sigma|l\singleDot f\equals v\rightarrow \sigma'|\e'$:
	\begin{itemize}
	  \item By \textsc{update} $\e'=\M{l}{l}{l}\singleDot\Kw{invariant}\oR\cR\cR$, thus $\mathit{monitored}(\e,l)$.
	  \item Every $l_1$ such that $l\in \mathit{rog}(\sigma,l_1)$ will verify the same case as the former step:
	  \begin{itemize}
	  	\item If it was $\mathit{garbage}$, clearly it still is.
	  	\item If it was $\mathit{monitored}$, it still is.
	    \item Otherwise it was $\mathit{valid}$ and $\mathit{wellEncapsulated}$:
			\begin{itemize}
				\item If $l\in \mathit{erog}(\sigma,l_1)$ we have a contradiction since $\mathit{mutatable}(l, \sigma, e)$, (by \thm{Mut Field})
		    	\item Otherwise, by our well formedess criteria that \Q@invariant@ only accesses \Q@imm@ and \Q@capsule@ fields, and by \thm{Determinism}, it is clearly the case that $\mathit{valid}$ still holds;
				By \thm{Head Not Circular} it cannot be the case that $l\in \mathit{erog}(\sigma',l_1)$, and so $l_1$ is still $\mathit{wellEncapsulated}$.
		  	\end{itemize}
	  \end{itemize}
	  \item Every other $l_0$ is not in the reachable object graph of $l$,
	  thus it being $\mathit{OK}$ could not have been effected by this reduction step.
	\end{itemize}

\item (\textsc{access}) $\sigma|l\singleDot f \rightarrow \sigma|v$:
	\begin{itemize}
		\item If $l$ was \emph{valid} and \emph{wellEncapsulated}:
		\begin{itemize}
			\item If we have now broken \emph{wellEncapsulated}, we must have made something in its \emph{erog} \emph{mutatable}. As we can only type \Q@capsule@ fields as \Q@mut@ and not \Q@imm@ fields, by \thm{Mut Field} we must have that $f$ is \Q@capsule@ and $l\singleDot f$ is being typed as \Q@mut@. By $\mathit{fieldGuarded}(\sigma_0,\e_0)$, the former step must have been inside a monitor \M{l}{\ctx_v[l.f]}{\e}
		    and the $l$ under reduction was the only occurrence of $l$.
		    Since $f$ is a capsule, we know that $l\notin \text{erog}(\sigma,l)$
		    by \thm{Head Not Circular}. Thus in our new step $l$ is not\emph{inside} $\ctx_v[v]$. Thus $l$ must be \emph{monitored} and hence it is $\mathit{OK}$.

		    \item Otherwise, $l$ is still $\mathit{OK}$
    	\end{itemize}

	\item Suppose some other $l_0$ was \emph{wellEncapsulated} and \emph{valid}:
	\begin{itemize}
			\item If $l$ was in the \emph{rog} of $l_0$, by \thm{Capsule Tree}, if $l$ was in the \emph{rog} of $l$, then $v$ can only be reached from $l_0$ by passing through $l$, and so we could not have made $l_0$ non \emph{wellEncapsulated}. In addition, since only things in the \emph{erog} can be referenced by \Q@invariant@, validity can not depend on $l$, and by \thm{Determinism} it is still the case that $l_0$ is \emph{valid}. And so we can't have effected $l_0$ being $\mathit{OK}$.
			\item Otherwise, this reduction step could not have affected $l_0$, so $l_0$ is still $\mathit{OK}$.
	\end{itemize}

	\item Nothing that was $\mathit{garbage}$ could have been made reachable by this expression, since the only value we produced was $v$ and it was reachable through $l$ (and so could not have been $\mathit{garbage}$), thus $l$ is still $\mathit{OK}$.

	\item As we don't change any monitors here, nothing that was $\mathit{monitored}$ could have been made un-$\mathit{monitored}$, and so it is still $\mathit{OK}$.
\end{itemize}

\item (\textsc{mcall}, \textsc{try enter} and \textsc{try ok}):

	These reduction steps do not modify memory, the memory locations reachable inside of main expression, or any monitor expressions. Therefore it cannot have any effect on the $\mathit{garbage}$, \emph{wellEncapsulated}, \emph{valid} (due to \thm{Determinism}), or $\mathit{monitored}$ properties of any memory locations, thus $\mathit{OK}$ still holds.

\item (\textsc{new}) $\sigma|\Kw{new}\ C\oR\vs\cR\rightarrow \sigma,l\mapsto C\{\vs\}| \M{l}{l}{l.\invariant}$:

	Clearly the newly created object, $l$, is \emph{monitored}. As for \textsc{mcall}, other objects and properties are not disturbed, and so $\mathit{OK}$ still holds.

\item (\textsc{monitor exit}) $\sigma|\M{l}{v}{\Kw{true}}\rightarrow \sigma|v$:
\begin{itemize}
	\item As monitor expressions are not present in the original source code, it must have been introduced by \textsc{update}, \textsc{mcall}, or \textsc{new}. In each case the 3\textsuperscript{rd} expression started of as $l\singleDot\Kw{invariant}\oR\cR$, and it has now (eventually) been reduced to $\Kw{true}$, thus by \thm{Determinism} $l$ is \emph{valid}.

	\item  If the monitor was introduced by \textsc{update}, then $v = l$. We must have had that $l$ was well encapsulated before \textsc{update} was executed (since it can't have been \emph{garbage} and \emph{monitored}, as \textsc{update} itself preserves this property and we haven't modified memory in anyway, we must still have that $l$ is \emph{wellEncapsulated}. As $l$ is \emph{valid} and \emph{wellEncapsulated}, it is $\mathit{OK}$.

	\item If the monitor was introduced by \textsc{mcall}, then it was due to calling a capsule mutator method that mutated a field $f$.
	\begin{itemize}
		\item A location that was \emph{garbage} obviously still is, and so is also $\mathit{OK}$.
		\item No location that was \emph{valid} could have been made invalid since this reduction rule performs no mutation of memory. If a location was \emph{wellEncapsulated} before, the only way it could be non \emph{wellEncapsulated} is if we somehow leaked a \Q@mut@ reference to something, but by our well-formedness rules, $v$ cannot be typed as \Q@mut@ and so we can't have affected \emph{wellEncapsulated}, hence such thing is still $\mathit{OK}$.
		\item The only location that could have been made un \emph{monitored} is $l$ itself. By our well formedness criteria, $l$ was only used to modify $l.f$, and we have no parameters by which we could have made $l.f$ non \emph{wellEncapsulated}, since that would violate \thm{Capsule Tree}. As nothing else in $l$ was modified, and it must have been \emph{wellEncapsulated} before the \textsc{mcall}, and so it still is. In addition since  $l$ is valid, it is $\mathit{OK}$.
	\end{itemize}
	\item Otherwise the monitor was introduced by \textsc{new}. Since we require that \Q@capsule@ fields and \Q@imm@ fields are only initialised to \Q@capsule@ and \Q@imm@ expressions, by \thm{Capsule Tree}, the resulting value, $l$, must be \emph{wellEncapsulated}, since $l$ is also \emph{valid} we have that $l$ is $\mathit{OK}$.

\end{itemize}

\item (\textsc{try error}) $\sigma,\sigma_0|\Kw{try}^\sigma\oC \mathit{error}\cC\ \Kw{catch}\ \oC\e\cC\rightarrow \sigma|\e$:

	By \thm{Strong Exception Safety}, we know that $\sigma_0$ is $\mathit{garbage}$ with respect to $\ctx_v[\e]$. By our well formedness criteria, no location inside $\sigma$ could have been \emph{monitored}.
	Since we don't modify memory, everything in $\sigma_0$ is $\mathit{garbage}$ and nothing inside $\sigma$ was previously monitored, it is still clearly the case that everything in $\sigma$ is $\mathit{OK}$. 
\end{enumerate}
\saveSpace
\section{The Hamster Example in Spec\#}\saveSpace
\label{s:hamster}
\lstset{language={[Sharp]C}, morekeywords={invariant,ensures,requires,expose,exists,capsule}}

In this section we describe exactly why we chose to annotate the example from Section~\ref{s:intro} in the way we did. For brevity, we will assume the default accessibility is \Q@public@, whilst in both Spec\# and C\#, it is actually \Q@private@.

\subheading{The \Q@Point@ Class} 
The typical way of writing a \Q@Point@ class in C\# is as follows:
\begin{lstlisting}
class Point {
	double x, y;
	Point(double x, double y) { this.x = x; this.y = y; }
}
\end{lstlisting}

This works exactly as is in Spec\#, however we have difficulty if we want to define equality of \Q@Point@s (see below).

\subheading{The \Q@Hamster@ Class} 
The \Q@Hamster@ class in C\# would simply be:
\begin{lstlisting}
class Hamster {
	Point pos;
	Hamster(Point pos) { this.pos = pos; }
}
\end{lstlisting}

Though this is legal in Spec\#, it is practically useless. Spec\# has no way of knowing whether \Q@pos@ is \emph{valid} or \emph{consistent}. If \Q@pos@ is not known to be valid, one will be unable to pass it to almost any method, since by default methods implicitly require their receivers and arguments to be valid (compare this with our invariant protocol, which guarantees that any reachable object is valid).
If \Q@pos@ is not known to be consistent, one will be unable to mutate it, by updating one of its fields or by passing it as an argument (or receiver) to a non \Q@Pure@ method.
Though we don't want \Q@pos@ to ever mutate, Spec\# currently has no way of enforcing that an \emph{instance} of a non immutable class is itself immutable\footnote{There is a the describes a simple solution to this problem: assign ownership of the object to a special predefined `freezer' object, which never gives up mutation permission~\cite{DBLP:conf/vstte/LeinoMW08}, however this does not appear to have been implemented; this would provide similar flexibility to the TM system we use, which allows an initially mutable object to be promoted to immutable.}, as such we will simply refrain from mutating it.

To enable Spec\# to reason about \Q@pos@'s validity, we will require that it be a \emph{peer} of the enclosing \Q@Hamster@; we can do this by annotating \Q@pos@ with \Q@[Peer]@. Peers are objects that have the same owner, implying that  whenever one is valid and/or consistent, the other one also is. This means that if we have a \Q@Hamster@, we can use its \Q@pos@, in the same ways as we could use the \Q@Hamster@.

To simplify instantiation of \Q@Hamster@s, their constructors will take unowned \Q@Point@s, Spec\# will then automatically make such \Q@Point@ a peer. This is achieved by taking a \Q@[Captured]@ \Q@Point@ in the constructor (note how similar this is to taking a \Q@capsule@ \Q@Point@). Note that unlike our system, this prevents multiple \Q@Hamster@s from sharing the same \Q@Point@, unless both \Q@Hamster@s have the same owner, if \Q@Point@ were immutable, there would be no such restriction.

With the aforementioned modifications, the \Q@Hamster@ becomes:
\begin{lstlisting}
class Hamster {
  [Peer] Point pos;
  Hamster([Captured] Point pos) { this.pos = pos; }
}
\end{lstlisting}

We don't want \Q@Point@ to be an immutable/value type, however if it were, the original unannotated version would not have any problems.

\subheading{The \Q@Cage@ Class} 
The natural way to write this class in C\#, if it had native support for class invariants like Spec\#, would be:
\begin{lstlisting}
class Cage {
  Hamster h;
  List<Point> path;
  Cage(Hamster h, List<Point> path){this.h=h; this.path=path;}
  invariant this.path.Contains(this.h.pos);
  void Move() { 
    int index = this.path.IndexOf(this.h.pos);
    this.h.pos = this.path[index % this.path.Count]; } 
}
\end{lstlisting}

However for the above \Q@invariant@ to be admissible in Spec\#, \Q@this.path@ and \Q@this.h@ must both be owned by \Q@this@. In addition, the \emph{elements} of \Q@this.path@ need to be owned by \Q@this@, since \Q@this.path.Conatains@ will read them. Note that \Q@this.h.pos@ also needs to be owned by \Q@this@, however since \Q@pos@ is declared as \Q@[Peer]@, if \Q@this@ owns \Q@this.h@, it also owns \Q@this.h.pos@. To fix the invariant, we will declare \Q@h@, \Q@path@, and the elements of \Q@path@ as \emph{reps} (i.e. they are owned by the containing object). Finally, since \Q@Move@ modifies \Q@this.h@, \Q@this.h@ needs to be made consistent, which requires that the owner (\Q@this@) be made invalid; this can be achieved by using an \Q@expose(this)@ statement. \Q@expose(this){@\emph{body}\Q@}@ marks \Q@this@ as invalid, executes \emph{body}, checks that the invariant of \Q@this@ holds, and then marks \Q@this@ valid again.
As we did with the \Q@Hamster@, we will simply take unowned \Q@h@ and \Q@path@ values, however we also need the elements of \Q@path@ to be unowned; since Spec\# has no \Q@[ElementsCaptured]@ annotation, we will require \Q@path@ to be unowned, and its elements (denoted by \Q@Owner.ElementProxy(path)@) to be owned by the same owner as \Q@path@ (which is \Q@null@).
\begin{lstlisting}
class Cage {
  [Rep] public Hamster h;
  [Rep, ElementsRep] List<Point> path;
	
  Cage([Captured] Hamster h, [Captured] List<Point> path)
    requires Owner.Same(Owner.ElementProxy(path), path);
  { this.h = h; this.path = path; }
	
  invariant this.path.Contains(this.h.pos);
  void Move() { 
    int index = this.path.IndexOf(this.h.pos);
    expose(this){this.h.pos=this.path[index%this.path.Count]; }} 
}
\end{lstlisting}

The above constructor now fails to verify, since Boogie is unconvinced that its precondition actually holds when we initialise \Q@this.path@. This is because the constructor for \Q@Object@ (the default base class if none is provided) is not marked as \Q@[Pure]@; since it is (implicitly) called upon entry to \Q@Cage@'s constructor, Boogie has no idea as to what memory could've mutated, and so it doesn't know whether the precondition still holds. The solution is to explicitly call it, but at the end of the constructor: \Q@{this.h = h; this.path = path; base();}@.

The above \Q@Cage@ code however does not work, since \Q@List@ operations, such as \Q@Contains@ and \Q@IndexOf@, will call the virtual \Q@Object.Equals@ method to compute equality of \Q@Point@s. However \Q@Object.Equals@ implements \emph{reference} equality, whereas we want \emph{value} equality.

\subheading{Defining Equality of \Q@Point@s}
The obvious solution in C\# is to just override \Q@Object.Equals@ accordingly, and let dynamic dispatch handle the rest:
\begin{lstlisting}
class Point {
  .. // as before
  override bool Equals(Object? o) {
    Point? that = o as Point;
    return that!=null && this.x == that.x && this.y == that.y;}
}
\end{lstlisting}

However this fails in Spec\# since \Q@Object.Equals@ is annotated with \Q@[Pure]@\\*\Q@[Reads(ReadsAttribute.Reads.Nothing)]@, and of course every overload of it must also satisfy this. The \Q@Reads@ annotations specifies that the method cannot read fields of \emph{any} object, not even the receiver, this makes overloading the method useless.
% Our best guess as to why \Q@Object.Equals@ is annotated like that is that they expect it to be the default reference-equality, annotating it like this could aid static verification as it implies that whether or not two objects are equal cannot change, even if their fields are modified.

We resorted to making our own \Q@Equal@ method. Since it is called in \Q@Cage@'s invariant, Spec\# requires it to be annotated as \Q@[Pure]@, and either annotated with\\* \Q@[Reads(ReadsAttribute.Reads.Nothing)]@ or \Q@[Reads(ReadsAttribute.Reads.Owned)]@\\* (the default, if the method is \Q@[Pure]@). The latter annotation means it can only read fields of objects owned by the \emph{receiver} of the method, so a \Q@[Pure] bool Equal(Point that)@ method can read the fields of \Q@this@, but not the fields of \Q@that@. Of course this would make the method unusable in \Q@Cage@ since the \Q@Point@s we are comparing equality against do not own each other. As such, the simplest solution is to pass the fields of the other point to the method:
\begin{lstlisting}
[Pure] bool Equal(double x, double y) {
  return x == this.x && y == this.y;}
\end{lstlisting}

Sadly however this mean we can no longer use \Q@List@'s \Q@Contains@ and \Q@IndexOf@ methods, rather we have to expand out their code manually; making these changes takes us to the version we presented in Section \ref{s:intro}.

\lstset{language=FortyTwo}

\lstset{morekeywords={expose}}

\section{More Case Studies}
\label{s:MoreCaseStudies}
\subheading{Family}
The following test case was designed to produce a worst case in the number of invariant checks. We have a \Q!Family! that (indirectly) contains a list of \Q!parents! and \Q!children!. The \Q!parents! and \Q!children! are of type \Q!Person!. Both \Q!Family! and \Q!Person! have an invariant, the invariant of \Q!Family! depends on its contained \Q!Person!s.

% TODO: Talk to mark about code style

% TODO: Swap parent/child updates in aritifact
\begin{lstlisting}
class Person { 
  final String name;
  Int daysLived;
  final Int birthday;
  Person(String name, Int daysLived, Int birthday) { .. }
  mut method Void processDay(Int dayOfYear) {
  	this.daysLived += 1;
    if (this.birthday == dayOfYear) {
    	Console.print("Happy birthday " + this.name + "!"); }}
  read method Bool invariant() {
    return !this.name.equals("") && this.daysLived >= 0 &&
      this.birthday >= 0 && this.birthday < 365; }
}
class Family { 
  static class Box { 
    mut List<Person> parents;
    mut List<Person> children;
    Box(mut List<Person> parents, mut List<Person> children){..}
    mut method Void processDay(Int dayOfYear) {
      for(Person c : this.children) { c.processDay(dayOfYear); }
      for(Person p : this.parents) { p.processDay(dayOfYear); }}
  }
  capsule Box box;
  Family(capsule List<Person> ps,capsule List<Person> cs) {
    this.box = new Box(ps, cs); }
  mut method Void processDay(Int dayOfYear) { 
    this.box.processDay(dayOfYear); }
  mut method Void addChild(capsule Person child) { 
    this.box.children.add(child); }
  read method Bool invariant() {
    for (Person p : this.box.parents) {
      for (Person c : this.box.children) {
        if (p.daysLived <= c.daysLived) { 
          return false; }}}
    return true; }
}
\end{lstlisting}
Note how we created a \Q!Box! class to hold the \Q!parents! and \Q!children!.
Thanks to this pattern, the invariant only needs to hold at the end of \Q!Family.processDay!, after all the \Q!parents! and \Q!children! have been updated. Thus \Q!Family.processDay! is atomic: it updates all its contained \Q!Person!s together.
%This capture the intention of consistently calling \Q@Person.processDay@ once for all the persons as an atomic operation.
%This capture the intention of atomically and consistently calling \Q@Person.processDay@ once for all the persons.
Had we instead made the \Q!parents! and \Q!children! \Q!capsule! fields of \Q!Family!, the invariant would be required to also hold between modifying the two lists. This could cause problems if, for example, a child was updated before their parent.

\loseSpace
We have a simple test case that calls \Q!processDay! on a \Q!Family! $1{,}095$ ($3\times365$) times.
%$3\times365$ times ($1{,}095$ days):
\begin{lstlisting}
// 2 parents (one 32, the other 34), and no children
var fam = new Family(List.of(new Person("Bob", 11720, 40),
    new Person("Alice", 12497, 87)), List.of());
    
for (Int day = 0; day < 365; day++) { // Run for 1 year
  fam.processDay(day);
}
for (Int day = 0; day < 365; day++) { // The next year
  fam.processDay(day);
  if (day == 45) {
    fam.addChild(new Person("Tim", 0, day)); }}

for (Int day = 0; r < 365; day++) { // The 3rd year
  fam.processDay(day);
  if (day == 340) {
    fam.addChild(new Person("Diana", 0, day)); }}
\end{lstlisting}
% The counts (including the invariant keyword, and read on the invariant method)
% Spec# 14 family, 2 main   = 16
% L42 	12 family, 1 person = 14
% Fake 42 = 14 (+2 for box ctor, -1 for family ctor)

The idea is that everything we do with the \Q!Family! is a mutation; the \Q!fam.processDay! calls also mutate the contained \Q!Person!s.

This is a worst case scenario for our approach compared to visible state semantics since it reduces our advantages:
our approach avoids invariant checks when objects are not mutated
but in this example most operations are mutations; 
similarly, our approach prevents the exponential explosion of nested invariant checks\footnote{As happened in our GUI case study, see Section \ref{s:case-study}.} when deep object graphs are involved, but in this example the object graph of \Q!fam! is very shallow.
\loseSpace

We ran this test case using several different languages: L42 (using our protocol) performs $4{,}000$ checks, D and Eiffel perform $7{,}995$, and finally, Spec\# performs only $1{,}104$.

Our protocol performs a single invariant check at the end of each constructor,  \Q!processDay! and \Q!addChild! call (for both \Q!Person! and \Q!Family!). 

The visible state semantics of both D and Eiffel perform additional invariant checks at the beginning of each call to \Q!processDay! and \Q!addChild!.

The results for Spec\# are very interesting, since it performs less checks than L42.
This is the case since \Q!processDay! in \Q!Person! just does a simple field update, which in Spec\# do not invoke runtime invariant checks. Instead, Spec\# tries to statically verify that the update cannot break the invariant; if it is unable to verify this, it requires that the update be wrapped in an \Q!expose! block. 

Spec\# relies on the absence of arithmetic overflow, and performs runtime checks to ensure this%
\footnote{%
Runtime checks are enabled by a compilation option; when they fail, unchecked exceptions are thrown.%
}, as such the verifier concludes that the field increment in \Q!processDay! cannot break the invariant.
Spec\# is able to avoid some invariant checks in this case 
by relying on all arithmetic operations performing runtime overflow checks;
whereas integer arithmetic in L42 has the common wrap around semantics.
%L42's integers have common wrap-around semantic.

%\footnote{%
%Such semantic can be enforced by
%a compilation option, disabled on default for performance reasons.
%Overflow is detected at runtime, by throwing an unchecked exceptions; just as for invariant failures.%
%}

% Concluding Spec\# is able to replaces some runtime invariant checks with more efficient runtime overflow checks.

%This is the case since \Q!processDay! in \Q!Person! just does a simple field increment, thus the Spec\# verifier is able to statically verify that this wont break the invariant, and so it does not require a corresponding \Q!expose! block, and hence does not perform a runtime invariant check.
%The Spec\# verifier is able to do this as it works on a language semantic where arithmetic overflow does not occur. Such semantic can be enforced by
%a compilation option (disable on default for performance reasons).
%%%%however one can turn on runtime checking for overflow.
%%%and check overflow errors at run time
%With this option turned on, eliding the invariant check is sound since overflow will have the same result as a runtime invariant check failure, namely it will throw an unchecked exception.

%This static reasoning is performed under the assumption that arithmetic overflow will not occur, thus Spec\# is considering a different semantic for \Q@Int@ then L42. Spec\# can inject run-time checks to enforce such arithmetic semantic.

The annotations we had to add in the Spec\# version\footnote{The Spec\# code is in the artifact.} were similar to our previous examples, however since the fields of \Q!Person! all have immutable classes/types, we only needed to add the invariant itself. The \Q!Family! class was similar to our \Q!Cage! example (see section \ref{s:intro}), however in order to implement the \Q!addChild! method we were forced to do a shallow clone of the new child (this also caused a couple of extra runtime invariant checks). Unlike L42 however, we did not need to create a box to hold the \Q!parents! and \Q!children! fields, instead we wrapped the body of the \Q!Family.processDay! method in an \Q!expose (this)! block. In total we needed 16 annotations, worth a total of 45 tokens, this is worse than the code following our approach that we showed above, which has 14 annotations and 14 tokens.

\subheading{Spec\# Papers}
Their are many published papers about the pack/unpack methodology used by Spec\#. To compare against their expressiveness we will consider the three mains ones that introduced their methodology and extensions:
\begin{itemize}
	\item \emph{Verification of Object-Oriented Programs with Invariants:}~\cite{DBLP:journals/jot/BarnettDFLS04} this paper introduces their methodology. In their examples section (pages 41--47), they show how their methodology would work in a class heirarchy with \Q!Reader! and \Q!ArrayReader! classes. The former represents something that reads characters, whereas the latter is a concrete implementation that reads from an owned array. They extend this further with a \Q!Lexer! that owns a \Q!Reader!, which it uses to read characters and parse them into tokens. They also show an example of a \Q!FileList! class that owns an array of filenames, and a \Q!DirFileList! class that extends it with a stronger invariant. All of these examples can be represented in L42\footnote{Our encodings are in the artifact.}. The most interesting considerations are as follow:
	\begin{itemize}
		\item Their \Q!ArrayReader! class has a \Q!relinquishReader! method that `unpacks' the \Q!ArrayReader! and returns its owned array.
%This allows other code to mutate the array freely,
%Then other code can mutate such relinquished array freely.
The returned array can then be freely mutated and passed around by other code.
%at the cost of being unable to use the \Q!ArrayReader! until it is packed again
However, afterwards the \Q!ArrayReader! will be `invalid', and so one can only call methods on it that do not require its invariant to hold. However, it may later be `packed' again (after its invariant is checked).
%Due to our guarantee that usable objects must have their invariant hold, we are unable to do this.
In contrast, our approach requires the invariant of all usable objects to hold.
%In our approach we do not have any visible unusable objects. 
We can still relinquish the array, but at the cost of making the \Q!ArrayReader! forever unreachable. This can be done by
% However, we can
 declaring \Q!relinquishReader! as a \Q!capsule method!, this works since our type modifier system guarantees that the receiver of such a method is not aliased, and hence cannot be used again. Note that Spec\# itself cannot represent the \Q!relinquishReader! method at all, since it does not provide explicit pack and unpack operations, rather its \Q!expose! statement performs both an unpack and a pack, thus we cannot unpack an \Q!ArrayReader! without repacking it in the same method.
%(unless we were to throw an unchecked-exception, however this is unsound~\cite{Leino2004ExceptionSF}).
		\item Their \Q!DirFileList! example inherits from a \Q!FileList! which has an invariant, and a final method, this is something their approach was specifically designed to handle. As L42 does not have traditional subclassing, we are unable to express this concept fully, but L42 does have code reuse via trait composition, in which case \Q!DirFileList! can essentially copy and paste the methods from \Q!FileList!, and they will automatically enforce the invariant of \Q!DirFileList!. %The disadvantage however is that if \Q!FileList! is a trait, it cannot also be a type, and so we cannot have a non-final type/class with a concrete invariant, however one can always have a final class with an invariant that wraps over a non final one. This is not an inherent limitation with our invariant protocol, but rather one of the L42 language itself.
	\end{itemize}

	\item \emph{Object Invariants in Dynamic Contexts:}~\cite{DBLP:conf/ecoop/LeinoM04} this paper shows how one can specify an invariant for a doubly linked list of \Q!int!s (which is an immutable value type). Unlike our protocol however, it allows the invariant of \Q!Node! to refer to sibling \Q!Node!s which are not owned/encapsulated by itself, but rather the enclosing \Q!List!. Our protocol can verify such a linked list\footnote{%
Our protocol allows for encoding this example, but
to express the invariant we would need to 
use reference equality, which the L42 language does not support.
%We were unable to implement this in L42, as the invariant uses reference equality, which the language does not implement.
} (since its elements are immutable), however we have to specify the invariant inside the \Q!List! class. We do not see this as a problem, as the \Q!Node! type is only supposed to be used as part of a \Q!List!, thus this restriction does not impact users of \Q!List!.
	
	\item \emph{Friends Need a Bit More: Maintaining Invariants Over Shared State:}~\cite{DBLP:conf/mpc/BarnettN04} this paper shows how one can verify invariants over interacting objects, where neither owns/contains the other. They have multiple examples which utilise the `subject/observer' pattern, where a `subject' has some state that an `observer' wants to keep track of. In their \Q!Subject!/\Q!View! example, \Q!View!s are created with references to \Q!Subject!s, and copies of their state. When a \Q!Subject!'s state is modified, it calls a method on its attached \Q!View!s, notifying them of this update. The invariant is that a \Q!View!'s copy of its \Q!Subject!'s state is up to date. Their \Q!Master!/\Q!Clock! example is similar, a \Q!Clock! contains a reference to a \Q!Master!, and saves a copy of the \Q!Master!'s time. The \Q!Master! has a \Q!Tick! method that increases its time, but unlike the \Q!Subject!/\Q!View! example, the \Q!Clock! is not notified. The invariant is that the \Q!Clock!'s time is never ahead of its \Q!Master!'s. Our protocol is unable to verify these interactions, because the
interacting objects are not immutable or encapsulated by each other.
\end{itemize}

\end{document}